\documentclass{article}

\usepackage{arxiv}

\usepackage[utf8]{inputenc} 
\usepackage[T1]{fontenc}    
\usepackage{url}            
\usepackage{booktabs}       
\usepackage{amsfonts}       
\usepackage{nicefrac}       
\usepackage{microtype}      
\usepackage{cancel}

\RequirePackage[OT1]{fontenc}
\RequirePackage{amsthm,amsmath}

\RequirePackage[authoryear]{natbib}

\RequirePackage[colorlinks,citecolor=blue,urlcolor=blue]{hyperref}
\hypersetup{colorlinks=true,linkcolor=blue, linktocpage}

\usepackage[english]{babel}
\usepackage{amsmath}
\usepackage{amssymb}
\usepackage{caption}
\usepackage{array}
\usepackage{fancyhdr}
\usepackage{framed}
\usepackage{multirow}
\usepackage{bbm}
\usepackage[dvipsnames]{xcolor}
\usepackage{mathrsfs}
\usepackage{pdfsync,xargs}
\usepackage{enumitem}
\usepackage[textwidth=80pt]{todonotes}
\usepackage{aliascnt}
\usepackage{cleveref}
\usepackage{lmodern}
\usepackage{makecell}

\usepackage[ruled, noend]{algorithm2e}
\usepackage{appendix}

\newtheorem{thm}{Theorem}

\crefname{theorem}{Theorem}{Theorems}
\Crefname{theorem}{Theorem}{Theorems}
\crefname{corollary}{Corollary}{Corollaries}
\Crefname{corollary}{Corollary}{Corollaries}
\Crefname{chapter}{Chapter}{Chapters}
\crefname{chapter}{Chapter}{Chapters}

\newtheorem{ex}{Example}
\crefname{example}{example}{examples}
\Crefname{example}{Example}{Examples}

\newtheorem{lem}{Lemma}
\crefname{lemma}{Lemma}{lemmas}
\Crefname{Lemma}{Lemma}{Lemmas}

\newtheorem{prop}{Proposition}
\crefname{proposition}{Proposition}{Propositions}
\Crefname{Proposition}{Proposition}{Propositions}

\newtheorem{coro}{Corollary}
\crefname{corollary}{corollary}{corollaries}
\Crefname{Corollary}{Corollary}{Corollaries}

\newtheorem{defi}{Definition}
\crefname{definition}{definition}{definitions}
\Crefname{Definition}{Definition}{Definitions}

\newtheorem{rem}{Remark}
\crefname{remark}{remark}{remarks}
\Crefname{remark}{Remark}{Remarks}

\newenvironment{enumerateList}{ \begin{enumerate}[label=(\roman*),wide=0pt, labelindent=\parindent]}{\end{enumerate}}

\newtheoremstyle{styledef}
  {6pt}
  {6pt}
  {\sffamily}
  {0em}
  {\bfseries}
  {}
  {1.5em}
  {}

\theoremstyle{styledef}
\newenvironment{hyp}[1]{
\begin{enumerate}[label=\textbf{\sf(#1\arabic*)},resume=hyp#1]\begin{sf}}
{\end{sf}\end{enumerate}}
\crefname{hyp}{}{ass}
\Crefname{hyp}{}{Ass}

\usepackage{titlesec}

\titleformat{\section}
  {\normalfont\fontsize{12}{15}\bfseries}{\thesection}{1em}{}

\titleformat{\subsection}
  {\normalfont\fontsize{10}{15}\bfseries}{\thesubsection}{1em}{}

\title{Monotonic Alpha-divergence Minimisation for Variational Inference}

\author{
  Kam\'elia~Daudel \\
1: LTCI, Télécom Paris \\
Institut Polytechnique de Paris, France \\
2: Department of Statistics \\
University of Oxford, United Kingdom \\
\texttt{kamelia.daudel@stats.ox.ac.uk} \\
 \And
Randal~Douc \\
SAMOVAR, Télécom SudParis \\
Institut Polytechnique de Paris, France \\
\texttt{randal.douc@telecom-sudparis.eu} \\
\AND
François Roueff \\
LTCI, Télécom Paris \\
Institut Polytechnique de Paris, France \\
\texttt{francois.roueff@telecom-paris.fr} \\
}

\renewenvironment{leftbar}[2][\hsize]
{
    \def\FrameCommand
    {
        {\color{#2}\vrule width 3pt}
        \hspace{0pt}
    }
    \MakeFramed{\hsize#1\advance\hsize-\width\FrameRestore}
}
{\endMakeFramed}

\newcommand{\kdtxt}[1]{{\sffamily {\color{cyan}{#1}}}}
\newcommand{\frtxt}[1]{{ {\color{red}{#1}}}}

\newcommand{\careful}[1]{\begin{leftbar}{yellow} {\sffamily {\color{red}{\textbf{Check that this comment has been addressed:}}} #1 {\color{red}{\textbf{(end)}}}}\end{leftbar}}

\renewcommand{\frtxt}[1]{{{#1}}}
\renewcommand{\kdtxt}[1]{{{#1}}}

\renewcommand{\careful}[1]{}

\newcommandx{\admiss}[1][1=f]{\mathsf{A}_{#1}}
\newcommandx{\amu}[1][1=\mu_n]{a_{#1}}
\newcommandx{\hamu}[1][1=\mu_n]{\hat a_{#1}}

\newcommand{\argmin}{\mathop{\mathrm{argmin}}}
\newcommand{\argmax}{\mathop{\mathrm{argmax}}}

\newcommand{\as}{\mathrm{a.s.}}

\newcommandx{\aux}[3][1=\mu, 2=\cte, 3=\alpha]{h_{#1, #2}^{(\alpha)}}

\newcommandx{\binfty}[1][1=\alpha]{|b|_{\infty, #1}}
\newcommandx{\balpha}[1][1=\alpha]{|b|_{#1}}
\newcommandx{\blbd}[2][1=(\lbd{}), 2=j]{ {b_{#2, \alpha}{#1}}}
\newcommandx{\bmuf}[2][1=\mu, 2=\alpha]{ {b_{#1, #2}}}
\newcommandx{\bmufk}[2][1=\mu, 2=\alpha]{\hat{b}_{#1, #2, M}}

\newcommandx{\couple}[2][1=\PQ, 2=\PP]{(#1 || #2)}
\newcommandx{\argcond}[2]{\lrb{\left.#1 \,\right|\, #2}}
\newcommand{\Cov}{\mathbb{C}\mathrm{ov}}
\newcommand{\cte}{\kappa}
\newcommandx{\cteaux}[1][1=\alpha]{L_{#1, 4}}
\newcommandx{\cteinf}[1][1=\alpha]{L_{#1, 2}}
\newcommandx{\cteLipSto}[2][1=N, 2=M]{|b|_{\hmu_{#1}, #2, \alpha}}
\newcommandx{\cteLipStoInf}[1][1=\alpha]{B_{#1}}
\newcommandx{\ctemono}[1][1=\alpha]{L_{#1, 1}}
\newcommandx{\ctestar}[1][1=\alpha]{L_{#1, 5}}
\newcommandx{\ctesup}[1][1=\alpha]{L_{#1, 3}}
\newcommand{\data}{\mathscr{D}}

\newcommandx{\diverg}[1][1=\alpha]{D_{#1}}
\newcommandx{\divergR}[1][1=\alpha]{D^{(\mathrm{AR})}_{#1}}
\newcommandx{\Domain}[1][1=\alpha]{\Delta_{#1}}

\newcommand{\eqdef}{:=}

\newcommand{\eqsp}{\;}

\newcommandx{\falpha}[1][1=\alpha]{f_{#1}}
\newcommand{\frho}{f_{\alpha, \varrho}}
\newcommandx{\aei}[1][1=\alpha]{$(#1, \Gamma)$-}
\newcommandx{\GammaAlpha}[1][1=\alpha]{ \Gamma}
\renewcommand{\geq}{\geqslant}
\newcommandx{\gmuf}[1][1=\mu]{ {g_{#1}}}
\newcommand{\hmu}{\hat{\mu}^{M}}

\newcommandx{\iteration}[1][1=\alpha]{\mathcal{I}_{#1}}
\newcommandx{\iterationK}[1][1=\alpha]{\hat{\mathcal{I}}_{#1, M}}

\newcommandx{\lbd}[2][1=]{
    \ifthenelse{\equal{#1}{}}
    {{\boldsymbol{\lambda}_{#2}}}
    {\lambda_{#1,#2}}
    }
\newcommandx{\lbdp}[2][1=]{
    \ifthenelse{\equal{#1}{}}
    {{\boldsymbol{\lambda'}_{#2}}}
    {\lambda_{#1,#2}}
    }

\renewcommand{\leq}{\leqslant}

\newcommand{\lrav}[1]{\left|#1 \right|}
\newcommand{\lr}[1]{\left(#1 \right)}
\newcommand{\lrb}[1]{\left[#1 \right]}
\newcommand{\lrc}[1]{\left\{#1 \right\}}
\newcommand{\lrcb}[1]{\left\{#1 \right\}}

\newcommand{\meas}[1]{\mathrm{M}_{#1}}
\newcommandx{\mixture}[2][1=(y), 2=\lbd{}]{q_{#2, \Theta}#1}

\newcommandx{\norm}[2][1=\infty]{\|#2\|_{#1}}
\newcommand{\nset}{\mathbb N}
\newcommand{\nstar}{\mathbb{N}^\star}
\newcommand{\PE}{\mathbb E}
\newcommandx{\posterior}[1][1=y]{p(#1|\mathscr{D})}
\newcommand{\PP}{\mathbb P}
\newcommand{\PQ}{\mathbb Q}
\newcommandx{\Psif}[1][1=\alpha]{\Psi_{#1}}
\newcommandx{\PsifAR}[1][1=\alpha]{\Psi_{#1}^{(AR)}}

\newcommandx{\ratio}[2][1=\alpha, 2=n]{\varphi_{#2}^{(#1)}}
\newcommandx{\normratio}[2][1=\alpha, 2=n]{\check{\varphi}_{#2}^{(#1)}}
\newcommandx{\normratiot}[2][1=\alpha, 2=n]{\check{\varphi}_{j,#2}^{(#1)}}
\newcommand{\changevar}{\Upsilon}
\newcommand{\changevarcomp}{\upsilon}
\newcommand{\gradA}[1]{#1^{\nabla}}

\newcommand{\intE}{\mathrm{Int}E}
\newcommandx{\respa}[3][1=y, 2=j, 3=\alpha]{\hat\varphi_{#2,n}^{(#3)} (#1)}
\newcommandx{\respat}[3][1=y, 2=\alpha, 3=j]{\varphi_{#3, n}^{(#2)} (#1)}
\newcommandx{\ratiogen}[2][1=\alpha, 2=j]{\varphi_{#2, n}^{(#1)}}
\renewcommand{\rho}{\varrho}
\newcommand{\rmd}{\mathrm d}
\newcommand{\rme}{\mathrm e}

\newcommand{\Rset}{\mathbb{R}}
\newcommand{\rset}{\mathbb{R}}

\newcommand{\set}[2]{\lrc{#1\eqsp: \eqsp #2}}

\newcommand{\simplex}{\mathcal{S}}

\newcommandx{\thetat}[2][1=j, 2=t]{\theta_{#1,#2}}
\newcommandx{\thetav}[2][1=]{
    \ifthenelse{\equal{#1}{}}
    {{\boldsymbol{\theta}_{#2}}}
    {\lambda_{#1,#2}}
    }

\newcommand{\Tset}{{\mathsf{T}}}
\newcommand{\Tsigma}{\mathcal{T}}

\newcommand{\Yset}{\mathsf Y}
\newcommand{\Ysigma}{\mathcal Y}

\newcommand{\Zset}{\mathsf{Z}}
\newcommand{\Zsigma}{\mathcal{Z}}

\newcommandx{\pscal}[3][1=]{\left\langle #2,#3\right\rangle_{#1}}
\newcommandx{\normev}[2][1=]{\left\|#2\right\|_{#1}}
\newcommand{\normop}[1]{\left\|#1\right\|_{\mathrm{op}}}

\newcommand{\parammixproba}{\mathrm{M}}

\def\mt{m}
\def\nut{a}

\begin{document}
\maketitle

\begin{abstract}
In this paper, we introduce a novel family of iterative algorithms which carry out $\alpha$-divergence minimisation in a Variational Inference context. They do so by ensuring a systematic decrease at each step in the $\alpha$-divergence between the variational and the posterior distributions. In its most general form, the variational distribution is a mixture model and our framework allows us to simultaneously optimise the weights and components parameters of this mixture model. Our approach permits us to build on various methods previously proposed for $\alpha$-divergence minimisation such as Gradient or Power Descent schemes and we also shed a new light on an integrated Expectation Maximization algorithm. Lastly, we provide empirical evidence that our methodology yields improved results on several multimodal target distributions and on a real data example.
\end{abstract}


\section{Introduction}
\label{sec:introduction}

Bayesian inference tasks often induce intractable and hard-to-compute posterior densities which need to be approximated. Among the class of approximating methods, Variational inference methods \citep[for example Variational Bayes][]{Jordan1999, beal.phd} have attracted a lot of attention as they have empirically been shown to be widely applicable to many high-dimensional machine-learning problems \citep{JMLR:v14:hoffman13a, 2014arXiv1401.0118R, kingma2014autoencoding}.

These optimisation-based methods introduce a simpler variational family $\mathcal{Q}$ and find the best approximation to the unknown posterior density among this family in terms of a certain measure of dissimilarity, the most common choice of measure of dissimilarity being the exclusive Kullback-Leibler divergence \citep{2016arXiv160100670B, 8588399}.

\careful{[R1] The paper needs to improve the writing. There are some awkward sentences such as `the forward Kullback-Leibler is known...'}

However, the exclusive Kullback-Leibler divergence is known to have some drawbacks. Indeed, its \textit{zero-forcing} behavior is responsible for returning variational approximations with light tails that severely underestimate the posterior covariance and that are unable to capture multimodality within the posterior density \citep{divergence-measures-and-message-passing, jerfel2021variational, prangle2021distilling}. This is especially inconvenient when the variational family $\mathcal{Q}$ does not exactly match the posterior density \citep{pmlr-v80-yao18a, NEURIPS2019_07a4e20a} and even more so if one wants to appeal to Importance Sampling methods to approximate integrals of interest in a Bayesian Inference setting \citep{jerfel2021variational, prangle2021distilling}.

To avoid this hurdle, advances in Variational Inference turned to more general families of divergences such as the $\alpha$-divergence \citep{zhu-rohwer-alpha-div, ZhuRohwer} and R\'{e}nyi's $\alpha$-divergence \citep{renyi1961,2012arXiv1206.2459V}. These families of divergences both recover the exclusive Kullback-Leibler when $\alpha \to 1$ and thanks to the hyperparameter $\alpha$, they provide a more flexible framework that can bypass the difficulties associated to the exclusive Kullback-Leibler divergence by choosing $\alpha <1$ \citep{divergence-measures-and-message-passing}. For this reason, they have notably been used in \cite{power-ep, 2015arXiv151103243H, 2016arXiv160202311L, NIPS2017_6866, NIPS2018_7816, daudel2020infinitedimensional, daudel2021mixture}.

In the spirit of Variational Inference methods based on the $\alpha$-divergence, we propose in this paper to build a framework for $\alpha$-divergence minimisation in a Variational Inference context. The particularity of our work is that our algorithms will ensure a monotonic decrease at each step in the $\alpha$-divergence between the variational and the posterior distributions. In addition, our work will apply to variational families as wide as the class of mixture models. The paper is then organised as follows:
\begin{itemize}[label=$\bullet$,wide=0pt, labelindent=\parindent]
\item In \Cref{sec:optimPb}, we introduce some notation and state the optimisation problem we aim at solving in terms of the posterior density, the variational density $q \in \mathcal{Q}$ and the $\alpha$-divergence.
\item In \Cref{sec:Theta}, we consider the typical Variational
  Inference case where $q$ belongs to a parametric family. In this
  particular case, we state in \Cref{thm:EMtheta} conditions which
  ensure a systematic decrease in the $\alpha$-divergence at each step
  for all $\alpha \in [0,1)$. We then show in \Cref{coro:argmax} that
  these conditions are satisfied for a well-chosen iterative
  scheme and we call the resulting approach the \textit{maximisation approach}. This approach is particularly
  convenient, a fact that we illustrate over \Cref{ex:Gaussian} and
  \Cref{lem:MF}. Furthermore, we derive in \Cref{coro:gradientDescent}
  additional iterative schemes satisfying the conditions of
  \Cref{thm:EMtheta} under the name \textit{gradient-based approach}, which we then use to underline the links between
  our approach and Gradient Descent schemes for
  $\alpha$-divergence and R\'{e}nyi's $\alpha$-divergence
  minimisation.

\item In \Cref{sec:MM}, we extend the results from \Cref{sec:Theta} to the more general case of mixture models. We derive in Theorems \ref{thm:EM:MixtureModel} and \ref{thm:WeightsMixture} conditions to simultaneously optimise the weights and the components parameters of a given mixture model, all the while maintaining the systematic decrease in the $\alpha$-divergence initially enjoyed in \Cref{sec:Theta}. These conditions are then met in \Cref{coro:argminMixtureModel} and \ref{coro:GDMixtureModel}, which respectively generalise the maximisation approach of \Cref{coro:argmax} and the gradient-based approach of \Cref{coro:gradientDescent} to the case of mixture models. 

\item \Cref{sec:related} is devoted to related work. We explain in more detail the links between our framework and existing Variational Inference methods for $\alpha$-divergence minimisation. We also connect our approach to the Power Descent algorithm from \cite{daudel2020infinitedimensional} and provide in \Cref{thm:admiss} additional monotonicity results which go beyond the case $\alpha \in [0,1)$. In addition, we obtain that an integrated Expectation Maximization algorithm introduced in \cite{cappe2008adaptive} can be recovered as a special case of our framework.

\item In \Cref{sec:expo-family}, we state generic results that solve our maximisation and gradient-based approaches when the variational family is based on the exponential family. More specifically, we introduce in \Cref{sec:expo-family-notation} additional notation for the exponential family and we recall several of its useful properties. \Cref{sec:argm-solut-param-expo} and \Cref{sec:gradient-smoothness-expo-family} then focus on the maximisation approach and gradient-based approach respectively and notably provide the theoretical justification behind some special cases mentioned earlier in the paper, such as Gaussian densities in \Cref{ex:Gaussian}. 
Lastly, \Cref{sec:partial-mixture-expo-family} extends once again the maximisation approach in order to incorporate variational families that were not previously included in our framework, yet draw on the exponential family (such as multivariate Student's $t$ densities in \Cref{ex:student-mixture}). 

     
\item In \Cref{sec:practicalities}, we focus on the important case of Gaussian Mixture Models (GMMs) and we discuss practical implementations of our algorithms for GMMs optimisation.

\item Finally, we show in \Cref{subsec:Exp} that our enhanced framework has empirical benefits that permit us to improve on the existing algorithms presented in \Cref{sec:related} when we consider a variety of multimodal targets as well as a real data example.
\end{itemize}

\section{Notation and Optimisation Problem}
\label{sec:optimPb}

Let $(\Yset,\Ysigma, \nu)$ be a measured space, where $\nu$ is a $\sigma$-finite measure on $(\Yset, \Ysigma)$. Assume that we have access to some observed variables $\data$ generated from a probabilistic model $p(\cdot|y)$ parameterised by a hidden random variable $y \in \Yset$ that is drawn from a certain prior $p_0$. The posterior density of the latent variable $y$ given the data $\data$ is then given by
$$
p(y | \data) = \frac{p(y, \data)}{p(\data)} = \frac{ p_0(y) p(\data|y)}{p(\data)} \eqsp,
$$
where $p(\data) = \int_\Yset p_0(y) p(\data|y) \nu(\rmd y)$ is called the {\em marginal likelihood} or {\em model evidence}. For many interesting choices of models, sampling directly from the posterior density is impossible and the marginal likelihood is also unknown or too costly to compute.

To address this difficulty, Variational Inference methods approximate the posterior density by a simpler probability density belonging to a given variational family $\mathcal{Q}$ (from which it is easy to sample from). They select the best probability density in $\mathcal{Q}$ by solving an optimisation problem that involves a certain measure of dissimilarity, which is chosen to be the $\alpha$-divergence in this paper due to its advantages compared to the more traditional exclusive Kullback-Leibler divergence \citep[see for example][]{divergence-measures-and-message-passing}. 

More precisely, let us denote by $\PP$ the probability measure on $(\Yset, \Ysigma)$ with corresponding density $p(\cdot|\data)$ with respect to $\nu$. As for the variational family, let us denote by $\PQ$ the probability measure on $(\Yset, \Ysigma)$ with associated density $q \in \mathcal{Q}$ with respect to $\nu$. We let $\falpha$ be the convex function on $(0, +\infty)$ defined by $\falpha[0](u) = -\log(u)$, $\falpha[1](u) = u \log u$ and $\falpha(u) = \frac{1}{\alpha(\alpha-1)} \left[  u^\alpha -1 \right]$ for all $\alpha \in \rset \setminus \lrcb{0, 1}$. Then, assuming that both $\PQ$ and $\PP$ are absolutely continuous w.r.t. $\nu$, the $\alpha$-divergence between $\PQ$ and $\PP$ \citep[extended by continuity to the cases $\alpha = 0$ and $\alpha = 1$ as for example done in][]{alpha-beta-gamma} is given by
\begin{align}\label{eq:gen:divQ}
\diverg\couple[\PQ][\PP] = \int_\Yset \falpha\left(\frac{q(y)}{p(y|\data)} \right) p(y|\data) \nu(\rmd y) \eqsp,
\end{align}
and the Variational Inference optimisation problem we aim at solving is
$\inf_{q \in \mathcal{Q}} \diverg \couple[\PQ][\PP]$. Notably, it can easily be proven that this optimisation problem is equivalent to solving
\begin{align}
\label{eq:objective}
\inf_{q \in \mathcal{Q}} \Psif(q; p)\eqsp \quad \mbox{with} \quad p(y) = p(y, \data) \quad \mbox{for all } y \in \Yset \eqsp,
\end{align}
where, for all measurable non-negative function $p$ on $(\Yset,\Ysigma)$ and for all $q \in \mathcal{Q}$, we have set
\begin{align*}
\Psif(q; p) = \int_\Yset \falpha \left(\frac{q(y)}{p(y)} \right) p(y) \nu (\rmd y) \eqsp.
\end{align*}
As the unknown constant $p(\data)$ does not appear anymore in the optimisation problem \eqref{eq:objective}, this formulation is often preferred. Therefore, we consider the general optimisation problem
\begin{align}\label{eq:GeneralProblem}
\inf_{q \in \mathcal{Q}} \Psif(q; p)\eqsp,
\end{align}
where $p$ is any measurable positive function on $(\Yset,\Ysigma)$. Note that we may drop the dependency on $p$ in $\Psif$ for notational ease and when no ambiguity occurs. 

\careful{[R1] In Section 3 there is a change of notation. $\Psi_\alpha$ only depends on the kernel, but not on two distributions, as in Section 2.
}

At this stage, we are left with the choice of the variational family $\mathcal{Q}$ appearing in the optimisation problem \eqref{eq:GeneralProblem}. The natural idea in Variational Inference and the starting point of our approach is then to work within a parametric family : letting $(\Tset,\Tsigma)$ be a measurable space, $K:(\theta,A) \mapsto \int_A k(\theta,y)\nu(\rmd y)$ be a Markov transition kernel on $\Tset \times \Ysigma$ with kernel density $k$ defined on $\Tset \times \Yset$, we consider a parametric family of the form
$$
\mathcal{Q} = \set{q: y \mapsto k(\theta,y)}{\theta \in \Tset} \eqsp.
$$

\section{An Iterative Algorithm for Optimising $\Psif(k(\theta, \cdot); p)$}
\label{sec:Theta}

In this section, our goal is to define iterative procedures which optimise \kdtxt{$\Psif(k(\theta, \cdot); p)$} with respect to $\theta$ and which are such that they ensure a \textit{systematic decrease} in $\Psif$ at each step. For this purpose, we start by introducing some mild conditions on $k$, $p$ and $\nu$ that will be used throughout the paper.
\begin{hyp}{A}
  \item \label{hyp:positive} The kernel density $k$ on $\Tset\times\Yset$, the function $p$ on $\Yset$ and
  the $\sigma$-finite measure $\nu$ on $(\Yset,\Ysigma)$ satisfy, for all
  $(\theta,y) \in \Tset \times \Yset$, $k(\theta,y)>0$, $p(y)\geq 0$ and $0< \int_\Yset p(y) \nu(\rmd y)<\infty$.
\end{hyp}
From this point onwards, and as announced in the previous section, we will drop the dependency on $p$ in $\Psif$ for notational convenience. Let us now construct a sequence $(\theta_n)_{n \geq 1}$ valued in $\Tset$ such that the sequence $(\Psif(k(\theta_n, \cdot)))_{n \geq 1}$ is decreasing. The core idea of our approach will rely on the following proposition.
\begin{prop}\label{prop:inegThetaMain} Assume \ref{hyp:positive}. For all $\alpha \in [0,1)$ and all $\theta, \theta' \in \Tset$ such that \kdtxt{$\Psif(k (\theta', \cdot) )< \infty$}, it holds that
\begin{align}\label{eq:ineqThetaMain}
\Psif(k(\theta, \cdot))& \leq \int_\Yset \frac{k(\theta', y)^\alpha p(y)^{1-\alpha} }{\alpha-1} \log \lr{\frac{k(\theta, y)}{k(\theta', y)}} \nu(\rmd y) + \Psif(k (\theta', \cdot) ) \eqsp.
\end{align}
Moreover, if $\alpha=0$, this inequality is an equality and, if $\alpha\in(0,1)$, equality in~(\ref{eq:ineqThetaMain}) is equivalent to having $k(\theta, y)=k(\theta', y)$ for $\nu$-a.e. $y\in\{p>0\}$.
\end{prop}
\begin{proof} We treat the two cases $\alpha = 0$ and $\alpha \in (0,1)$ separately.

\begin{enumerate}[label=(\alph*),wide=0pt, labelindent=\parindent]
  \item Case $\alpha = 0$ with $\falpha[0](u) = - \log(u)$ for all $u > 0$. This case is immediate since
\begin{align*}
\Psif[0](k(\theta, \cdot)) = - \int_\Yset p(y) \log \left( \frac{k(\theta, y)}{ k(\theta', y)} \right) \nu(\rmd y)  +  \Psif[0](k(\theta', \cdot)) \eqsp.
\end{align*}

\item Case $\alpha \in (0,1)$ with $\falpha(u) = \frac{1}{\alpha(\alpha-1)} \lrb{u^\alpha -1}$ for all $u > 0$. We have that
\begin{align*}
\Psif(k(\theta, \cdot)) &= \int_\Yset \frac{\lrb{\lr{ \frac{k(\theta, y)}{p(y)}}^\alpha - 1}}{\alpha(\alpha-1)}  p(y) \nu(\rmd y) \\
&= \int_\Yset \lr{\frac{k(\theta', y)}{p(y)}}^{\alpha} \frac{\lrb{\lr{ \frac{k(\theta, y)}{k(\theta', y)}}^\alpha - 1}}{\alpha(\alpha-1)}  p(y) \nu(\rmd y) + \Psif(k(\theta', \cdot))
\end{align*}
Furthermore, using that $\log(u^\alpha) \leq u^\alpha - 1$ for
  all $u > 0$ with equality only if $u=1$, and since $\alpha \in
(0,1)$, this inequality becomes
$$
\frac{u^\alpha-1}{\alpha(\alpha -1)} \leq \frac{\log(u^\alpha)}{\alpha(\alpha -1)} =\frac{\log(u)}{\alpha -1}     \eqsp.
$$
Applying this with $u=k(\theta, y)/k(\theta', y)$ in the previous
expression, we get that
\begin{align*}
\Psif(k(\theta, \cdot))& \leq \int_\Yset \frac{k(\theta', y)^\alpha p(y)^{1-\alpha} }{\alpha-1} \log \lr{\frac{k(\theta, y)}{k(\theta', y)}} \nu(\rmd y) + \Psif(k (\theta', \cdot) )
\end{align*}
which is exactly \eqref{eq:ineqThetaMain}. Moreover, the
  equality holds if $k(\theta, y)/k(\theta', y)=1$ for $\nu$-a.e. $y$
  such that $k(\theta', y)^\alpha p(y)^{1-\alpha}>0$, which,
  by~\ref{hyp:positive}, only happens if $p(y)>0$.
\end{enumerate}
\end{proof}
This result then allows us to deduce \Cref{thm:EMtheta} below.
\begin{thm}\label{thm:EMtheta} Assume \ref{hyp:positive}. Let $\alpha \in [0,1)$ and starting from an initial $\theta_1 \in \Tset$, let $(\theta_{n})_{n \geq 1}$ be defined iteratively such that for all $n \geq 1$,
\begin{align}\label{eq:ineqThetaMainTwo}
\int_\Yset \frac{k(\theta_n, y)^\alpha p(y)^{1-\alpha} }{\alpha-1} \log \lr{\frac{k(\theta_{n+1}, y)}{k(\theta_{n}, y)}} \nu(\rmd y) \leq 0 \eqsp.
\end{align}
Further assume that $\Psif(k(\theta_1, \cdot)) < \infty$. Then,
$\Psif(k(\theta_{n+1}, \cdot)) \leq \Psif(k(\theta_{n}, \cdot))$ for
all $n \geq 1$ with an equality occurring only
  if~(\ref{eq:ineqThetaMainTwo}) is also an equality.
\end{thm}

\begin{proof}  Applying
    \Cref{prop:inegThetaMain} iteratively for $n=1,2,\dots$ with $\theta = \theta_{n+1}$ and
  $\theta' = \theta_n$ in \eqref{eq:ineqThetaMain} combined with
  \eqref{eq:ineqThetaMainTwo},   we get that $\Psif(k(\theta_{n},
  \cdot))<\infty$ and
  $$
  \Psif(k(\theta_{n+1}, \cdot)) \leq \int_\Yset \frac{k(\theta_n,
    y)^\alpha p(y)^{1-\alpha} }{\alpha-1} \log
  \lr{\frac{k(\theta_{n+1}, y)}{k(\theta_{n}, y)}} \nu(\rmd
  y)+\Psif(k(\theta_{n}, \cdot)) \leq 
  \Psif(k(\theta_{n}, \cdot))\;.
  $$
Suppose now that
    $\Psif(k(\theta_{n+1}, \cdot)) = \Psif(k(\theta_{n}, \cdot))$ for
    some $n\geq1$. Then we have equalities in the previous display and
    thus an equality in~(\ref{eq:ineqThetaMainTwo}) and
    in~(\ref{eq:ineqThetaMain}) with $\theta = \theta_{n+1}$ and
    $\theta' = \theta_n$.
\end{proof}

At this point, we seek to find iterative schemes satisfying \eqref{eq:ineqThetaMainTwo}. \kdtxt{To do so, for all $\alpha \in [0,1)$, all $n \geq 1$ and all $y \in \Yset$, we introduce the notation:
\begin{align}\label{eq:ratio}
&\ratio(y) = {k(\theta_n, y)^{\alpha} p(y)^{1-\alpha}}  \\
&\normratio=\frac{\ratio}{\int\ratio\;\rmd\nu} \eqsp \nonumber.
\end{align}
Note that under \ref{hyp:positive}, the normalisation in $\normratio$ satisfies $0 < \int\ratio\;\rmd\nu < \infty$, where the right inequality in particular follows from Jensen's inequality applied to the strictly concave function $u \mapsto u^{1-\alpha}$. We then state our first corollary.}

\begin{coro}[\kdtxt{Maximisation approach}]\label{coro:argmax}  Assume \ref{hyp:positive}. Let $\alpha \in [0,1)$ \kdtxt{ and let $(b_n)_{n \geq 1}$ be a non-negative sequence. S}tarting from an initial $\theta_1 \in \Tset$, let $(\theta_{n})_{n \geq 1}$ be defined iteratively as follows
\kdtxt{
\begin{align}
\theta_{n+1} &= \argmax_{\theta \in \Tset} \int_\Yset \lrb{ \ratio(y) + b_n k(\theta_n , y)} \log \lr{\frac{k(\theta, y)}{k(\theta_{n}, y)}} \nu(\rmd y) \eqsp, \quad n \geq 1 \eqsp \kdtxt{,} \label{eq:updateTheta}
\end{align}
where we assume that this argmax is uniquely defined at each
step. Then \eqref{eq:ineqThetaMainTwo} holds at all times $n \geq 1$
and we can apply \Cref{thm:EMtheta}.}
\end{coro}

\kdtxt{
\begin{proof} Under \ref{hyp:positive} and using that $\alpha \in [0,1)$, we can rewrite \eqref{eq:ineqThetaMainTwo} as
\begin{align*} 
\int_\Yset   \ratio(y) \log \lr{\frac{k(\theta_{n+1}, y)}{k(\theta_{n}, y)}} \nu(\rmd y) \geq 0 \eqsp.
\end{align*}
Since $\int_\Yset k(\theta_n, y) \log \lr{\frac{k(\theta_{n+1}, y)}{k(\theta_{n}, y)}} \nu(\rmd y) = - \diverg[1] \couple[K(\theta_n, \cdot)][K(\theta_{n+1}, \cdot)] \leq 0$, we also obtain that 
\begin{align*}
  \int_\Yset   \ratio(y) \log \lr{\frac{k(\theta_{n+1}, y)}{k(\theta_{n}, y)}} \nu(\rmd y) \geq \int_\Yset \lrb{   \ratio(y) + b_n  k(\theta_n, y)} \log \lr{\frac{k(\theta_{n+1}, y)}{k(\theta_{n}, y)}} \nu(\rmd y) \eqsp.
\end{align*}
\kdtxt{
Thus, \eqref{eq:ineqThetaMainTwo} holds at all times $n$ if the following condition is satisfied:
\begin{align*}
  \int_\Yset \lrb{   \ratio(y) + b_n k(\theta_n , y)} \log\left(\frac{k(\theta_{n+1}, y)}{k(\theta_{n}, y)}\right) \nu(\rmd y)\geq0 \eqsp, \quad n \geq 1 \eqsp,
\end{align*}  
which is implied by the definition of $\theta_{n+1}$ in \eqref{eq:updateTheta}.}
  \end{proof}
}
We now make three comments regarding \Cref{coro:argmax}.
\begin{itemize}
\item The r.h.s. of \eqref{eq:updateTheta} can be simplified, yielding the equivalent iterative scheme
\begin{align*}
\theta_{n+1} &= \argmax_{\theta \in \Tset} \int_\Yset \lrb{ \ratio(y)
               + b_n k(\theta_n , y)} \log k(\theta, y) \nu(\rmd y) \eqsp, \quad n \geq 1 \eqsp , 
\end{align*}
provided that  $\int_\Yset \lrb{ \ratio(y)
+ b_n k(\theta_n , y)} |\log k(\theta_n, y)| \nu(\rmd
y)<\infty$ for all $n \geq 1$. 
\item  While it suffices to find any $\theta_{n+1}$ so that \eqref{eq:ineqThetaMainTwo} holds to obtain a systematic decrease in $\Psif$, defining $\theta_{n+1}$ as in~(\ref{eq:updateTheta}) enables us to solve this argmax problem for notable choices of kernel density $k$. A remarkable aspect is indeed that \eqref{eq:updateTheta} is written as a maximisation problem involving the logarithm of the ratio $k(\theta, y) / k(\theta_n, y)$ whereas $\Psi_\alpha$ is not directly expressed in terms of logarithm of this ratio for $\alpha\notin\{0,1\}$. As a result, we can use the \textit{maximisation approach} of \Cref{coro:argmax} to derive simple update rules for $(\theta_n)_{n \geq 1}$. An example and a lemma are provided thereafter to illustrate this fact and more general results regarding exponential families will follow later in \Cref{sec:expo-family}.

\item The sequence $(b_n)_{n\geq1}$ appearing in \eqref{eq:updateTheta} is arbitrary and can as a result be chosen by the practitioner. It is responsible for introducing a regularisation term in the argmax problem \eqref{eq:updateTheta} so that the larger $b_n$ is, the closer $\theta_{n+1}$ is to $\theta_n$.
\end{itemize}
We next provide a motivating example where the maximisation approach is applicable.

\begin{ex}[Maximisation approach for a Gaussian density] \label{ex:Gaussian} \hspace{-0.3cm} We consider a $d$-dimen-sional Gaussian density, in which case $\Yset=\rset^d$, $\nu$ is the
    $d$-dimensional Lebesgue measure and $k(\theta, y) = \mathcal{N}(y; m, \Sigma)$, where
  $\theta = (m, \Sigma) \in \Tset$ denotes the mean and covariance
  matrix of the Gaussian density. We let the parameter set $\Tset$ include all possible means $m$ in $\rset^d$ and all possible positive-definite covariance matrices $\Sigma$. Finally, denoting by $\normev{\cdot}$ the Euclidean norm, we assume that the non-negative
  function $p$ defined on $\Yset$ satisfies
  \begin{align}\label{eq:gaussian-p-cond}
       0 <  \int_\Yset \left(1+\|y \|^{2/(1-\alpha)}\right)\,p(y)\;\rmd
         y<\infty\;.
  \end{align}
  Then~\ref{hyp:positive} holds. Furthermore, starting from any
  $\theta_1 = (m_1, \Sigma_1) \in \Tset$ so that $\Psif(k(\theta_1, \cdot)) < \infty$ and denoting
    $\theta_n = (m_n, \Sigma_n)$ for all $n \geq 1$, it holds that: for all $n \geq 1$ and all $\gamma_n\in(0,1]$, there exists $b_n\geq0$ such that the argmax problem \eqref{eq:updateTheta} has a unique solution defined by
  \begin{align}\label{eq:update-gaussian}
         \begin{split}
       & m_{n+1}= \gamma_n\,\hat{m}_n +
         (1-\gamma_n)\,m_n \;,     \\
&\Sigma_{n+1} =  \gamma_{n}\, \widehat{\Sigma}_{n} + (1 -\gamma_{n}) \,\Sigma_{n} + \gamma_{n} (1 -\gamma_{n})\,
\lr{\widehat{m}_{n}-m_{n}} \lr{\widehat{m}_{n}-m_{n}}^T  \eqsp, 
       \end{split}
  \end{align}
  where, using the definition of $\normratio$ in \eqref{eq:ratio}, we set
     \begin{align*}
    &\widehat{m}_n =\int_\Yset y\, \normratio(y)\;\nu(\rmd y)\;,\\
    &\widehat{\Sigma}_{n}= \int_\Yset yy^T\, \normratio(y)\;\nu(\rmd y) -\widehat{m}_n\widehat{m}_n^T\;.
     \end{align*}   
    Here, the detailed derivation of \eqref{eq:update-gaussian} is deferred to \Cref{sec:argm-solut-param-expo} as we focus on interpreting this result in light of \Cref{coro:argmax} (in particular we will see in \Cref{coro:GaussianTemp} that the case $\gamma_n=1$ corresponding to $b_n=0$ will require an additional non-degenerate condition on $p$). 
    By \Cref{coro:argmax}, the systematic decrease in $(\Psif(k(\theta_n,\cdot)))_{n\geq1}$ holds for any choice of sequence $(\gamma_n)_{n\geq1}$ valued in $(0,1]$ in the updates \eqref{eq:update-gaussian}. In addition, $\gamma_n$ permits us to build a tradeoff at time $n \geq 1$ between selecting an update close to the current parameter (that is, taking $\gamma_n$ close to zero) and choosing the Gaussian density with exactly the same mean and covariance matrix as $\normratio$ (that is, setting $\gamma_n=1$). This is a key idea that is linked to the regularisation term appearing in \eqref{eq:updateTheta} and that we will revisit several times throughout the paper.
\end{ex}
The maximisation approach of \Cref{coro:argmax} is also applicable to the commonly-used \emph{mean-field} variational family, which approximates the true density by a density with independent components parameterised separately. The following lemma indeed states that the global argmax problem can be separated into component-wise argmax problems in that case.
  
\begin{lem}[Maximisation approach for the mean-field family]\label{lem:MF} A generic member of the mean-field variational family is given by $k(\theta,y) = \prod_{\ell =1}^{L} k^{(\ell)}(\theta^{(\ell)}, y^{(\ell)})$ with $\theta = (\theta^{(1)}, \ldots, \theta^{(L)}) \in \Tset$ and $y=(y^{(1)}, \ldots, y^{(L)}) \in \Yset$. Then, starting from $\theta_1 \in \Tset$ and denoting $\theta_n = (\theta_n^{(1)}, \ldots, \theta_n^{(L)})$ for all $n \geq 1$, solving \eqref{eq:updateTheta} yields the following update formulas:
    \begin{align*}
    \theta^{(\ell)}_{n+1}&= \argmax_{\theta^{(\ell)}} \int_\Yset \lrb{  \ratio(y) + b_n k(\theta_n, y)} \log \lr{\frac{k^{(\ell)}(\theta^{(\ell)}, y^{(\ell)})}{k^{(\ell)}(\theta_n^{(\ell)}, y^{(\ell)})}   } \nu(\rmd y) \eqsp, \quad \kdtxt{\ell = 1 \ldots L}, \quad n \geq 1 \eqsp.
    \end{align*}
  \end{lem}
The maximisation approach is not the only way to satisfy \eqref{eq:ineqThetaMainTwo}. Indeed, this can also be achieved by taking a \textit{gradient-based} approach and relying on additional smoothness conditions (see \Cref{appendix:smooth} for the definition of $\beta$-smoothness), as written in \Cref{coro:gradientDescent}. 

\begin{coro}[Gradient-based approach] \label{coro:gradientDescent}
  Assume \ref{hyp:positive}. \kdtxt{Let $\Tset \subseteq \rset^d$ be a
    convex set, l}et $\alpha \in [0,1)$ and let $(\gamma_n)_{n \geq 1}$
  be valued in $(0,1]$. Starting from an initial $\theta_1 \in \Tset$, let $(\theta_{n})_{n \geq 1}$ be defined iteratively as follows
\begin{align}
\theta_{n+1} &= \theta_n - \frac{\gamma_n}{\beta_n} \nabla g_n(\theta_n)\eqsp, \quad n \geq 1 \eqsp, \label{eq:updateThetaGD}
\end{align}
where $(g_n)_{n \geq 1}$ is the sequence of functions defined by: for all $n \geq 1$ and all $\theta \in \Tset$,
\begin{align}\label{eq:coroDefGn}
g_n(\theta) = \int_\Yset \frac{\kdtxt{\ratio(y)} }{\alpha-1} \log \lr{\frac{k(\theta, y)}{k(\theta_{n}, y)}} \nu(\rmd y) \eqsp
\end{align}
and $g_n$ is assumed to be $\beta_n$-smooth. Then \eqref{eq:ineqThetaMainTwo} holds for all $n \geq 1$ and we can apply \Cref{thm:EMtheta}.
\end{coro}

\begin{proof}
  Since $\gamma_n \in (0,1]$ and $g_n$ is a $\beta_n$-smooth function
  by assumption, we can apply \Cref{lemma:inegThetaGD} and we obtain
  that for all $n \geq 1$,
\begin{align*}
g_n(\theta_n) - g_n \lr{\theta_n- \frac{\gamma_n}{\beta_n} \nabla g_n(\theta_n)} \geq \frac{\gamma_n}{2\beta_n} \normev{ \nabla g_n(\theta_n)}^2 \eqsp.
\end{align*}
Thus, by definition of $\theta_{n+1}$ in \eqref{eq:updateThetaGD}, we have
$0 = g_n(\theta_n) \geq g_n(\theta_{n+1})$, which in turn implies \eqref{eq:ineqThetaMainTwo} and the proof is concluded.
\end{proof}
Let us now reflect on the implications of \Cref{coro:gradientDescent}. Under common differentiability assumptions, we can write: for all $n \geq 1$ and all $\kdtxt{\theta '} \in \Tset$
\begin{align*}
\nabla g_n(\theta') = \int_\Yset \frac{\kdtxt{\ratio(y)}}{\alpha-1} \left.\frac{\partial\log k(\theta, y)}{\partial\theta}\right|_{(\theta,y) = (\theta',y)}  \nu(\rmd y) \eqsp.
\end{align*}
Then, \eqref{eq:updateThetaGD} becomes
\begin{align}\label{eq:updateThetaGDTwo}
\theta_{n+1} &= \theta_n - \frac{\gamma_n}{\beta_n} \int_\Yset
               \frac{\kdtxt{\ratio(y)}}{\alpha-1}
               \left.\frac{\partial\log k(\theta, y)}{\partial\theta}\right|_{(\theta,y) = (\theta_n,y)} \nu(\rmd y) \eqsp, \quad n \geq 1 \eqsp.
\end{align}
Under this form, the iterative scheme \eqref{eq:updateThetaGDTwo} bears similarities with Gradient Descent iterations for $\alpha$-divergence and R\'{e}nyi's $\alpha$-divergence minimisation. Indeed, given a learning rate policy $(r_n)_{n \geq 1}$ and setting $p = p(\cdot, \data)$, such Gradient Descent iterations are given respectively by
\begin{align*}
& \theta_{n+1} = \theta_n - r_n \int_\Yset \frac{\ratio(y)}{\alpha-1} \left.\frac{\partial\log k(\theta, y)}{\partial\theta}\right|_{(\theta,y) = (\theta_n,y)} \nu(\rmd y) \eqsp, \quad n \geq 1 \\
& \theta_{n+1} = \theta_n - r_n \int_\Yset  \frac{\normratio(y)}{\alpha -1} \left.\frac{\partial\log k(\theta, y)}{\partial\theta}\right|_{(\theta,y) = (\theta_n,y)} \nu(\rmd y) \eqsp, \quad n \geq 1
\end{align*}
(we refer to \Cref{subsec:GDsteps} for details regarding how these updates are obtained). Building on this comment, let us now give an example where the conditions on $(g_n)_{n \geq 1}$ from \Cref{coro:gradientDescent} are satisfied and show how Gradient Descent steps (in that case for R\'{e}nyi's $\alpha$-divergence minimisation) can originate from our gradient-based approach.



\begin{ex}[Gradient-based approach for a Gaussian density]\label{ex:ATwostatisfied} \hspace{-0.25cm} We consider the case of a $d$-dimensional Gaussian density with \kdtxt{$k(\theta, y) = \mathcal{N}(y; m, \Sigma)$, where this time $\theta = m$. Further assume that $\Tset \subseteq \rset^d$ is a convex subset and that $\Sigma = \sigma^2  \boldsymbol{I_d}$, where $\sigma^2> 0$ and $\boldsymbol{I_d}$ denotes the $d$-dimensional identity matrix.} Finally, denoting by $\normev{\cdot}$ the Euclidean norm, we assume that the non-negative function $p$ defined on $\Yset$ satisfies
\begin{align} \label{eq:cond-gradient-gauss-ex}
0< \int_\Yset \lr{1 + \| y \|^{1/(1-\alpha)}} p(y) \rmd y < \infty \eqsp.
\end{align}  
Then, \ref{hyp:positive} holds.   
Furthermore, denoting $\theta_n = m_n$ and setting $\beta_n = \sigma^{-2}(1-\alpha)^{-1}\int \ratio \rmd \nu$ for all $n \geq 1$, 
the conditions on $(g_n)_{n \geq 1}$ from \Cref{coro:gradientDescent} are satisfied so that the mean of the Gaussian density $\mathcal{N}(y; m, \Sigma)$ can be optimised as follows: for all $n \geq 1$ and all $\gamma_n \in [0,1)$
    \begin{align}
    \kdtxt{m_{n+1}} 
    &= \gamma_n \hat{m}_n + \lr{1 - \gamma_n} m_n  \eqsp, \label{eq:MeanGaussianUpdateToRecover}
    \end{align}
    where $\hat{m}_n = \int_\Yset y ~ \normratio(y) ~ \nu(\rmd y)$ (we refer to \Cref{coro:GaussianTempGrad} for the detailed derivation). Setting $p = p(\cdot, \data)$, the update \eqref{eq:MeanGaussianUpdateToRecover} can notably be seen as a Gradient Descent step for R\'{e}nyi's $\alpha$-divergence minimisation with a learning rate $r_n = \sigma^{2} (1-\alpha) \gamma_n$, where $\gamma_n \in [0,1)$.
\end{ex}
We now make two important comments.
\begin{itemize}
  \item \Cref{coro:gradientDescent} sheds light on the links between our approach and the more traditional Gradient Descent methodology for optimising objective functions based on the $\alpha$-divergence in Variational Inference \citep{2016arXiv160202311L, NIPS2017_6866}.
  
  Unlike the usual Gradient Descent methodology, \Cref{coro:gradientDescent} requires a smoothness condition on $g_n$. Smoothness conditions are tools that are often used to obtain stronger convergence guarantees for Variational Inference algorithms \citep[see for example][]{pmlr-v80-buchholz18a, 10.1214/19-AOS1855}. Yet, it can be difficult to satisfy these smoothness assumptions in practice, even if some results have been derived for specific variational families when using the exclusive Kullback-Leibler divergence as the objective function \citep{domke2020provable}. 
  
  To the best of our knowledge, no such results have been proven for $\alpha$-divergence minimisation. In our case, we obtain that a smoothness condition on $g_n$ translates into a systematic decrease in $\Psif$. As we shall see in the forthcoming section, being able to establish monotonicity results on $(\theta_n)_{n \geq 1}$ will come in handy as we try to go beyond the framework considered in \Cref{sec:Theta}.

  \item To put things into perspective with the maximisation approach from earlier, observe that the updates on the means in Examples \ref{ex:Gaussian} and \ref{ex:ATwostatisfied} coincide. The difference between the two examples is due to the fact that the former provides an update for the covariance matrix as well, \textit{without sacrificing the monotonicity of the overall algorithm}. It is in fact hard to derive a covariance matrix update in the Gaussian case using the gradient-based approach (see \Cref{sec:gradient-smoothness-expo-family} for details). In that sense, the maximisation approach provides an interesting alternative that can bypass this difficult smoothness assumption on $g_n$. We will further delve into this aspect as we reach \Cref{sec:expo-family}.
\end{itemize}

So far, the $\log$ function appearing in \Cref{thm:EMtheta} considerably eased the derivation of iterative update formulas for well-chosen kernel densities $k$ that can exploit this log structure (Examples \ref{ex:Gaussian}-\ref{ex:ATwostatisfied} and \Cref{lem:MF}). Yet, these choices of kernel densities can be too restrictive to fully capture complex and multimodal posterior densities. Since working with a larger variational family might lead to more accurate approximations of the posterior density, an idea is then to investigate whether the iterative update formulas from \Cref{sec:Theta} can be generalised to the mixture model variational family (for example whether we can extend the updates for a Gaussian kernel from \Cref{ex:Gaussian} to Gaussian Mixture Models). 

As we shall see in the next section, further theoretical developments will be required in order to derive valid iterative schemes that optimise both the mixture weights and the mixture components parameters of a given mixture model.  

\section{Extension to Mixture Models}
\label{sec:MM}

Let us first formally define the class of mixture
models we are going to be working with. Given $J \in \nstar$, we introduce the simplex of $\rset^J$:
$$
\simplex_J = \set{ \lbd{}= (\lambda_1, \ldots, \lambda_J) \in \rset^J}{ \forall j \in \lrcb{1, \ldots , J}, \eqsp \lambda_j \geq 0 \eqsp \mbox{and} \eqsp \sum_{j=1}^J \lambda_j = 1} \eqsp,
$$
we define $\simplex_J^+ = \set{ \lbd{} \in \simplex_J}{ \forall j \in \lrcb{1, \ldots , J}, \eqsp \lambda_j > 0}$  
and we denote $\Theta = (\theta_1, \ldots, \theta_J) \in \Tset^J$. 
We consider the mixture model variational family given by
\begin{align}\label{eq:def:MM}
\mathcal{Q} = \set{q: y \mapsto \mu_{\lbd{}, \Theta}k(y) \eqdef \sum_{j=1}^{J} \lambda_j k(\theta_j,y)}{\lbd{} \in \simplex_J, \Theta \in \Tset^J}
\end{align}
that is, we are interested in solving the optimisation problem
\begin{align*}
\inf_{\lbd{} \in \simplex_J, \Theta \in \Tset^J} \Psif(\mu_{\lbd{}, \Theta} k)\eqsp,
\end{align*}
with $J >1$. Let us next denote $\lbd{n} = (\lambda_{j,n})_{1 \leq j \leq J}$ and $\Theta_n = (\theta_{j,n})_{1 \leq j \leq J}$ for all $n \geq 1$. For convenience, we also introduce the shorthand notation $\mu_n k = \mu_{\lbd{n}, \Theta_n} k$ and
\begin{align}\label{eq:respat}
& \respat[y] = k(\theta_{j,n}, y) \lr{\frac{\mu_n k(y)}{  p(y)}}^{\alpha-1} \\
& \normratiot = \frac{\ratiogen}{\int \ratiogen \rmd \nu } \nonumber
\end{align}
for all $\alpha \in [0,1)$, all $j = 1 \ldots J$, all $n \geq 1$ and all $y \in \Yset$.

Our goal in this section is to derive iterative schemes for the mixture weights and the mixture components parameters $(\lbd{n}, \Theta_n)_{n \geq 1}$ ensuring that the sequence $(\Psif(\mu_n k))_{n \geq 1}$ is decreasing. As \Cref{thm:EMtheta} holds for any choice of parametric family, a first idea is to apply \Cref{thm:EMtheta} to the variational family \eqref{eq:def:MM}, which gives \Cref{coro:MixtureModel} below. 
\begin{coro}\label{coro:MixtureModel}
  Assume \ref{hyp:positive}. Let $J \in \nstar$, let $\alpha \in [0,1)$ and starting from an initial parameter set $(\lbd{1},\Theta_1) \in \simplex_J^+ \times \Tset^J$, let $(\lbd{n},\Theta_{n})_{n \geq 1}$ be defined iteratively such that for all $n \geq 1$,
  \begin{align}\label{eq:ineqMixtureMain}
    \int_\Yset \frac{(\mu_n k(y))^\alpha p(y)^{1-\alpha} }{\alpha-1} \log \lr{\frac{\mu_{n+1} k(y)}{\mu_n k(y)}} \nu(\rmd y) \leq 0 \eqsp.
  \end{align}
  Further assume that $\Psif(\mu_1  k) < \infty$. Then, $\Psif(\mu_{n+1}k) \leq \Psif(\mu_n k)$ for all $n \geq 1$.
  \end{coro}
  Observe that we are in a less favourable situation in \Cref{coro:MixtureModel} with $J > 1$ compared to the cases we previously studied in \Cref{sec:Theta}. Indeed, we now have a ratio of sums inside the $\log$ function in \eqref{eq:ineqMixtureMain}, meaning that the approach from \Cref{thm:EMtheta} to derive simple iterative schemes does not directly transfer to the variational family \eqref{eq:def:MM} for choices of kernel densities $k$ identified in Examples \ref{ex:Gaussian}-\ref{ex:ATwostatisfied} and \Cref{lem:MF}. However, by carefully exploiting the condition \eqref{eq:ineqMixtureMain}, we are able to overcome this difficulty in our second main theorem below.
  
\begin{thm}\label{thm:EM:MixtureModel} Assume \ref{hyp:positive}. Let \kdtxt{$J \in \nstar$}, $\alpha \in [0,1)$ and starting from an initial  parameter set $(\lbd{1},\Theta_1) \in \simplex_J^+ \times \Tset^J$, let $(\lbd{n},\Theta_{n})_{n \geq 1}$ be defined iteratively such that for all $n \geq 1$,
\begin{align}
&\int_\Yset \sum_{j=1}^{J}  \lambda_{j,n} \frac{\respat[y]}{\alpha-1} \log \lr{\frac{\lambda_{j,n+1}}{\lambda_{j,n}}} \nu(\rmd y) \leq 0 \label{eq:posMixtureW} \\
&\int_\Yset \sum_{j=1}^{J}  \lambda_{j,n} \frac{\respat[y]}{\alpha-1} \log \lr{\frac{k(\theta_{j,n+1}, y)}{k(\theta_{j,n}, y)}} \nu(\rmd y) \leq 0 \label{eq:posMixtureP} \eqsp.
\end{align}
\kdtxt{Further assume that $\Psif(\mu_1 k) < \infty$. Then, $\Psif(\mu_{n+1} k) \leq \Psif(\mu_n k)$ for all $n \geq 1$.}
\end{thm}
\kdtxt{
\begin{proof} By \Cref{coro:MixtureModel}, we can conclude if we show that \eqref{eq:posMixtureW} and \eqref{eq:posMixtureP} together imply \eqref{eq:ineqMixtureMain}. 
    To show this, first observe that since the function $ u \mapsto \frac{1}{\alpha -1} \log(u)$ is convex and $\alpha \in [0,1)$, Jensen's inequality implies that: for all $y \in \Yset$ and all $n \geq 1$,
  \begin{align*}
    \frac{1}{\alpha -1}  \log \lr{\frac{\mu_{n+1}k(y)}{\mu_{n}k(y)}} & = \frac{1}{\alpha -1}  \log \lr{\sum_{j=1}^J \frac{\lambda_{j,n} k(\theta_{j,n}, y)}{ \sum_{\ell = 1}^J \lambda_{\ell,n} k(\theta_{\ell,n},y)} \frac{ \lambda_{j,n+1} k(\theta_{j,n+1}, y)}{\lambda_{j,n} k(\theta_{j,n}, y)}} \nonumber \\ & \leq \sum_{j = 1}^J  \frac{\lambda_{j,n} k(\theta_{j,n}, y)}{ \sum_{\ell = 1}^J \lambda_{\ell,n} k(\theta_{\ell,n}, y)} \frac{1}{\alpha -1} \log \lr{\frac{\lambda_{j, n+1} k(\theta_{j, n+1},y)}{ \lambda_{j,n} k(\theta_{j,n}, y)}} 
  \end{align*}
  that is:
  \begin{align*}
    \frac{1}{\alpha -1}  \log \lr{\frac{\mu_{n+1}k(y)}{\mu_{n}k(y)}} \leq  \sum_{j = 1}^J \frac{\lambda_{j,n}}{\alpha -1} \frac{ k(\theta_{j,n}, y)}{\mu_n k(y)} \log \lr{\frac{\lambda_{j, n+1} k(\theta_{j, n+1},y)}{ \lambda_{j,n} k(\theta_{j,n}, y)}} \eqsp.
    \end{align*}
  Multiplying by $(\mu_n k(y))^{\alpha} p(y) ^{1-\alpha}$ on both sides, integrating with respect to $\nu(\rmd y)$ and using the definition of $\respat[y]$ in \eqref{eq:respat}, this in turn implies that: for all $n \geq 1$,
  \begin{multline*}
  \int_\Yset \frac{(\mu_n k(y))^\alpha p(y)^{1-\alpha} }{\alpha-1} \log \lr{\frac{\mu_{n+1} k(y)}{\mu_n k(y)}} \nu(\rmd y) \\ \leq  \int_\Yset \sum_{j = 1}^J \lambda_{j,n} \frac{\respat[y]}{\alpha -1} \log \lr{\frac{\lambda_{j, n+1} k(\theta_{j, n+1},y)}{ \lambda_{j,n} k(\theta_{j,n}, y)}} \nu(\rmd y) \eqsp. 
  \end{multline*}
  As a consequence, the condition
  \begin{align}\label{eq:inter:MM:proof}
   \int_\Yset \sum_{j = 1}^J \lambda_{j,n} \frac{\respat[y]}{\alpha -1} \log \lr{\frac{\lambda_{j, n+1} k(\theta_{j, n+1},y)}{ \lambda_{j,n} k(\theta_{j,n}, y)}} \nu(\rmd y) \leq 0
  \end{align}  
  implies \eqref{eq:ineqMixtureMain}. The condition \eqref{eq:inter:MM:proof} in itself is then straightforwardly implied by \eqref{eq:posMixtureW} and \eqref{eq:posMixtureP} and the proof is concluded.
 \end{proof}
}
Strikingly, \eqref{eq:posMixtureW} does not depend on $\Theta_{n+1}$ nor does \eqref{eq:posMixtureP} depend on $\lbd{n+1}$ in \Cref{thm:EM:MixtureModel}. This means that we can treat these two conditions separately and thus that the weights and components parameters of the mixture can be optimised simultaneously. This result was far from immediate by looking at the initial condition \eqref{eq:ineqMixtureMain} and, as we shall see laterly in \Cref{subsec:Exp}, will lead to reduced computational power in practice. 

In addition, we have recovered in \eqref{eq:posMixtureP} the key property used in \Cref{sec:Theta}: compared to the condition \eqref{eq:ineqMixtureMain} which involved the $\log$ of weighted sums of kernel densities, \eqref{eq:posMixtureP} considers a sum of logs of kernel densities. This suggests that we can extend the updates derived in \Cref{sec:Theta} for well-chosen kernel densities $k$ to the more general case of mixture models.

Observe finally that the dependency in $\lambda_{j,n+1}$ appearing in \eqref{eq:posMixtureW} is simpler than the dependency in $\theta_{j,n+1}$ appearing in \eqref{eq:posMixtureP} and that is expressed through the kernel density $k$. As a result, we will first study \eqref{eq:posMixtureW}, in the hope of deriving iterative update formulas for the mixture weights that do not require a specific choice of kernel density $k$. Interestingly, while the natural idea is to perform direct optimisation of the left-hand side of \eqref{eq:posMixtureW}, we will derive a more general expression for the mixture weights \kdtxt{update}, which shall induce numerical advantages later illustrated in \Cref{subsec:Exp}.

\subsection{Choice of $(\lbd{n})_{n \geq 1}$}
\label{sec:choice-lbdn_n}

In the following theorem, we identify an update formula which satisfies \eqref{eq:posMixtureW}, regardless of the choice of the kernel \kdtxt{density} $k$.

\begin{thm}\label{thm:WeightsMixture} Assume \ref{hyp:positive}. Let $\alpha \in [0,1)$, let $(\eta_n)_{n \geq 1}$ be valued in $(0,1]$, let $(\cte_n)_{n \geq 1}$ be such that $(\alpha -1) \cte_n \geq 0$ at all times $n \geq 1$ and let $(\Theta_n)_{n \geq 1}$ be any sequence valued in $\Tset^J$. Starting from an initial $\lbd{1} \in \simplex_J^+$, let $(\lbd{n})_{n \geq 1}$ be defined iteratively such that for all $n \geq 1$
\begin{align}
\lambda_{j,n+1} &=  \frac{\lambda_{j,n} \lrb{\int_\Yset \respat[y] \nu(\rmd y) + (\alpha-1) \cte_n}^{\eta_n}}{\sum_{\ell=1}^J \lambda_{\ell,n} \lrb{\int_\Yset \respat[y][\alpha][\ell] \nu(\rmd y) + (\alpha-1) \cte_n}^{\eta_n}} \eqsp, \quad j = 1 \ldots J\;. \label{eq:aPMC:updateOne}
\end{align}
Then \eqref{eq:posMixtureW} holds.
\end{thm}

\begin{proof}
We first check that the integrals appearing
    in~(\ref{eq:aPMC:updateOne}) are finite:
    \begin{equation}
      \label{eq:int-finit-update-lambda}
    \int_\Yset \respat[y]\;\nu(\rmd
        y)      < \infty\;,\qquad j=1,\dots,J\,,\;n\geq1\;.
    \end{equation}
Using~(\ref{eq:respat}), that
    $\lambda_{j,n} > 0$, that $\lambda_{j,n} k(\theta_{j,n}, y) \leq
    \mu_n k(y)$ and Jensen's inequality applied to the concave
    function $u \mapsto u^{1-\alpha}$, we have,
    for any $j=1,\dots,J$ and $n\geq1$,
    
    \begin{align*}
    \int_\Yset \respat[y]\;\nu(\rmd
        y) & = \int_\Yset  k(\theta_{j,n}, y) \lr{\frac{\mu_n k(y)}{  p(y)}}^{\alpha-1} \nu(\rmd
        y) \\
        & \leq  \frac{1}{\lambda_{j,n}} \int_\Yset \mu_n k(y) \lr{\frac{  p(y)}{\mu_n k(y)}}^{1- \alpha} \nu(\rmd
        y) \\
        & \leq \frac{1}{\lambda_{j,n}} \lr{\int_\Yset p(y) \nu(\rmd
          y)}^{1-\alpha}\;.
    \end{align*}
    The bound~(\ref{eq:int-finit-update-lambda}) follows by~\ref{hyp:positive}.
  Now, to prove \eqref{eq:posMixtureW}, we treat the cases
  $\eta_n = 1$ and $\eta_n \in (0,1)$ separately.
\begin{enumerate}[label=(\alph*),wide=0pt, labelindent=\parindent]
  \item Case $\eta_n = 1$. Since $(\alpha-1) \cte_n \geq 0$ with $\alpha \in (0,1)$, we have that
      $$
      \cte_n \sum_{j=1}^{J} \lambda_{j,n} \log ( \lambda_j/ \lambda_{j,n} )\geq 0
      $$
      where we have used that $\sum_{j=1}^{J} \lambda_{j,n} \log ( \lambda_j/ \lambda_{j,n} )\leq \sum_{j=1}^{J} \lambda_{j,n} (\lambda_j / \lambda_{j,n} - 1) = 0$. In other words, to obtain \eqref{eq:posMixtureW} in the particular case $\eta_n = 1$, it is enough to show
      \begin{align*}
      \int_\Yset \sum_{j=1}^{J}  \lambda_{j,n} \frac{\respat[y]}{\alpha-1} \log \lr{\frac{\lambda_{j,n+1}}{\lambda_{j,n}}} \nu(\rmd y) + \cte_n \sum_{j=1}^{J} \lambda_{j,n} \log \lr{ \frac{\lambda_{j,n+1}}{\lambda_{j,n}} } \leq 0
      \end{align*}
      that is, since~(\ref{eq:int-finit-update-lambda}) holds,
\begin{align}\label{eq:toShowEta1}
      \sum_{j=1}^{J}  \lambda_{j,n} \lrb{ \int_\Yset  \frac{\respat[y]}{\alpha-1}  \nu(\rmd y) + \cte_n} \log \lr{ \frac{\lambda_{j,n+1}}{\lambda_{j,n}} } \leq 0 \eqsp.
      \end{align}
      Notice then that by definition of $(\lambda_{j,n+1})_{1 \leq j \leq J}$ when $\eta_n = 1$, we can write
      \begin{align*}
      \lbd{n+1} = \argmin_{\lbd{} \in \simplex_J^+} \sum_{j=1}^{J}  \lambda_{j,n} \lrb{ \int_\Yset \frac{\respat[y]}{\alpha-1} \nu(\rmd y) + \cte_n} \log \lr{\frac{\lambda_j}{\lambda_{j,n}}} \eqsp.
      \end{align*}
      [Indeed, setting $\beta_j = \lambda_{j,n} \lrb{ \int_\Yset \respat[y] \nu(\rmd y) + (\alpha-1) \cte_n}$ and $\bar{\beta}_j = \beta_j / \sum_{\ell=1}^J \beta_{\ell}$ for all $j = 1 \ldots J$, we have that $\sum_{j=1}^{J} \bar{\beta}_j \log \lr{\bar{\beta}_j/\lambda_j} \geq 0$
      and that this quantity is minimal when $\lambda_j = \bar{\beta}_j$ for $j = 1 \ldots J$.] This implies \eqref{eq:toShowEta1} and settles the case $\eta_n = 1$.

  \item For the particular case $\eta_n \in (0,1)$, we will use that for all $\epsilon > 0$ and all $u> 0$,
  $$
  \log(u) = \frac{1}{\epsilon} \log(u^\epsilon) \geq \frac{1}{\epsilon} \lr{1 - \frac{1}{u^\epsilon}} \eqsp.
  $$
  Indeed, since $\int_\Yset  \frac{\respat[y]}{\alpha-1}  \nu(\rmd y) + \cte_n \leq 0$ for all $j = 1 \ldots J$, we can then write that for all $\epsilon > 0$,
  \begin{multline}\label{eq:interTwo}
    \sum_{j=1}^{J}  \lambda_{j,n} \lrb{ \int_\Yset  \frac{\respat[y]}{\alpha-1}  \nu(\rmd y) + \cte_n} \log \lr{ \frac{\lambda_{j,n+1}}{\lambda_{j,n}} } \\
    \leq \frac{1}{\epsilon} \sum_{j=1}^{J}  \lambda_{j,n} \lrb{ \int_\Yset  \frac{\respat[y]}{\alpha-1}  \nu(\rmd y) + \cte_n} \lrb{1 - \lr{\frac{\lambda_{j,n}}{\lambda_{j,n+1}} }^\epsilon} \eqsp.
  \end{multline}
  Now notice that by definition of $(\lambda_{j,n+1})_{1 \leq j \leq J}$ we can write
  \begin{align*}
  \lbd{n+1} = \argmin_{\lbd{} \in \simplex_J^+} \frac{1}{\epsilon} \sum_{j=1}^{J}  \lambda_{j,n} \lrb{ \int_\Yset  \frac{\respat[y]}{\alpha-1}  \nu(\rmd y) + \cte_n} \lrb{1 - \lr{\frac{\lambda_{j,n}}{\lambda_j}}^\epsilon}
  \end{align*}
  when $\epsilon$ satisfies $\eta_n = \frac{1}{1+\epsilon}$. [Indeed setting $\beta_j = \lambda_{j,n} \lrb{ \int_\Yset \respat[y] \nu(\rmd y) + (\alpha-1) \cte_n}^{\frac{1}{1+\epsilon}}$ and $\bar{\beta}_j = \beta_j / \sum_{\ell=1}^J \beta_{\ell}$ for all $j = 1 \ldots J$, we have by convexity of the function $u \mapsto u^{1+\epsilon}$ that $\sum_{j=1}^{J} \lr{\bar{\beta}_j /\lambda_j}^{1+\epsilon} \lambda_j \geq (\sum_{j=1}^{J} {\bar{\beta}_j})^{1+\epsilon}$ and that this quantity is minimal when $\lambda_j = \bar{\beta}_j$ for $j = 1 \ldots J$.] We then deduce that taking $\epsilon = \eta_n^{-1} - 1$ (it is always possible since $\eta_n \in (0,1)$ by assumption) yields
  $$
  \frac{1}{\epsilon} \sum_{j=1}^{J}  \lambda_{j,n} \lrb{ \int_\Yset  \frac{\respat[y]}{\alpha-1}  \nu(\rmd y) + \cte_n} \lrb{1 - \lr{\frac{\lambda_{j,n}}{\lambda_{j,n+1}} }^\epsilon} \leq 0
  $$
  which in turn yields \eqref{eq:posMixtureW} [since combined with \eqref{eq:interTwo} it implies \eqref{eq:toShowEta1} which itself implies \eqref{eq:posMixtureW} as seen in the case $\eta_n = 1$]. This settles the case $\eta_n \in (0,1)$.
\end{enumerate}
\end{proof}
Notice that as a byproduct of the proof of \Cref{thm:WeightsMixture}, the mixture weights update given by \eqref{eq:aPMC:updateOne} can be rewritten under the form: for all $n \geq 1$,
$$
\lbd{n+1} = \argmin_{\lbd{}\in \simplex_J^+} h_n(\lbd{})
$$
where, setting $\epsilon = \eta_n^{-1} -1$, we have defined for all $\lbd{} \in \simplex_J^+$,
\begin{align*}
h_n(\lbd{}) &= \begin{cases}
                  \sum_{j=1}^{J}  \lambda_{j,n} \lrb{ \int_\Yset \frac{\respat[y]}{\alpha-1} \nu(\rmd y) + \cte_n} \log \lr{\frac{\lambda_j}{\lambda_{j,n}}}, & \mbox{if } \eta_n = 1 \eqsp, \\
                 \frac{1}{\epsilon} \sum_{j=1}^{J}  \lambda_{j,n} \lrb{ \int_\Yset  \frac{\respat[y]}{\alpha-1}  \nu(\rmd y) + \cte_n} \lrb{1 - \lr{\frac{\lambda_{j,n}}{\lambda_j}}^\epsilon}, & \mbox{if } \eta_n \in (0,1) \eqsp.
               \end{cases}
\end{align*}
More specifically, $h_n(\lbd{})$ acts as an upper bound of the left-hand side of \eqref{eq:posMixtureW} and we recover exactly the left-hand side of \kdtxt{\eqref{eq:posMixtureW}} in the particular case $\eta_n = 1$ and $\cte_n = 0$. \newline

Now that we have established the mixture weights updates \eqref{eq:aPMC:updateOne} in \Cref{thm:WeightsMixture}, we are interested in deriving update formulas for the sequence $(\Theta_n)_{n \geq 1}$ satisfying \eqref{eq:posMixtureP}, which we will then pair up with \eqref{eq:aPMC:updateOne}  in order to apply \Cref{thm:EM:MixtureModel}. From this point onwards, all the proofs of the coming results will be deferred to the appendices to ease the reading.

\subsection{Choice of $(\Theta_n)_{n \geq 1}$}
\label{sec:choice-theta_n}

We investigate two different approaches for choosing $(\Theta_n)_{n \geq 1}$.

\subsubsection{A Maximisation Approach}
\label{sec:maxim-appr-finite-mixture}
\kdtxt{As done in \Cref{coro:argmax}, an idea is to consider an update for $(\Theta_n)_{n \geq 1}$ of the form}
\begin{align*}
\Theta_{n+1} = \argmax_{\Theta \in \Tset^J} g_n(\Theta) \eqsp, \quad n \geq 1 \eqsp,
\end{align*}
where the function $g_n$ is constructed as a lower bound on $\Theta \in \Tset^J$ of the function 
$\Theta \mapsto \int_\Yset \sum_{j=1}^{J} \lambda_{j,n} \respat[y]
\log \lr{{k(\theta_{j}, y)}/{k(\theta_{j,n}, y)}} \nu(\rmd y)$ 
that satisfies $g_n(\Theta_n) = 0$. In doing so, the function $\tilde{g}_n : \Theta \mapsto
g_n(\Theta)/(\alpha-1)$ evaluated at $\Theta_{n+1}$ is an upper bound of the left-hand side of
\eqref{eq:posMixtureP} and $\tilde g_n(\Theta_{n+1}) \leq 0$
implies \eqref{eq:posMixtureP}. 
This leads us to \Cref{coro:argminMixtureModel} below.
\begin{coro}[Generalised maximisation approach]\label{coro:argminMixtureModel}
Assume \ref{hyp:positive}. Let $\alpha \in [0,1)$, let
  $(\eta_n)_{n \geq 1}$ be valued in $(0,1]$ and let $(\cte_n)_{n \geq
    1}$ be such that $(\alpha -1) \cte_n \geq 0$ at all times $n \geq
  1$. Furthermore, let $(b_{j,n})_{n \geq 1}$ be a non-negative
    sequence for all $j = 1 \ldots J$. Starting from an initial
  parameter set $(\lbd{1},\Theta_1) \in \simplex_J^+ \times \Tset^J$,
  let $(\lbd{n},\Theta_{n})_{n \geq 1}$ be defined iteratively for all
  $n \geq 1$ in such a way that \eqref{eq:aPMC:updateOne} holds and 
\begin{align}
  \theta_{j,n+1} &= \argmax_{\theta \in \Tset} \int_\Yset \kdtxt{\lrb{ \respat[y] + b_{j,n} k(\theta_{j,n}, y)}} \log \lr{\frac{k(\theta, y)}{k(\theta_{j,n}, y)}} \nu(\rmd y) \eqsp, \quad j = 1 \ldots J \eqsp\kdtxt{,} \label{eq:aPMC:updateP}
  \end{align}
where we assume that this argmax is uniquely defined at each
step. Then, we can apply \Cref{thm:EM:MixtureModel}.

\end{coro}
The proof of this result is deferred to \Cref{subsec:coro:argminMixtureModel:proof}. Under the assumptions of \Cref{coro:argminMixtureModel}, Algorithm \ref{algo:GGEMalpha} leads to a systematic decrease in $\Psif$ at each step. This result effectively generalises the monotonicity property from \Cref{coro:argmax} to the case of mixture models and we can deduce simple iterative schemes satisfying \eqref{eq:aPMC:updateP} for well-chosen kernel densities $k$. To illustrate this, we provide below the extension of \Cref{ex:Gaussian} to Gaussian Mixture Models (and \Cref{lem:MF} can be extended in the same way, see \Cref{subsec:extension}).

{\SetAlgoNoLine
\SetInd{0.8em}{-1.4em}
\begin{algorithm}[t]
\caption{Maximisation approach algorithm for mixture models}
\label{algo:GGEMalpha}
\textbf{At iteration $n$,}

For all $j = 1 \ldots J$, set
\begin{align*}
\lambda_{j,n+1} &=  \frac{\lambda_{j,n} \lrb{\int_\Yset \respat[y] \nu(\rmd y) + (\alpha-1) \cte_n}^{\eta_n}}{\sum_{\ell=1}^J \lambda_{\ell,n} \lrb{\int_\Yset \respat[y][\alpha][\ell] \nu(\rmd y) + (\alpha-1) \cte_n}^{\eta_n}} \\
\theta_{j,n+1} &= \argmax_{\theta \in \Tset} \int_\Yset \lrb{ \respat[y] + b_{j,n} k(\theta_{j,n}, y) }\log \lr{ \frac{k(\theta, y)}{k(\theta_{j,n}, y)}} \nu(\rmd y) \eqsp.
\end{align*} \
\end{algorithm} 
}

\begin{ex}[Maximisation approach for Gaussian Mixture Models] \label{ex:GMMmax}

  We consider the case of $d$-dimensional Gaussian mixture densities, in which case $\Yset = \rset^d$, $\nu$ is the $d$-dimensional Lebesgue measure and $k(\theta_j, y) = \mathcal{N}(y; m_j, \Sigma_j)$, where $\theta_j = (m_j, \Sigma_j) \in \Tset$ denotes the mean and covariance matrix of the $j$-th Gaussian component density. We let the parameter set $\Tset$ include all possible means in $\rset^d$ and all possible positive-definite covariance matrices. Finally, we assume that the non-negative function $p$ defined on $\Yset$ satisfies \eqref{eq:gaussian-p-cond}. 
  
  Then \ref{hyp:positive} holds. Moreover, starting from any $(\lbd{1}, \Theta_1) \in \simplex_J^+ \times \Tset^J$ so that $\Psif(\mu_1 k) < \infty$ and denoting $\theta_{j,n} = (m_{j,n}, \Sigma_{j,n})$ for all $j = 1 \ldots J$ and all $n \geq 1$, it holds that: for all $j = 1 \ldots J$, all $n \geq 1$ and all $\gamma_{j,n} \in (0,1]$, there exists $b_{j,n} \geq 0$ such that the argmax problem \eqref{eq:aPMC:updateP} has a unique solution defined by:
\begin{align*}
\quad m_{j,n+1}&= \gamma_{j,n}  \widehat{m}_{j,n} + (1 -\gamma_{j,n}) m_{j,n} \\
\Sigma_{j,n+1} &=  \gamma_{j,n} \widehat{\Sigma}_{j,n} + (1 -\gamma_{j,n}) \Sigma_{j,n} + \gamma_{j,n} (1 -\gamma_{j,n})
\lr{\widehat{m}_{j,n}-m_{j,n}} \lr{\widehat{m}_{j,n}-m_{j,n}}^T  \eqsp, 
\end{align*}
where, using the definition of $\normratiot$ in \eqref{eq:respat}, we set
     \begin{align*}
    &\widehat{m}_{j,n} =\int_\Yset y\, \normratiot (y)\;\nu(\rmd y)\;,\\
    &\widehat{\Sigma}_{j,n}= \int_\Yset yy^T\, \normratiot
      (y)\;\nu(\rmd y)
      -\widehat{m}_{j,n}\widehat{m}_{j,n}^T\;.
     \end{align*}
  The detailed derivation of the mean and covariance updates above is deferred to \Cref{sec:argm-solut-param-expo}. The interpretation of these updates is akin to the one already made in \Cref{ex:Gaussian}: $b_{j,n}$ acts as a regularisation parameter which, through $\gamma_{j,n}$, permits a tradeoff between an update close to the current parameter $\theta_{j,n}$ and choosing the Gaussian density with exactly the same mean and covariance matrix as $\normratiot$. As per written in \Cref{coro:argminMixtureModel}, these updates are compatible with the mixture weights updates \eqref{eq:aPMC:updateOne}, resulting in a systematic decrease of $(\Psif(\mu_n k))_{n \geq 1}$.
\end{ex}
We next present another possible update formula for $(\lbd{n}, \Theta_{n})_{n \geq 1}$.

\subsubsection{A Gradient-based Approach}
\label{subsec:GDapproach}

\kdtxt{In the spirit of \Cref{coro:gradientDescent}, we} now resort to \kdtxt{G}radient \kdtxt{D}escent steps to satisfy \eqref{eq:posMixtureW}.

\begin{coro}[\kdtxt{Generalised gradient-based
    approach}] \label{coro:GDMixtureModel} Assume
  \ref{hyp:positive}. \kdtxt{Let $\Tset \subseteq \rset^d$ be a convex
    set, l}et $\alpha \in [0,1)$, let $(\eta_n)_{n \geq 1}$ be valued
  in $(0,1]$ and let $(\cte_n)_{n\geq 1}$ be such that $(\alpha -1)
  \cte_n \geq 0$ at all times $n$. Furthermore, for all $j = 1 \ldots
  J$, let $(\gamma_{j,n})_{n \geq 1}$ be valued in $(0,1]$. Starting
  from an initial parameter set $(\lbd{1},\Theta_1) \in \simplex_J^+
  \times \Tset^J$, let $(\lbd{n},\Theta_{n})_{n \geq 1}$ be defined
  iteratively for all $n \geq 1$ in such a way that \eqref{eq:aPMC:updateOne} holds and
\begin{align}
\theta_{j, n+1} &= \theta_{j,n} - \frac{\gamma_{j,n}}{\beta_{j,n}} \nabla g_{j,n}(\theta_{j,n}) \eqsp, \quad j = 1 \ldots J \eqsp, \label{eq:updateThetaGDMixtureModels}
\end{align}
where for all $j = 1 \ldots J$, $(g_{j,n})_{n \geq 1}$ is defined by: for all $n \geq 1$ and all $\theta \in \Tset$,
\begin{align*} 
g_{j,n}(\theta) = \int_\Yset \frac{\respat[y]}{\alpha-1} \log \lr{\frac{k(\theta, y)}{k(\theta_{j,n}, y)}} \nu(\rmd y) \eqsp.
\end{align*}
and $g_{j,n}$ is assumed to be $\beta_{j,n}$-smooth. Then, we can apply \Cref{thm:EM:MixtureModel}.
\end{coro}

{\SetAlgoNoLine
\SetInd{0.8em}{-1.4em}
\begin{algorithm}[t]
\caption{Gradient-based approach algorithm for mixture models}
\label{algo:GGEMalphaGD}
\textbf{At iteration $n$,}

For all $j = 1 \ldots J$, set
\begin{align*}
\lambda_{j,n+1} &=  \frac{\lambda_{j,n} \lrb{\int_\Yset \respat[y] \nu(\rmd y) + (\alpha-1) \cte_n}^{\eta_n}}{\sum_{\ell=1}^J \lambda_{\ell,n} \lrb{\int_\Yset \respat[y][\alpha][\ell] \nu(\rmd y) + (\alpha-1) \cte_n}^{\eta_n}} \\
\theta_{j,n+1} &= \theta_{j,n} - \frac{\gamma_{j,n}}{\beta_{j,n}} \nabla g_{j,n}(\theta_{j,n}) \eqsp.
\end{align*} \
\end{algorithm}
}
The proof of this result is deferred to \Cref{subsec:coro:GDMixtureModel:proof}. Under the assumptions of \Cref{coro:GDMixtureModel}, Algorithm \ref{algo:GGEMalphaGD} ensures a systematic decrease in $\Psif$ at each step and \Cref{coro:GDMixtureModel} thus extends \Cref{coro:gradientDescent} to mixture models. Much like what we did for \Cref{coro:gradientDescent}, we want to identify how our updates relate to Gradient Descent-based techniques for optimising $\Theta$. 
Under common differentiability assumptions, we have: for all $n \geq 1$, all $j = 1 \ldots J$ and all $\theta ' \in \Tset$,
$$
\nabla g_{j,n}(\theta') = \int_\Yset \frac{\respat[y]}{\alpha-1} \left.\frac{\partial\log k(\theta, y)}{\partial\theta}\right|_{(\theta,y) = (\theta', y)} \nu(\rmd y) \eqsp,
$$
so that \eqref{eq:updateThetaGDMixtureModels} becomes
\begin{align}\label{eq:updateThetaGDMixtureModelsExpl}
\theta_{j, n+1} &= \theta_{j,n} - \frac{\gamma_{j,n}}{\beta_{j,n}} \int_\Yset \frac{\respat[y]}{\alpha-1} \left.\frac{\partial\log k(\theta, y)}{\partial\theta}\right|_{(\theta,y) = (\theta_{j,n}, y)} \nu(\rmd y) \eqsp, \quad j = 1 \ldots J \eqsp.
\end{align}
The link with Gradient Descent-based Variational Inference shall become apparent by (i) writing the update formulas that ensue from Gradient Descent iterations for $\alpha$-divergence and R\'{e}nyi's $\alpha$-divergence minimisation and (ii) understanding how \Cref{ex:ATwostatisfied} generalises to Gaussian Mixture Models. Given $(\lbd{n}, \Theta_n) \in \simplex_J^+ \times \Tset^J$, an index $j$ in $1 \ldots J$ and letting $p = p(\cdot, \data)$, performing \textit{one} Gradient Descent iteration w.r.t. $\theta_{j,n}$ for $\alpha$-divergence and R\'{e}nyi's $\alpha$-divergence minimisation indeed respectively amounts to updating $\theta_{j,n+1}$ as follows
\begin{align*}
\theta_{j, n+1} &= \theta_{j,n} - r_{j,n} \lambda_{j,n} \int_\Yset \frac{\respat[y]}{\alpha-1} \left.\frac{\partial\log k(\theta, y)}{\partial\theta}\right|_{(\theta,y) = (\theta_{j,n}, y)} \nu(\rmd y) \eqsp, \\ 
\theta_{j, n+1} &= \theta_{j,n} - r_{j,n} \frac{\lambda_{j,n} }{ \sum_{j=1}^J \lambda_{j,n} \int_\Yset \respat[y] \nu(\rmd y)  } \int_\Yset  \frac{\respat[y]}{\alpha-1} \left.\frac{\partial\log k(\theta, y)}{\partial\theta}\right|_{(\theta,y) = (\theta_{j,n}, y)} \nu(\rmd y) \eqsp 
\end{align*}
where $r_{j,n} > 0$ is the learning rate (we refer to \Cref{subsec:GDstepsMM} for details regarding these updates). There is thus a similarity between \eqref{eq:updateThetaGDMixtureModelsExpl} and the Gradient Descent updates above. To fully comprehend the connection between our gradient-based approach and Gradient Descent steps, we present below the generalisation of \Cref{ex:ATwostatisfied} to Gaussian Mixture Models. The smoothness assumption on $g_{j,n}$ will be satisfied in that example and we will recover a Gradient Descent scheme for R\'{e}nyi's $\alpha$-divergence minimisation as a special case.

\begin{ex}[Gradient-based approach for Gaussian Mixture Models] \label{ex:GMM} \hspace{-0.2cm}
We consider the case of $d$-dimensional Gaussian mixture densities with $k(\theta_j, y) = \mathcal{N}(y; m_j, \Sigma_j)$, where this time $\theta_j = m_j$. Further assume that $\Tset \subseteq \rset^d$ is a convex subset and that $\Sigma_j = \sigma_j^2  \boldsymbol{I_d}$, where $\sigma_j^2> 0$ and $\boldsymbol{I_d}$ denotes the $d$-dimensional identity matrix. Finally, we assume that the non-negative function $p$ defined on $\Yset$ satisfies \eqref{eq:cond-gradient-gauss-ex}.

Then, \ref{hyp:positive} holds. Setting $\beta_{j,n} = \sigma_j^{-2} (1 -\alpha)^{-1} \int_\Yset \respat[y] \nu(\rmd y)$ and denoting $\theta_{j,n} = m_{j,n}$, the function $g_{j,n}$ is $\beta_{j,n}$-smooth for all $j = 1 \ldots J$ and all $n \geq 1$ so that the means of the Gaussian densities $(\mathcal{N}(y;m_j, \Sigma_j))_{1\leq j \leq J}$ can be optimised as follows: for all $n \geq 1$,
  \begin{align*}
      m_{j, n+1} 
      &= \gamma_{j,n} \hat{m}_{j,n} + \lr{1 - \gamma_{j,n}} m_{j,n} \eqsp, \quad j = 1 \ldots J \eqsp, 
      \end{align*}
  where $\normratiot$ is defined in \eqref{eq:respat}, $\gamma_{j,n} \in [0,1)$ and $\hat{m}_{j,n} = \int_\Yset y ~ \normratiot(y) ~ \nu(\rmd y)$ (we refer to \Cref{coro:GaussianTempGrad} for details). In particular, letting $(\gamma'_{j,n})_{n\geq 1}$ be valued in $(0,1]$ for all $j = 1 \ldots J$,
  $$
  \gamma_{j,n} = \gamma_{j,n}' \frac{\lambda_{j,n} \int_\Yset \respat[y] \nu(\rmd y)}{\int_\Yset {\mu_n k(y)^\alpha p(y)^{1-\alpha} } \nu(\rmd y)} \in (0,1], \quad j= 1\ldots J, \quad n \geq 1, 
  $$
  (since $\int_\Yset \mu_n k(y)^\alpha p(y)^{1-\alpha}  \nu(\rmd y) = \sum_{j=1}^J \lambda_{j,n}  \int_\Yset \respat[y] \nu(\rmd y)$). 
  The resulting iterative algorithm is given by the following update at time $n \geq 1$:
  \begin{align}
    m_{j, n+1} &= m_{j,n} + \gamma_{j,n}' \frac{\int_\Yset \lambda_{j,n} \respat[y] (y-m_{j,n}) \nu(\rmd y)}{\int_\Yset {\mu_n k(y)^\alpha p(y)^{1-\alpha} } \nu(\rmd y)} \eqsp, \quad j = 1 \ldots J \eqsp. \label{eq:GD:GMM:Renyi}
  \end{align}
Interestingly, \eqref{eq:GD:GMM:Renyi} can be seen as a Gradient Descent step w.r.t. $\theta_{j,n}$ for R\'{e}nyi's $\alpha$-divergence minimisation with a learning rate $r_{j,n} = \sigma_j^2(1-\alpha) \gamma_{j,n}'$. Hence, if we were to solely rely on the Gradient Descent literature, the convergence of the iterative sequence $(\theta_{j,n})_{n \geq 1}$ defined by \eqref{eq:GD:GMM:Renyi} would require the sequence $(\lbd{n})_{n \geq 1}$ to be constant. This is in contrast with \Cref{coro:GDMixtureModel}, which allows for a simultaneous optimisation of $\lbd{}$ and $\Theta$ according to \eqref{eq:aPMC:updateOne} and \eqref{eq:GD:GMM:Renyi}.
\end{ex}
We now add on the comments made in \Cref{sec:Theta} for the gradient-based methodology (which we built in \Cref{coro:gradientDescent} and have since extended to mixture models in \Cref{coro:GDMixtureModel}): 
\begin{itemize}
  \item A core insight from \Cref{coro:GDMixtureModel}, which is exemplified in \Cref{ex:GMM}, is that under a smoothness assumption on $g_{j,n}$ our mixture weights iterative updates are compatible with gradient-based updates, themselves linked to the Gradient Descent literature. In other words, we have embedded Gradient Descent-based iterative updates, which only act on $\Theta$, within a larger framework where simultaneous updates for the mixture components parameters and the mixture weights are well-supported theoretically.

  \item Putting things into perspective with the maximisation approach once again, notice that our previous conclusions from \Cref{sec:Theta} still hold. Namely: (i) the updates on the means from Examples \ref{ex:GMMmax} and \ref{ex:GMM} coincide, meaning that contrary to the gradient-based approach, the maximisation approach enables covariance matrices updates on top of means updates (ii) the maximisation approach permits us to bypass the smoothness assumption made in \Cref{coro:GDMixtureModel}. As a whole, these properties make the maximisation approach a compelling alternative to the gradient-based approach.  
\end{itemize}


We presented two approaches to construct iterative schemes
$(\lbd{n}, \Theta_n)_{n \geq 1}$ that ensure a monotonic decrease in $\Psif$ at each step and lead to simple updates formulas for Gaussian Mixture Models. 
We now describe how our framework can be linked to the existing literature.

\section{Related Work}
\label{sec:related}

{
In this section, we detail how our work relates to and improve on previous algorithms proposed for $\alpha$-/R\'{e}nyi's $\alpha$-divergence minimisation.

\subsection{R\'{e}nyi Divergence Variational Inference \citep{2016arXiv160202311L}} 
\label{relwork:LiTurner}

In \cite{2016arXiv160202311L}, they seek to maximise the Variational R\'{e}nyi (VR) Bound via (Stochastic) Gradient Ascent. This objective function is derived from R\'{e}nyi's $\alpha$-divergence and is given by: for all variational density $q \in \mathcal{Q}$ and all $\alpha \in \rset \setminus \lrcb{1}$, 
\begin{align} \label{eq:VRboundDef}
  \mathcal{L}_\alpha(q, \data) \eqdef \frac{1}{1-\alpha} \log \lr{ \int_\Yset q(y)^{\alpha} p(y, \data)^{1-\alpha} \nu(\rmd y)} \eqsp.
\end{align}
Contrary to us, the work from \cite{2016arXiv160202311L} does not consider the case where $q$ belongs to the mixture models variational family \eqref{eq:def:MM}, that is $q = \mu_{\lbd{}, \Theta} k$. 

Yet, a parallel can be drawn in the GMM case between \cite{2016arXiv160202311L} and our approach by observing that the gradient-based updates on the means in \Cref{ex:GMM} each coincide with a Gradient Ascent step on the objective function \eqref{eq:VRboundDef} for a well-chosen learning rate (since the gradient of $\mathcal{L}_\alpha$ is proportional to the gradient of R\'{e}nyi's $\alpha$-divergence with a factor $-\alpha^{-1}$, this follows from the remarks made in \Cref{subsec:GDapproach} regarding Gradient Descent steps for R\'{e}nyi's $\alpha$-divergence minimisation in the GMM case).




Our work hence provides a theoretical framework which enables simultaneous optimisation of the mixture weights $\lbd{}$ and of the mixture components parameters $\Theta$. In addition, beyond the gradient-based updates, we propose the novel maximisation updates and we allow for covariance matrices optimisation. We also emphasise that our maximisation approach will apply to well-chosen kernels $k$ beyond the Gaussian case, as we will detail in \Cref{sec:expo-family}. \newline





\subsection{The Power Descent Algorithm \citep{daudel2020infinitedimensional}}

In order to identify the connection between our work and the Power Descent algorithm introduced in \cite{daudel2020infinitedimensional}, let us briefly present the latter. 

The Power Descent is a \textit{gradient-based} algorithm that operates on probability measures and performs $\alpha$-divergence minimisation for all $\alpha \in \rset \setminus \lrcb{1}$. More precisely, \kdtxt{equipping $\Tset$ with a $\sigma$-field $\Tsigma$ and} denoting by $\meas{1}(\Tset)$ the space of probability measures on \kdtxt{$(\Tset, \Tsigma)$}, the Power Descent optimises $\Psif(\mu k)$ with respect to $\mu \in \meas{1}(\Tset)$, where $\mu k(y) = \int_\Tset \mu(\rmd \theta) k(\theta,y)$ for all $\mu\in \meas{1}(\Tset)$ and all $y \in \Yset$. Given an initial probability measure $\mu_1 \in \meas{1}(\Tset)$, it does so by performing several one-step transitions of the Power Descent algorithm:
\begin{equation}
\label{eq:def:mu}
\mu_{n+1}= \iteration (\mu_n) \eqsp, \quad n \geq 1 \eqsp,
\end{equation}
where, for all $\mu \in \meas{1}(\Tset)$ and all $\theta \in \Tset$,
\begin{align*}
\bmuf(\theta) &= \mathlarger{\int_\Yset} k(\theta,y)  \frac{1}{\alpha-1} \lrb{\left(\dfrac{\mu k(y)}{p(y)}\right)^{\alpha-1}-1} \nu(\rmd y) \\
\iteration (\mu)(\rmd \theta) &= \dfrac{\mu(\rmd \theta) ~ [(\alpha-1)(\bmuf(\theta)+\cte) + 1]^{\frac{\eta}{1-\alpha}}}{\mu([(\alpha-1)(\bmuf+\cte)+1]^{\frac{\eta}{1-\alpha}})} \eqsp. 
\end{align*}
\cite{daudel2020infinitedimensional} motivated the Power Descent algorithm by establishing a monotonicity result for this algorithm obtained as a particular case of \cite[Theorem 1]{daudel2020infinitedimensional}. In their result, the monotonic decrease in $\Psif$ of the scheme \eqref{eq:def:mu} holds for all $\mu \in \meas{1}(\Tset)$, all $\alpha \in \rset \setminus \lrcb{1}$, all $\eta \in (0,1]$ and all $\cte$ such that $(\alpha-1)\cte \geq 0$.} We provide below a more general version of their result, where the monotonic decrease in $\Psif$ holds for well-chosen values of $\eta$ that are larger than $1$ when $\alpha < 0$ (we refer to \Cref{sec:power} for the proof of this result).
\begin{prop}
\label{thm:admiss}
Assume that $p$ and $k$ are as in \ref{hyp:positive}. Let $(\alpha, \eta)$ belong to any of the following
cases.
\begin{enumerate}[label=(\roman*)]
\item\label{item:admiss-alpha-a} $\alpha \leq -1$ and $\eta \in (0,(\alpha-1)/\alpha]$;
\item\label{item:admiss-alpha-b}  $\alpha \in (-1,0)$ and $\eta \in (0,1-\alpha]$;
\item\label{item:admiss-alpha-c} $\alpha \in [0, 1)$ or $\alpha > 1$ and $\eta \in (0,1]$.
\end{enumerate}
Moreover, let $\mu \in \meas{1}(\Tset)$ be such that $\Psif(\mu k)<\infty$ and let $\cte$ be such that $(\alpha-1)\cte \geq 0$. Then, the two following assertions hold.
\begin{enumerate}[label=(\roman*)]
\item \label{item:mono1Prev} We have  $\Psif (\iteration (\mu) k) \leq \Psif(\mu k)$.
\item \label{item:mono2Prev} We have $\Psif (\iteration (\mu) k) =\Psif(\mu k)$ if and only if $\mu=\iteration (\mu)$.
\end{enumerate}
\end{prop} 
Building on the monotonicity result provided by \cite[Theorem 1]{daudel2020infinitedimensional} for the Power Descent algorithm - that we generalised in \Cref{thm:admiss} - \kdtxt{\cite{daudel2020infinitedimensional} then applied this algorithm to mixture weights optimisation by letting the initial probability measure $\mu_1 \in \meas{1}(\Tset)$ be a weighted sum of Dirac measures of the form $\mu_1 = \sum_{j=1}^J \lambda_{j, n} \delta_{\theta_j}$, with $\Theta \in \Tset$ and $\lbd{1} \in \simplex_J$. For that choice of $\mu_1$, $\mu_n$ in \eqref{eq:def:mu} can be written as $\mu_n = \sum_{j=1}^{J} \lambda_{j,n} \delta_{\theta_{j}}$ at time $n$ and the Power Descent amounts to performing the update
\begin{align}\label{eq:PDmixtureSpecial}
  \lambda_{j, n+1} &= \dfrac{\lambda_{j,n} ~ [(\alpha-1)(\bmuf[\mu_n](\theta_j)+\cte) + 1]^{\frac{\eta}{1-\alpha}}}{\sum_{\ell = 1}^J \lambda_{\ell, n} [(\alpha-1)(\bmuf[\mu_n](\theta_\ell)+\cte)+1]^{\frac{\eta}{1-\alpha}}} \eqsp , \quad j = 1 \ldots J.
  \end{align}
  Interestingly, the update \eqref{eq:PDmixtureSpecial} corresponds to the update on the mixture weights \eqref{eq:aPMC:updateOne} we have identified in \Cref{thm:WeightsMixture} for $\Theta_{n} = \Theta$, $\eta_n = \eta/(1-\alpha)$ and $\cte_n = \cte$. Steaming from this link between our approach and \cite{daudel2020infinitedimensional}, we now make two important comments. 
\begin{itemize}
  \item \textit{Benefits of our approach compared to \cite{daudel2020infinitedimensional}.} The monotonicity result from \cite{daudel2020infinitedimensional} generalised in \Cref{thm:admiss} requires the sequence $(\Theta_n)_{n\geq 1}$ to be constant when applied to mixture weights optimisation in \eqref{eq:PDmixtureSpecial}. This restricts the variational family to mixture models with fixed mixture components parameters. 
  To remedy this problem, \cite{daudel2020infinitedimensional} proposed a fully-adaptive algorithm that alternates between an Exploitation step optimising the mixture weights according to \eqref{eq:PDmixtureSpecial} and an Exploration step acting on the mixture components parameters. However, they established no theoretical guarantees for their complete Exploitation-Exploration algorithm, as the choice of the Exploration step remained mostly unexplored.


  A strong improvement of our approach is then that we provide theoretically-sound updates for $(\lbd{n}, \Theta_n)_{n \geq 1}$, with the particularity that our mixture weights updates relate to the framework of \cite{daudel2020infinitedimensional}. In that sense, our approach supplements the work done in \cite{daudel2020infinitedimensional} (albeit by an entirely different proof technique). 
  Furthermore, we do not need to alternate between mixture weights and mixture components parameters updates in our algorithms, as both can be carried out simultaneously (as done in Algorithms \ref{algo:GGEMalpha} and \ref{algo:GGEMalphaGD}). In practice, this will permit us to reduce the computational cost (as the samples will be shared throughout the mixture weights and the mixture components parameters approximated updates, see \Cref{sec:practicalities}). 

  \item \textit{A gradient-based mixture weights update.} \cite{daudel2020infinitedimensional} establishes that the Power Descent belongs to a family of gradient-based algorithms which includes the Entropic Mirror Descent algorithm, a typical optimisation algorithm for optimisation under simplex constraints, as a special case. Viewed from this angle, the parameter $\eta$ in \eqref{eq:PDmixtureSpecial} can be understood as a learning rate with $\bmuf[\mu_n]$ playing the role of the gradient of $\Psif$. Connecting the Power Descent to our framework thus sheds light on the gradient-based nature of the mixture weights update \eqref{eq:aPMC:updateOne} from \Cref{thm:WeightsMixture} and provides a better understanding of the role of the parameter $\eta_n$ appearing in this update. This aspect will be helpful to interpret our numerical experiments in \Cref{subsec:Exp}.
  
\end{itemize}
}

\subsection{The M-PMC Algorithm \citep{cappe2008adaptive}} 


For any measurable positive function $p$ on $(\Yset,\Ysigma)$, the M-PMC algorithm \citep{cappe2008adaptive} aims at solving the optimisation problem
\begin{align*}
\sup_{(\lbd{} \in \simplex_J, \Theta \in \Tset^J)} \int_\Yset \log \left( \sum_{j=1}^{J} \lambda_j k(\theta_j, y) \right) p(y) \nu(\rmd y) \eqsp,
\end{align*}
or equivalently, at minimising the inclusive Kullback-Leibler divergence $\diverg[0] \couple[\mu_{\lbd{}, \Theta} K][\PP]$ w.r.t. to $(\lbd{}, \Theta)$, where $\PP(A) = \int_A p~\rmd \nu / \int_\Yset p~\rmd \nu $ for all $A \in \Ysigma$. This is done in \cite[Section 2]{cappe2008adaptive} by introducing the following iterative updates: for all $j = 1 \ldots J$ and all $n \geq 1$,
\begin{align*}
\lambda_{j,n+1} &= \int_\Yset \frac{\lambda_{j,n} k(\theta_{j,n}, y)}{\sum_{\ell = 1}^J \lambda_{\ell,n} k(\theta_{\ell,n}, y)} \frac{p(y)}{\int_\Yset p(y) \nu(\rmd y)}\nu(\rmd y) \\ 
\theta_{j,n+1}&= \argmax_{\theta_j \in \Tset} \int_\Yset \frac{\lambda_{j,n} k(\theta_{j,n}, y)}{\sum_{\ell = 1}^J \lambda_{\ell,n} k(\theta_{\ell,n}, y)} \log(k(\theta_{j}, y)) p(y) \nu(\rmd y) \eqsp. 
\end{align*}
\cite{cappe2008adaptive} motivated the two updates above by noticing that they can be seen as integrated versions under the target distribution of the update formulas for the Expectation-Maximisation (EM) algorithm applied to the mixture density parameter estimation problem
\begin{align*}
\sup_{(\lbd{} \in \simplex_J, \Theta \in \Tset^J)} \sum_{m=1}^M \log \left( \sum_{j=1}^{J} \lambda_j k(\theta_j, Y_m) \right) \eqsp,
\end{align*}
meaning that these updates ensure a systematic increase in the integrated likelihood at each step. As these updates also correspond to the case $\alpha = 0$, $\eta_n = 1$, $\cte_n =0$, $a_{j,n} = 1$ and $b_{j,n} = 0$ in Algorithm \ref{algo:GGEMalpha}, the M-PMC algorithm is in fact included in our framework and we can interpret our theoretical results in light of this algorithm. 

More precisely, we have generalised an integrated EM algorithm by preserving its monotonicity property for a wide range of hyperparameters. A particularly striking fact is that the monotonicity property holds for $\alpha \in [0,1)$, hence updates akin to an EM procedure can be derived past the traditional case of likelihood optimisation. As we shall see, the additional layers of flexibility obtained in Algorithm \ref{algo:GGEMalpha} will also have important practical consequences due to the underlying gradient-based structure behind the mixture weights (and the mixture components parameters updates in the GMMs case) we have uncovered.  \newline 

\Cref{table:recap} summarises the main improvements of our framework compared to the existing literature and in the coming section, we revisit the maximisation and the gradient-based approaches when the variational family is based on the exponential family.


\begin{table}
  \centering
\begin{tabular}{ll}
      \toprule
    & Improvements of our framework \\
       \midrule
    R\'{e}nyi Gradient Descent & Simultaneous optimisation w.r.t. $(\lbd{}, \Theta)$ \\
    \cite{2016arXiv160202311L} & (prev. mixture weights $\lbd{}$ optimisation not considered) \\
     & For GMMs : maximisation approach encompasses R\'{e}nyi \\
     & Gradient Descent and provides covariance matrices updates \\
       \midrule
  Power Descent & Simultaneous optimisation w.r.t. $(\lbd{},\Theta)$ \\
  \cite{daudel2020infinitedimensional}
  & (prev. $(\Theta_n)_{n \geq 1}$ constant)\\
      \midrule
  M-PMC algorithm & Extension of an Integrated EM algorithm to: \\
  \cite{cappe2008adaptive} & $\alpha \in [0,1)$, $\eta_n \in (0,1]$, $(\alpha -1) \cte_n \geq 0$ and $b_{j,n} \geq 0$ \\
  & (prev. $\alpha = 0$, $\eta_n = 1$, $\cte_n = 0$ and $b_{j,n} = 0$) \\
  \bottomrule
\end{tabular}
\caption{Key improvements of our framework for mixture models optimisation compared to related works (prev. stands for previously in the literature).}   \label{table:recap}
\end{table}

\section{Exponential Family Distributions: a Closer Look}
  \label{sec:expo-family}
  \kdtxt{
  In this section, we state generic results for the maximisation and the gradient-based approaches in the important case where the variational family is based on the exponential family. 
  Those results will in particular enable us to show how the update formulas in Examples \ref{ex:GMMmax} and \ref{ex:GMM} are derived, before allowing us to incorporate novel variational families in our framework such as multivariate Student's $t$ densities. To do so, let us first review the exponential family and some of its well-known properties.}

\subsection{Notation and Useful Definitions}
  \label{sec:expo-family-notation}
  We start with the definition of an exponential family distribution.
  
  \begin{defi}[Exponential family]\label{def:expo-family}
    Let $\nu$ be a $\sigma$-finite measure on $(\Yset,\Ysigma)$, $E$
    be a Euclidean space
    endowed with an inner product $\pscal[E]{\cdot}{\cdot}$ and its
    corresponding norm $\normev[E]{\cdot}$,  $h$ be a non-negative function
    defined on $\Yset$ and  $S:\Yset\to E$ be a measurable
    function. 
    In its canonical form, \kdtxt{a member of the exponential family}
    is defined \kdtxt{by the
      parametric probability density}
      $$
      k^{(o)}(\zeta,y) = h(y)\;\exp\left(\pscal[E]{\zeta}{S(y)}-A(\zeta)\right)\;,\quad \;y\in\Yset\;, \eqsp \zeta\in E_0
      $$
      where
      $$
      E_0:=\set{ \zeta\in E}{\int_\Yset 
        h(y) \;\rme^{\pscal[E]{\zeta}{S(y)}}\; \nu (\rmd y)<\infty }
      $$
      and where the values taken by $A$ on the subset $E_0$ are
      obtained from the normalising
      constraint $\int k^{(o)}(\zeta,\cdot)\;\rmd\nu=1$. In this
      setting, $\nu$ is called the dominating measure and $S$ the
      natural statistic. In what follows,  $\intE_0$ denotes the interior of
      $E_0$. 
    \end{defi}  
    Recall now that $S$ defines an exhaustive statistic for this model,
    that $E_0$ is a convex subset of $E$ and that $A$ is infinitely
    differentiable in $\intE_0$. Moreover, for all $\zeta\in\intE_0$,
    the expectation and covariance operator of $S$ under the density
    $k^{(o)}(\zeta,\cdot)$ are the gradient $\nabla A$ and the Hessian
    $\nabla\nabla^TA$ of $A$ taken at $\zeta$, respectively. For later reference, we also recall
      \begin{align}\label{eq:partialA:expo-family}
        \nabla A(\zeta)= \int_\Yset S(y)\,k^{(o)}(\zeta,y)\;\nu(\rmd y)\;
      \end{align}
      and we introduce the following assumption.
      \begin{hyp}{B}
        \item \label{hyp:expo} The kernel density $k^{(o)}$ defined on
          $E_0 \times \Yset$ is a member of the exponential family as in
          \Cref{def:expo-family}. Moreover, we assume that:
        \begin{enumerate}[label=(\alph*)]
          \item \label{item:expo2} The function $h$ on $\Yset$ is positive.
          \item \label{item:expo1} There is no  affine hyperplane of
            $E$ to which $S(y)$ belongs for $\nu$-almost all $y\in \Yset$.
        \end{enumerate}
      \end{hyp}
      Notice that the assumption \ref{hyp:expo}\ref{item:expo2} can be
      circumvented by changing the dominating measure $\nu$ and
      therefore the assumptions \ref{hyp:expo}\ref{item:expo2} and
      \ref{hyp:expo}\ref{item:expo1} should be considered
      together. These assumptions notably imply that the covariance
      operator of $S$ under the density $k^{(o)}(\zeta,\cdot)$ is
      positive definite and thus so is $\nabla\nabla^TA(\zeta)$ for
      all $\zeta\in\intE_0$.  In addition, one often relies on a
      non-canonical parameterisation (see for example
      \Cref{ex:Gaussian}), leading to the kernel density
      \begin{align}\label{eq:partialA:expo-family-non-canonical}
      k(\theta,y)=k^{(o)}(\changevarcomp(\theta),y) = h(y)\;\exp\left(\pscal[E]{\changevarcomp(\theta)}{S(y)}-A\circ\changevarcomp(\theta)\right)\;,\quad \theta\in\Tset\;,\;y\in\Yset\;
      \end{align}
      where $\changevarcomp$ maps the parameter space $\Tset$ to a
      subset of $\intE_0$. In that case, the assumption \ref{hyp:expo}\ref{item:expo2} will ensure that the assumption made on $k$ in \ref{hyp:positive} holds.

      \subsection{The Maximisation Approach for the Exponential Family}

      \label{sec:argm-solut-param-expo}

      Recall that the maximisation approach proposed in \Cref{coro:argmax} aims at solving the optimisation problem \eqref{eq:updateTheta} given by
      $$\theta_{n+1} = \argmax_{\theta \in \Tset} \int_\Yset \lrb{ \ratio(y) + b_n k(\theta_n , y)} \log \lr{\frac{k(\theta, y)}{k(\theta_{n}, y)}} \nu(\rmd y) \eqsp, \quad n \geq 1\;.
      $$
      Furthermore, in the more general case of mixture models described in \Cref{coro:argminMixtureModel}, the maximisation approach leads to the argmax problem \eqref{eq:aPMC:updateP}. This second argmax problem is similar to \eqref{eq:updateTheta} in the sense that it updates $\Theta_n=(\theta_{j,n})_{1\leq j \leq J}$ by solving $J$ component-wise argmax problems of the same form as \eqref{eq:updateTheta} (with
      $\theta_{n+1}$, $\ratio$, $b_n$ and $\theta_n$ being replaced by
      $\theta_{j,n+1}$, $\ratiogen$, $b_{j,n}$ and $\theta_{j,n}$,
      for $j=1 \dots J$):
       $$
       \theta_{j,n+1} = \argmax_{\theta \in \Tset} \int_\Yset \kdtxt{\lrb{ \respat[y] + b_{j,n} k(\theta_{j,n}, y)}} \log \lr{\frac{k(\theta, y)}{k(\theta_{j,n}, y)}} \nu(\rmd y) \eqsp, \quad j = 1 \ldots J \eqsp , \eqsp n \geq 1 \eqsp.
       $$ 

Our goal is now to solve the argmax problems above for a member of the exponential family. To do so, we first state a theorem in which the kernel $k$ is in its canonical form, that is $v$ is the identity mapping  in \eqref{eq:partialA:expo-family-non-canonical} and $k(\theta, \cdot) = k^{(o)}(\zeta, \cdot)$. Other parameterisations will then be built using this theorem and as a result we will deduce corollaries that include non-canonical parameterisations and that are applicable to the general case of mixture models. All the proofs of the results stated in \Cref{sec:argm-solut-param-expo} are deferred to \Cref{sec:argm-solut-param-expo:proofs}.
            
  \begin{thm}\label{thm:generic-update-canonic-expo-family}
    Let $k^{(o)}$ satisfy \ref{hyp:expo} and $\mathcal{O}$ be an open
    subset of $E_0$. Let $\check{\varphi}$ be a probability density function with
    respect to the measure $\nu$ such
    that
  \begin{equation}
    \label{generic-update-canonic-expo-family-hyp}
 \int \normev[E]{S}\,\check{\varphi} \;\rmd\nu<\infty\;.
\end{equation}
Let $\zeta_0\in \mathcal{O}$ and $b\geq0$.
Then the argmax problem
   \begin{equation}
    \label{generic-update-canonic-expo-family-argmax}
  \argmax_{\zeta\in\mathcal{O}}     \int_\Yset \lrb{\check{\varphi}(y) + b \, k^{(o)}(\zeta_0 , y)}
 \log\left(\frac{k^{(o)}(\zeta, y)}{k^{(o)}(\zeta_0, y)}\right) \nu(\rmd y)
\end{equation}
admits at least a solution if and only if
  \begin{equation}
    \label{generic-update-canonic-expo-family-cond}
  \mathbf{s}^*:=\frac1{1+b}  \int S \, \check{\varphi}\;\rmd \nu + \frac
  b{1+b} \nabla
  A(\zeta_0)    
\end{equation}
belongs to the image set $\nabla A(\mathcal{O})$, in which
case~(\ref{generic-update-canonic-expo-family-argmax}) has a unique
solution $\zeta^*$ defined by
  \begin{equation}
    \label{generic-update-canonic-expo-family-solution}
  \nabla A(\zeta^*) = \mathbf{s}^*\;.
\end{equation}
\end{thm}
\Cref{thm:generic-update-canonic-expo-family} shows the equivalence between the argmax problem \eqref{generic-update-canonic-expo-family-argmax} and the equation~(\ref{generic-update-canonic-expo-family-solution}) under the canonical parameterisation. From there, using the expression of $\nabla A$ given in \eqref{eq:partialA:expo-family} and considering the non-canonical exponential family probability density \eqref{eq:partialA:expo-family-non-canonical}, we can interpret the equation~(\ref{generic-update-canonic-expo-family-solution}) as an
equality between two means of the statistic $S$ computed under two
different distributions. Setting
$\zeta^*=\changevarcomp(\theta^*)$ and
$\zeta_0=\changevarcomp(\theta_0)$, we indeed obtain
\begin{equation}    \label{eq:generic-update-canonic-expo-family-solution-theta}
  \int_\Yset S(y) k(\theta^*, y)\nu(\rmd y) = \int_\Yset S(y) \psi(y)\nu(\rmd y)
\end{equation}
where $\psi$ is
the mixture
\begin{equation*}
  \psi(y)= \frac1{1+b}  \check{\varphi}(y) + \frac
  b{1+b} k(\theta_0,y) \eqsp, \quad y \in \Yset \;.
\end{equation*}
As exemplified in the following corollary, an adequate parameterisation $\zeta=\changevarcomp(\theta)$ may then lead to a simple solution of \eqref{eq:generic-update-canonic-expo-family-solution-theta}, which in turn provides a solution to the argmax problem
\begin{equation}
  \label{eq:gaussian-mapping-to-maximize-one}
\argmax_{\theta \in \Tset} \int_\Yset \lrb{\check{\varphi}(y) + b k(\theta_0, y)} \log \lr{\frac{k(\theta, y)}{k(\theta_0, y)}} \nu(\rmd y) \eqsp.
\end{equation}

\begin{coro}[Gaussian density] \label{coro:GaussianTemp} Let
  $d\geq1$ and $\Yset=\rset^d$.  Let $\mathcal{M}_{>0}(d)$ denote the set
  of symmetric positive definite $d\times d$ matrices.  We consider
  the case of a $d$-dimensional Gaussian density with
  $k(\theta, y) = \mathcal{N}(y; m, \Sigma)$, where
  $\theta = (m, \Sigma) \in \Tset=\rset^d\times\mathcal{M}_{>0}(d)$.
  Let $\check{\varphi}$ be a  probability density function
  on $\Yset$ such that \eqref{generic-update-canonic-expo-family-hyp} holds, that is in the Gaussian case 
  \begin{equation*}
\int_\Yset \|y\|^2\,\check{\varphi}(y) \;\rmd y<\infty\;. 
  \end{equation*}
Let
  $\theta_0= (m_0, \Sigma_0)\in\Tset$ and $b\geq0$. If $b=0$, assume
  moreover that $\check{\varphi}$ is non-degenerate in the sense that
  its covariance matrix is positive definite. Then \eqref{eq:gaussian-mapping-to-maximize-one}
  has a unique solution $\theta^* = (m^*, \Sigma^*) \in \Tset$ defined by
  $m^* = \PE[Y]$ and $\Sigma^* = \Cov(Y)$, where 
    $Y$ is a random variable valued in $\rset^d$ with density
    $$
    \psi(y) = \frac{1}{1+b} \check{\varphi}(y) + \frac{b}{1+b} \mathcal{N}(y; m_0, \Sigma_0) \eqsp, \quad y \in \Yset \eqsp.
    $$
  \end{coro}
 It now remains to choose $\check{\varphi}$ adequately to relate \eqref{eq:gaussian-mapping-to-maximize-one} to \eqref{eq:updateTheta}, or rather to its generalisation \eqref{eq:aPMC:updateP}. That way, we will in particular get that the update formulas given in Examples \ref{ex:Gaussian} and \ref{ex:GMMmax} are straightforward consequences of \Cref{coro:GaussianTemp}. This is the purpose of \Cref{cor:iterative-updates-expo-family} below, in which we also rewrite the assumption \eqref{generic-update-canonic-expo-family-hyp} made on $\check{\varphi}$ under a more convenient form.
  
  

\begin{coro}\label{cor:iterative-updates-expo-family}
  Let $k^{(o)}$ satisfy \ref{hyp:expo}. Further assume that the kernel $k$ is of the
  form \eqref{eq:partialA:expo-family-non-canonical} with a one-to-one 
  mapping $\changevarcomp$ defined on $\Tset$
  such that its image $\changevarcomp(\Tset)$ is an open subset of
  $E_0$. Let $\alpha \in [0,1)$ and let
  $p:\Yset\to\rset_+$ be such that
  \begin{equation}
    \label{generic-update-canonic-expo-family-cond-on-p}
  0<  \int_\Yset \left(1+ \norm[E]{S(y)}^{1/(1-\alpha)}\right)\; p(y) \;\nu(\rmd y) < \infty \;
\end{equation}
holds. Let $J \in \nset$, $n \geq 1$ and
$(\lbd{n},\Theta_n) \in \simplex_J^+ \times \Tset^J$. 
Then, for all $j = 1 \ldots J$, \eqref{generic-update-canonic-expo-family-hyp} holds with $\check{\varphi} = \normratiot$ as defined in \eqref{eq:respat} and setting $\gamma_{j,n} = \int \ratiogen \rmd \nu / (\int \ratiogen \rmd \nu +b_{j,n})$, we have that
\begin{align}\label{eq:maxApproachGen}
\argmax_{\theta \in \Tset} \int_\Yset \kdtxt{\lrb{ \respat[y] + b_{j,n} k(\theta_{j,n}, y)}} \log \lr{\frac{k(\theta, y)}{k(\theta_{j,n}, y)}} \nu(\rmd y)
\end{align}
admits at least a solution if and only if
\begin{align}\label{eq:maxApproachGen-sjdef}
\mathbf{s}_j^* \eqdef \gamma_{j,n} \int S ~ \normratiot ~\rmd \nu + (1 - \gamma_{j,n}) \nabla A\circ\changevarcomp(\theta_{j,n}) 
\end{align}
belongs to the image set $\nabla A\circ\changevarcomp(\Tset)$, in which case \eqref{eq:maxApproachGen} has a unique solution $\theta_j^*$ defined by
\begin{align*} 
\nabla  A\circ\changevarcomp(\theta_j^*) = \mathbf{s}_j^*
\end{align*}
and we can set $\theta_{j,n+1} = \theta^*_j$ for all $j = 1 \ldots J$
in~(\ref{eq:aPMC:updateP}) of \Cref{coro:argminMixtureModel}.
\end{coro}
Since the argmax problem \eqref{eq:maxApproachGen} is immune to change of variables thanks to its argmax form, the maximisation approach can be solved using the canonical parameter $\zeta$ in \Cref{thm:generic-update-canonic-expo-family} and transported to the parameter $\theta$ via the one-to-one mapping $v$ as written in \Cref{cor:iterative-updates-expo-family}.

The updates from Examples \ref{ex:Gaussian} and \ref{ex:GMMmax} then follow from \Cref{cor:iterative-updates-expo-family}, paired up with \Cref{coro:GaussianTemp} and with the fact that \eqref{generic-update-canonic-expo-family-cond-on-p} simplifies to \eqref{eq:gaussian-p-cond} (since in the Gaussian case we have set $S(y) = (y, - y y^T/2)$ for all $y \in \Yset$). More generally, we can always solve the argmax problem \eqref{eq:maxApproachGen}, that is the argmax problem of
  \Cref{coro:argminMixtureModel} with $k$ belonging to the exponential family, for $b_{j,n} > 0$ large enough.\footnote[1]{Under the assumptions of \Cref{cor:iterative-updates-expo-family},
  $\changevarcomp(\Tset)$ is an open subset of $E_0$ and $\nabla A$ is a $\mathcal{C}^\infty$-diffeomorphism on $\intE_0$. Hence, $\nabla A\circ\changevarcomp(\Tset)$ is an open subset of $E$. Since $\mathbf{s}_j^*$ in~(\ref{eq:maxApproachGen-sjdef}) tends to
  $\nabla A\circ\changevarcomp(\theta_{j,n})$ as $b_{j,n}\to \infty$, $\mathbf{s}_j^*$ ends up belonging to
  $\nabla A\circ\changevarcomp(\Tset)$ for $b_{j,n}$ large enough. The desired result follows by using that the regularisation parameter $b_{j,n}>0$ in \Cref{coro:argminMixtureModel} can be freely chosen at each step $n \geq 1$.
  }

  We now move on to the study of the gradient-based approach for the exponential family. 

  \subsection{The Gradient-based Approach for an Exponential Family Distribution}

\label{sec:gradient-smoothness-expo-family} 

Remember that to construct the sequence $(\theta_n)_{n\geq 1}$ according to the gradient-based approach from \Cref{coro:gradientDescent} we need to compute the gradient of the function $g_n$ defined in \eqref{eq:coroDefGn} by
\begin{align*}
  g_n(\theta) = \int_\Yset \frac{\kdtxt{\ratio(y)} }{\alpha-1} \log \lr{\frac{k(\theta, y)}{k(\theta_{n}, y)}} \nu(\rmd y) \eqsp, \quad \theta \in \Tset \eqsp,
  \end{align*}
and verify that $g_n$ satisfies a $\beta_n$-smoothness condition, which leads to updates of the form
$$
\theta_{n+1} = \theta_n - \frac{\gamma_n}{\beta_n} \nabla g_n(\theta_n) \eqsp, \quad n \geq 1
$$
with $\gamma_n \in (0,1]$. Furthermore, the extension to mixture models in \Cref{coro:GDMixtureModel} involves the gradient of $J$ functions $g_{j,n}$ that resemble $g_n$ ($g_n$, $\ratio$ and $\theta_n$ are replaced by $g_{j,n}$, $\ratiogen$ and $\theta_{j,n}$ in the definition of $g_n$ above) and where each function $g_{j,n}$ is assumed to be $\beta_{j,n}$-smooth. 

Let us now solve the gradient-based approach for a member of the exponential family. Like in the previous section, we start with a result handling the case where the kernel $k$ is in its canonical form (that is, $v$ is the identity mapping in \eqref{eq:partialA:expo-family-non-canonical} and $k(\theta, \cdot) = k^{(o)}(\zeta, \cdot)$). All the proofs of the results stated in \Cref{sec:gradient-smoothness-expo-family} are deferred to \Cref{sec:gradient-smoothness-expo-family:proofs}.

  \begin{prop}\label{thm:gradiant-g-canonic-expo-family}
    Let $k^{(o)}$ satisfy \ref{hyp:expo}. Let $\check{\varphi}$ be a
    probability density function with respect to $\nu$
    satisfying~(\ref{generic-update-canonic-expo-family-hyp}). 
  Let
    $\zeta_0\in \intE_0$ and set
  \begin{equation}
    \label{generic-update-canonic-expo-family-g-def}
  g^{(o)}(\zeta) = -\int_\Yset \check{\varphi}(y)\;\log\left(\frac{k^{(o)}(\zeta, y)}{k^{(o)}(\zeta_0, y)}\right) \nu(\rmd y) \eqsp, \quad  \zeta \in \intE_0 \;.
\end{equation}
Then, for any convex subset $C_0\subseteq\intE_0$, the function $g^{(o)}$
is $\beta_0$-smooth over $C_0$ with
\begin{align}     \label{generic-update-canonic-expo-family-g-smooth}
\nabla g^{(o)}(\zeta)=\nabla A(\zeta)-\int S\,\check{\varphi}\;\rmd\nu
\quad\text{and}\quad\beta_0 \geq \sup_{\zeta\in C_0}\normop{\nabla\nabla^T A(\zeta)}\;,
\end{align}
where $\normop{\cdot}$ denotes the operator norm.
\end{prop}
Letting $C_0$ be the line segment $\{\zeta_0-\frac{t}{\beta_0}\left(\nabla A(\zeta_0)-\int S\,\check{\varphi}\;\rmd\nu\right) \eqsp : \eqsp {t\in[0,1]} \}$ in \Cref{thm:gradiant-g-canonic-expo-family}, we have that $C_0$ stays in $\intE_0$ for $\beta_0$ large enough. Since the Hessian of $A$ is continuous on $\intE_0$, the inequality in \eqref{generic-update-canonic-expo-family-g-smooth} can also be satisfied. Hence, the following gradient step
\begin{align} \label{generic-update-canonic-expo-family-g-gradientstep}
\zeta = \zeta_0 - \frac{\gamma}{\beta_0} \lr{\nabla A(\zeta_0)-\int_\Yset S\,\check{\varphi}\;\rmd\nu} \eqsp
\end{align}
can always be performed for all $\gamma \in (0,1]$ and we can always apply the gradient-based approach when the kernel $k$ is in its canonical form. When considering non-canonical exponential family probability
densities of the form \eqref{eq:partialA:expo-family-non-canonical} and
under some reasonable assumptions on the mapping $\changevarcomp$ (for
instance, if $\changevarcomp$ is a $\mathcal C^2$ diffeomorphism),
one is left with choosing between two alternative ways of exploiting
\Cref{thm:gradiant-g-canonic-expo-family}:
\begin{enumerateList}
\item\label{item:generic-update-expo-family-gradient-update-canonical} Apply a gradient step on the canonical parameter $\zeta$ using
  the objective $g^{(o)}$ in~(\ref{generic-update-canonic-expo-family-g-def}) and translate it
  into an update of the parameter $\theta$ via a change of
  variable. In this case, we set $\zeta_0=\changevarcomp(\theta_0)$, we apply
  the step~(\ref{generic-update-canonic-expo-family-g-gradientstep})
  and we then perform the change of variable $\theta=\changevarcomp^{-1}(\zeta)$ to get back to a parameter in
  $\Tset$. Overall, this leads to an update of the form
\begin{align*}
\theta = \changevarcomp^{-1}\lr{\changevarcomp(\theta_0) -\frac{\gamma}{\beta_0} \lr{\nabla A\circ\changevarcomp(\theta_0)-\int_\Yset S\,\check{\varphi}\;\rmd\nu}} \eqsp.
\end{align*}
\item\label{item:generic-update-expo-family-gradient-update-noncanonical} Apply a gradient step directly on the parameter $\theta$ using
  the objective 
  \begin{equation*}
   g(\theta) =
   -\int\check{\varphi}(y)\;\log\left(\frac{k(\theta, y)}{k(\theta_0, y)}\right) \rmd y \eqsp, \quad \theta \in \Tset\;.
 \end{equation*}
 In that case, since $g=g^{(o)}\circ\changevarcomp$ where
 $g^{(o)}$ is again as in~(\ref{generic-update-canonic-expo-family-g-def}) with
 $\zeta_0=\changevarcomp(\theta_0)$, this leads to an update of the form
\begin{align*}
\theta = \theta_0 - \frac{\gamma}{\beta} \nabla\changevarcomp(\theta_0)\cdot \lr{\nabla A\circ\changevarcomp(\theta_0)-\int_\Yset S\,\check{\varphi}\;\rmd\nu} \eqsp.
\end{align*}
\end{enumerateList}
In both approaches, $\gamma\in(0,1]$ and the smoothness indices
$\beta_0$ and $\beta$ have to be taken large enough in order to
guarantee a decrease of the objective function. In practice, it is not
clear which parameterisation leads to the simpler or more efficient
algorithm. We investigate a particular setting in the following
result.
\begin{coro}[Gaussian density with known covariance matrix] \label{coro:GaussianTempGrad} Let $d\geq1$ and
  $\Yset=\rset^d$.  We consider the exponential family of a
  $d$-dimensional Gaussian density with
  $k(\theta, y) = \mathcal{N}(y; \theta, \Sigma)$, where
  $\theta\in \Tset:=\rset^d$ and $\Sigma$ is a (known) covariance matrix
  in $\mathcal{M}_{>0}(d)$.
  Let $\check{\varphi}$ be a probability density function w.r.t. $\nu$ such that \eqref{generic-update-canonic-expo-family-hyp} holds, that is in the Gaussian case
  \begin{equation*}
  \int \|y\|\,\check{\varphi}(y) \;\rmd y<\infty\;.
\end{equation*}
Set $\changevarcomp(\theta)=\Sigma^{-1}\theta$ and define the canonical kernel
$k^{(o)}$ with $h(y)=(2\pi)^{-d/2}\lrav{\Sigma}^{-1/2}\,\rme^{-y^T\Sigma^{-1}y/2}$,
 $S(y)=y$
and $A(\zeta)=\zeta^T\Sigma\zeta/2$.
Then~(\ref{eq:partialA:expo-family-non-canonical}) holds and the 
methods~\ref{item:generic-update-expo-family-gradient-update-canonical}
and~\ref{item:generic-update-expo-family-gradient-update-noncanonical}
lead to the the following gradient-based updates, respectively,
\begin{align}\label{eq:gradient-step-canonical-gaussian}
  &\theta=\theta_0-\frac{\gamma}{\beta_0}\Sigma\lr{\theta_0-\int y ~ \check{\varphi}(y)
    \;\rmd y}\\
  \label{eq:gradient-step-noncanonical-gaussian}
  &  \theta=\theta_0-\frac{\gamma}{\beta}\Sigma^{-1}\,\lr{\theta_0-\int y  \check{\varphi}(y)
    \;\rmd y}
\end{align}
where $\beta_0$ is the largest eigenvalue of $\Sigma$ and $\beta$ is
the largest eigenvalue of $\Sigma^{-1}$. They correspond to the
smoothness indices of $g^{(o)}$ and $g$ over $\rset^d$, respectively. 
\end{coro} 
\Cref{coro:GaussianTempGrad} illustrates how, even in a
simple framework, the gradient-based updates strongly depend
on the parameterisation chosen. They in fact coincide only if
$\Sigma$ is scalar in the above corollary, in which case
$\Sigma/\beta_0$ and
$\Sigma^{-1}/\beta$ above are both equal to the identity
matrix. 

Following the reasoning of \Cref{sec:argm-solut-param-expo},
we can then apply \Cref{coro:GaussianTempGrad} with $\check{\varphi} = \normratio$ (and more generally with $\check{\varphi}= \normratiot$). As \eqref{generic-update-canonic-expo-family-hyp} is implied by
\eqref{eq:cond-gradient-gauss-ex} (using \Cref{lem:cond-p-varphi} in
\Cref{sec:argm-solut-param-expo:proofs} and that 
here $S(y) = y$), this enables us to deduce the updates in
Examples \ref{ex:ATwostatisfied} and \ref{ex:GMM}. Note that we did not
introduce a convex subset
$C_0$ in \Cref{coro:GaussianTempGrad} (this is due to the fact that
the smoothness index is constant over the whole parameter space for
both gradient-based updates in that case), which is why
$C_0$ does not appear in Examples \ref{ex:ATwostatisfied} and
\ref{ex:GMM}. 

Let us next put into perspective \Cref{coro:GaussianTemp}, in which we considered a Gaussian density with varying mean vector and covariance matrix, and \Cref{coro:GaussianTempGrad}, where the covariance matrix is assumed to be known. While \Cref{coro:GaussianTempGrad} kept the computations straightforward and provided gradient-based updates that do not require to introduce a convex subset $C_0$, this is no longer the case if we take the same exact setting as in \Cref{coro:GaussianTemp}. Indeed, when $\theta=(m,\Sigma)\in\Tset=\rset^d\times\mathcal{M}_{>0}(d)$ and under the condition \eqref{generic-update-canonic-expo-family-hyp}, the gradient of $g$ is given by 
\begin{align*}
\nabla g(\theta)&=\lr{\Sigma^{-1}\lr{m-\int y
                      \check{\varphi}(y)\;\rmd y},
                      \frac12\Sigma^{-1}\lr{\int yy^T  \check{\varphi}(y)\;\rmd y-\Sigma+mm^T}\Sigma^{-1}}\;.
\end{align*}
The smoothness index $\beta$ is now no longer constant here, which makes it much more involved to choose a convenient convex subset $C_0$ and to bound the smoothness index over it. \newline

At this stage, we can solve the maximisation and gradient-based approaches when the kernel $k$ belongs to the exponential family for variational families as large as (finite) mixture models. The maximisation approach appears to be preferable to the gradient-based one as it does not require a smoothness assumption and does not depend on the parameterisation. 

In the next section, we investigate how our maximisation approach may further generalise to include extensions of the exponential family distribution such as Student's $t$ distributions, hence further motivating the maximisation approach over the gradient-based approach.

\subsection{Extension to Linear Mixture (LM) Models}
  \label{sec:partial-mixture-expo-family} 

  The goal of this section is to examine what becomes of the
  maximisation approach for a specific extension of the exponential
  family distribution, which we will refer to as \emph{exponential linear mixture} (ELM) family. The ELM family will in
  particular encompass the Student's $t$ distribution with mean $m$,
  covariance matrix $\Sigma$ and $\nut$ degrees of freedom which can be
  defined as the continuous mixture
  $$
  \mathrm{t}({y}; \mt,\Sigma,\nut) = \int_0^\infty
  \mathcal{N}(y;\mt,z^{-1}\Sigma)\;\check{\tau}_{\nut}(\rmd z)\;,
  $$
where $\check{\tau}_{\nut}$ denotes the $\chi^2$ distribution with
$\nut$ degrees of freedom. As a result, building a
  maximisation approach within the ELM family will lead to new
  iterative schemes for variational families that go beyond the cases
  considered thus far. To this end, let us first provide a precise definition of the ELM family we want to study and introduce the main assumptions of this section. All the proofs from \Cref{sec:partial-mixture-expo-family} are deferred to \Cref{sec:proof:mix:exp}.

  \subsubsection{The (E)LM Family: Definition and Notation}

  Let $F$ be an Euclidean space endowed with the inner product $\pscal[F]{\cdot}{\cdot}$ and norm $\normev[F]{\cdot}$. We denote by $\mathcal{L}(F,E)$ and $\mathcal{L}(F)$ the (finite-dimensional) linear spaces of linear
  operators from $F$ to $E$ and from $F$ to itself, respectively. We now provide the definition of the (E)LM family.

  \begin{defi}[(E)LM family] \label{defi:LMEFfam} Let
    $k^{(o)}:E_0\times\Yset\to\rset_+$ be a kernel density w.r.t. the
    dominating measure $\nu$, with $E_0\subseteq E$ and let $\changevar:\tilde{\Tset}\to F$. Set
    $\Zset=\mathcal{L}(F,E)$ and denote by $\Zsigma$ its Borel
    $\sigma$-field. Let $\parammixproba$ be a subset of probability measures on $(\Zset,\Zsigma)$ satisfying 
    \begin{enumerate}[label=(\roman*)]
      \item\label{asump:equiv-calT} for all $\tau,\tau'$ in $\parammixproba$,
          the distributions $\tau$ and $\tau'$ are equivalent;
        \item\label{asump:gradA-well-def}  for all $\tau\in\parammixproba$
          and  $\tau $-almost all $\ell \in \Zset$,
          the image set $\ell\circ\changevar( \tilde{\Tset})$ is included in  $E_0$.
        \end{enumerate}
      A member of the  \emph{linear mixture} (LM) family with canonical kernel $k^{(o)}$, mixing
      class $\parammixproba$ and parameter mapping $\changevar$, is
      defined for a parameter $\theta=(\vartheta,\tau)\in\Tset=\tilde{\Tset}\times\parammixproba$ by the density w.r.t. $\nu$
    \begin{align}\label{eq:DefiKernelLinMixt}
      k\frtxt{^{(1)}}(\theta,y)= \int_\Zset k^{(o)}(\ell\circ
        \changevar(\vartheta),y)\; \tau (\rmd \ell)\;,\qquad y\in\Yset\;.
    \end{align}
    If moreover $k^{(o)}$  satisfies~\ref{hyp:expo}, we say that
    $k^{(1)}(\theta,\cdot)$ is a member of the \emph{exponential linear
      mixture} (ELM) family. 
  \end{defi}
Here, Assumption~\ref{asump:equiv-calT}
implies that we only need to check
Assumption~\ref{asump:gradA-well-def} for one $\tau$ in $\parammixproba$.
As for Assumption~\ref{asump:gradA-well-def}, it ensures that for all
$\vartheta\in\tilde{\Tset}$ and all $\tau\in\parammixproba$, the function
$(y,\ell)\mapsto k^{(o)}(\ell\circ \changevar(\vartheta),y)$ appearing in
\eqref{eq:DefiKernelLinMixt} is a probability density function w.r.t. $\nu\otimes\tau$, so that for all $\theta \in \Tset$, the
function $y \mapsto k^{(1)}(\theta,y)$ is a probability density
function w.r.t. $\nu$. 

Furthermore, the mapping $\changevar$ in \Cref{defi:LMEFfam} allows us to introduce various parameterisations. It is also important to note that for $\theta=(\vartheta,\tau)\in\Tset$, the probability density function
$k^{(1)}(\theta,\cdot)$  only depends on $\xi=\changevar(\vartheta)\in\changevar(\tilde{\Tset})$ and
$\tau\in\parammixproba$. As we shall see later, Student's $t$ distributions in particular will
fit this general setting. Before that, let us state conditions leading
to a systematic decrease in $\Psif$ when the kernel $k^{(1)}$ is as in
\Cref{defi:LMEFfam} and from there, let us see how a maximisation
approach for the ELM family can be derived.

\subsubsection{Monotonic Decrease Conditions for the (E)LM Family}
\label{sec:monot-decr-cond-LM}

In the case of an ELM family, we are not able to directly solve the
argmax problem~(\ref{eq:aPMC:updateP}) as we did for the exponential
family in \Cref{sec:argm-solut-param-expo}. Instead, given an (E)LM family $k^{(1)}$ as in \Cref{defi:LMEFfam}, we come back to
a monotonic decrease condition of the form
\begin{equation}
  \label{eq:generic-decrease-theta-cond}
 \int_\Yset \check{\varphi}(y) \log\lr{\frac{k^{(1)}(\theta,y)}{k^{(1)}(\theta_0,y)}} \nu(\rmd y)
  \geq 0 \;.
\end{equation}
As detailed in the remark below, for selected choices of $\check{\varphi}$ and of $(\theta,\theta_0)$, \eqref{eq:generic-decrease-theta-cond} can indeed be linked to the conditions~(\ref{eq:ineqThetaMainTwo}) and~(\ref{eq:posMixtureP}) appearing in \Cref{thm:EMtheta} and \Cref{thm:EM:MixtureModel} respectively.

\begin{rem}\label{rem:generic-cond-elm}
Let $\alpha\in[0,1)$. Taking $\check{\varphi}=\normratio$ defined by~(\ref{eq:ratio}) and $(\theta,\theta_0) = (\theta_{n+1},\theta_n)$,
(\ref{eq:generic-decrease-theta-cond}) becomes the 
condition~(\ref{eq:ineqThetaMainTwo}) of \Cref{thm:EMtheta}. Applying
(\ref{eq:generic-decrease-theta-cond}) for $j=1,\dots,J$
with $\check{\varphi}=\normratiot$ defined by~(\ref{eq:respat}) and
$(\theta,\theta_0) = (\theta_{j,n+1},\theta_{j,n})$ yields the
condition~(\ref{eq:posMixtureP}) of \Cref{thm:EM:MixtureModel}.  
\end{rem}
Let $\mathrm{D}_{\tau',\tau}=\frac{\rmd\tau'}{\rmd\tau}$ denote the
Radon-Nikodym derivative of $\tau'$ w.r.t. $\tau$. We then have the following proposition (we defer the proofs of the results from \Cref{sec:monot-decr-cond-LM} to \Cref{sec:proof:mix:exp:prelim:res}).

\begin{prop}[Conditions for a monotonic decrease within the (E)LM family]\label{prop:argmax-linear-mixture-exponential}
  Let $k^{(1)}$ be an LM family as in \Cref{defi:LMEFfam}. Let
  $\check{\varphi}:\Yset\to\rset_+$ be a probability density function
  with respect to $\nu$. Let
  $\vartheta_0,\vartheta \in \tilde{\Tset}$ and
  $\tau_0,\tau\in\parammixproba$ such that
  \begin{align}
    \label{eq:update-cond-tau-general_mixture}
&   \int_{\Yset} \lr{\int_{\Zset}
               \varphi(\ell, y)\;
                \log
               \lr{\mathrm{D}_{\tau,\tau_0}(\ell)}\tau_0 (\rmd \ell)} \nu(\rmd y) \geq0\;,\\
    \label{eq:update-cond-theta-general_mixture}
&   \int_{\Yset} \lr{\int_{\Zset} 
               \varphi(\ell, y) 
                \log
               \lr{\frac{k^{(o)}(\ell \circ \changevar(\vartheta),y)}{k^{(o)}(\ell \circ \changevar(\vartheta_0),y)}}
               \tau_0 (\rmd \ell)}\nu(\rmd y) \geq0
  \end{align}
where we set
\begin{align}
  \label{eq:check-varphi-varphiL-carphiQ}
  &\varphi(\ell,y)=\check{\varphi}(y)\,\frac{k^{(o)}(\ell\circ\changevar(\vartheta_0),y)}{k^{(1)}(\theta_0,y)} \;.
\end{align}
Then (\ref{eq:generic-decrease-theta-cond}) holds with
$\theta=(\vartheta,\tau)$ and $\theta_0=(\vartheta_0,\tau_0)$ .
\end{prop}
One possible way to obtain~(\ref{eq:update-cond-tau-general_mixture})
is to observe that since the left-hand side of~(\ref{eq:update-cond-tau-general_mixture}) is zero for $\tau=\tau_0$, (\ref{eq:update-cond-tau-general_mixture}) is fulfilled by setting
\begin{align}
  \label{eq:update-cond-tau-general_mixture_argmax2}
  \tau=\argmax_{\tau'\in\parammixproba}
   \int_{\Yset} \lr{\int_{\Zset}
               \varphi(\ell, y)\;
                \log
               \lr{\mathrm{D}_{\tau',\tau_0}(\ell)}\tau_0 (\rmd \ell)} \nu(\rmd y)\;,
\end{align}
assuming that this argmax is well-defined. This will notably be done later in order to get the update formula~(\ref{eq:update-student-nut-param}) of \Cref{ex:student-mixture} (see the proof of \Cref{ex-thm:student-mixture} in \Cref{sec:proof:mix:exp}).

As for the updating of $\vartheta$, we can again adopt a maximisation
approach to derive a convenient $\vartheta$. This is the purpose of the following result, where we directly update $\xi =\changevar(\vartheta)$ since $\changevar$ is a known mapping in \Cref{defi:LMEFfam}.
\begin{coro}\label{cor:monot-decr-general-LM}
  Consider an LM family as in \Cref{defi:LMEFfam} and $\theta_0=(\vartheta_0,\tau_0)\in\tilde{\Tset}\times\parammixproba$. Set $\xi_0=\changevar(\vartheta_0)$, let
  $\check{\varphi}:\Yset\to\rset_+$ be a probability density function
  with respect to $\nu$ and define
  $\varphi:\Zset\times\Yset\to\rset_+$ by~(\ref{eq:check-varphi-varphiL-carphiQ}). Let
  $\vartheta \in \tilde{\Tset}$ be such that
  $\changevar(\vartheta)$ is a solution of the argmax
  problem
\begin{align}
\label{eq:update-cond-theta-general_mixture_argmax}
\argmax_{\xi\in\changevar(\tilde{\Tset})}    \int_{\Yset}\lr{\int_{\Zset}
               \lrb{\varphi(\ell,y)+ b_0(\ell)\,k^{(o)}(\ell(\xi_0),y)}
             \log
             \lr{\frac{k^{(o)}(\ell(\xi),y)}{k^{(o)}(\ell(\xi_0),y)}}
             \;\tau_0 (\rmd \ell)}\nu(\rmd y) \;,
\end{align}
where $b_0$ is any non-negative
function defined on $\Zset$ such that $\int b_0\;\rmd\tau_0<\infty$. 
Then (\ref{eq:update-cond-theta-general_mixture}) holds.  
\end{coro}
As we will see, (\ref{eq:update-cond-tau-general_mixture_argmax2}) and~(\ref{eq:update-cond-theta-general_mixture_argmax}) can be achieved for continuous mixtures such as the one used to define Student's $t$ distributions but they require a thorough analysis. To this end, we next formulate and solve a maximisation approach for a kernel $k^{(1)}$ belonging to the ELM family.

\subsubsection{The Maximisation Approach for the ELM Family}

Let us show that we can solve an argmax of the form~(\ref{eq:update-cond-theta-general_mixture_argmax}) for
an adequate choice of $b_0$ by fully exploiting the fact that
$k^{(o)}(\zeta,y)$ is a canonical kernel of the exponential family
distribution which satisfies~\ref{hyp:expo} and by relying on
additional assumptions to be introduced alongside some helpful
notation. 
Our first assumption is the following.
\begin{hyp}{C}
\item\label{asump:Upsilon-image} The image set $\changevar(\tilde{\Tset})$ is an open subset
  of $F$ and for all $\tau\in\parammixproba$,  $\tau$-a.e. $\ell$ is full rank.
\end{hyp}
\ref{asump:Upsilon-image} guarantees that, for all
$\vartheta\in\tilde{\Tset}$, $\tau\in\parammixproba$ and $\tau$-almost
all $\ell$ in $\Zset$, $\ell\circ\changevar(\vartheta)$ belongs to
$\intE_0$, the interior set of $E_0$ (see \Cref{lem:C1cond} in
\Cref{sec:proof:mix:exp} for details). This will come in handy in the
upcoming derivations and we now introduce some helpful notation,
before presenting our next two assumptions.  We let $\normop{\cdot}$
denote operator norms, for instance, for any $\ell\in \Zset$,
$$
\normop{\ell}=\sup\set{\normev[F]{\ell(x)}}{x\in E\,,\,\normev[E]{x}\leq1}\;.
$$
We further denote the adjoint of the linear operator $\ell$ by $\ell^T$, for instance, in the case $\ell\in \Zset$, $\ell^T\in\mathcal{L}(E,F)$ is 
  defined by
  \begin{align}\label{eq:defajoint}
  \pscal[E]{\ell(x)}{y}=  \pscal[F]{x}{\ell^T(y)}\qquad\text{for all $(x,y)\in
    F\times E$}
  \end{align}
and, for convenience, we also denote
\begin{align}
\label{eq:Q:gradA-def}
  \gradA{\ell}:=\ell^T\circ\nabla A\circ \ell\;.
\end{align}
Here, $A$ is the function appearing in the definition of
an exponential family (\Cref{def:expo-family}). Recall that $A$ is infinitely differentiable on $\intE_0$ so that by~\ref{asump:Upsilon-image},
$\gradA{\ell}$ is well-defined on $\changevar(\tilde{\Tset})$ for
$\tau$-a.e. $\ell$ with $\tau\in\parammixproba$.  Furthermore, we use
$\PP_{\vartheta,\tau}$ and $\PE_{\vartheta,\tau}$ to denote the
probability and its corresponding expectation under which the pair
$(Y,L)$ has density
$(y,\ell)\mapsto k^{(o)}(\ell\circ \changevar(\vartheta),y)$ with
respect to $\nu\otimes\tau$. Namely, for any non-negative measurable
$g$ on $\Yset\times\Zset$,
\begin{align}\label{eq:DefiKernelLinMixt-PE}
  \PE_{\vartheta,\tau}\lrb{g(Y,L)} = \int g(y,\ell)\, k^{(o)}(\ell\circ
  \changevar(\vartheta),y)\; \nu(\rmd y)\tau (\rmd \ell)\;.
\end{align}
We now introduce, for any $(\vartheta_0,\tau_0)\in\tilde{\Tset}\times\parammixproba$, two auxiliary functions 
$\tilde{\mathrm{m}}_{\vartheta_0,\tau_0}$ and
$\tilde{\mathrm{s}}_{\vartheta_0\tau_0}$, both defined from $\Yset$ to
$\rset_+$, such that the two following conditions hold.
  \begin{hyp}{C} 
\item\label{item:condB-sup-exp-moment} For all
  $\vartheta_0\in\tilde{\Tset}$ and  $\tau_0\in\parammixproba$, we have
  \begin{equation*}
\PE_{\vartheta_0,\tau_0}\argcond{\normop{L}}{Y} \leq  \tilde{\mathrm{m}}_{\vartheta_0,\tau_0}(Y)\quad\PP_{\vartheta_0,\tau_0}-\as
  \end{equation*}
  \item\label{item:condB-stilde-ratio} For all
  $\vartheta_0\in\tilde{\Tset}$ and  $\tau_0\in\parammixproba$, we have
  \begin{equation*}
    \forall\xi\in\changevar(\tilde{\Tset})\,,\,\exists\epsilon,C>0\,,\,
    \PE_{\vartheta_0,\tau_0}\argcond{\rme^{\epsilon\normev[F]{\gradA{L}(\xi)}}}{Y}\leq C\,\rme^{\tilde{\mathrm{s}}_{\vartheta_0,\tau_0}(Y)}
\quad\PP_{\vartheta_0,\tau_0}-\as
\end{equation*}
\end{hyp}
\ref{item:condB-sup-exp-moment}
and~\ref{item:condB-stilde-ratio} are made to ensure that the
integrals used to solve the argmax problem from
\Cref{cor:monot-decr-general-LM} will be well-defined. We finally present our last assumption, which corresponds to some sort
of identifiability assumption (see \Cref{lem:ident-elm} in
\Cref{sec:proof:mix:exp} for details). It is used to obtain the
uniqueness for the argmax problem we will solve. 
\begin{hyp}{C}
  \item \label{item:condB-ident} For all
    $\xi\neq\xi'\in\changevar(\tilde{\Tset})$ and  $\tau_0 \in\parammixproba$, we have
    $\tau_0 \lr{\set{\ell\in\Zset}{\ell(\xi)\neq \ell(\xi')}}>0$. 
\end{hyp}  
We can now state the following result, whose proof can be found in \Cref{sec:proof-crefthm:-objective-equiv}.
\begin{thm}\label{thm:argrmax-objective-equivalence}
  Consider an ELM family as in \Cref{defi:LMEFfam} and let $\theta_0 = (\vartheta_0, \tau_0) \in \tilde{\Tset} \times \parammixproba$. Assume
  \ref{asump:Upsilon-image}--\ref{item:condB-ident} hold.  Let $\check{\varphi}$ be a probability density function w.r.t. $\nu$ such that
  \begin{equation}
    \label{eq:cond-checkVarphi-linear-mixture}
    \int_\Yset \check{\varphi}(y)\,\lr{\normev[E]{S(y)}\,\tilde{\mathrm{m}}_{\vartheta_0,\tau_0}(y)
      +\rme^{\tilde{\mathrm{s}}_{\vartheta_0,\tau_0}(y)}}\;\nu(\rmd y) <\infty
    \;, 
  \end{equation}
  define $\varphi:\Zset\times\Yset\to\rset_+$
  by~(\ref{eq:check-varphi-varphiL-carphiQ}) and set for all $(\ell,y,\xi) \in\Zset \times \Yset \times \changevar(\tilde{\Tset})$
  such that $\ell\circ\changevar(\vartheta_0)\in\intE_0$
  \begin{align}
    \label{eq:check-WW-well-defined-carphiQ}
    &\tilde{\varphi}(\ell) =\int_{\Yset}\varphi(\ell,y)\;\nu(\rmd y)\;,\\
    \label{eq:gradA-WW-well-defined-carphiQ}
    & \gradA{\mathrm{w}}_{\tau_0}(\xi)=\int_{\Zset} \tilde{\varphi}(\ell)\,\gradA{\ell}(\xi)\;\tau_0 (\rmd \ell)\;.
  \end{align}
  Then $\tilde{\varphi}<\infty$ $\tau_0$-a.s., and
  $\gradA{\mathrm{w}}_{\tau_0}$ is well-defined from
  $\changevar(\tilde{\Tset})$ to $F$ and
  one-to-one on $\changevar(\tilde{\Tset})$, hence bijective from
  $\changevar(\tilde{\Tset})$ to its image
  $\gradA{\mathrm{w}}_{\tau_0}\circ\changevar(\tilde{\Tset})$, and
  satisfies, for all $\vartheta \in \tilde{\Tset}$,
    \begin{align}\label{eq:cond-gradA-w-cont-y-wise-expr-}
      \gradA{\mathrm{w}}_{\tau_0}\circ\changevar(\vartheta)
      = \PE_{\vartheta,\tau_0}\lrb{L^T\circ S(Y)\,\tilde{\varphi}(L)} \eqsp.
    \end{align}
  Moreover,
  setting $\xi_0=\changevar(\vartheta_0)$, for
  any $b\in\rset_+$, the argmax problem
\begin{equation}
  \label{eq:argrmax-objective:armax-prob-check-varphi-Q }
\argmax_{\xi\in\changevar(\tilde{\Tset})}\int_{\Yset} \lr{\int_{\Zset}
  \lrb{\varphi(\ell,y) + b\,\tilde{\varphi}(\ell)\,k^{(o)}(\ell(\xi_0),y)}\,
  \log
  \lr{\frac{k^{(o)}(\ell(\xi),y)}{k^{(o)}(\ell(\xi_0),y)
    }}\;\tau_0 (\rmd \ell)} \nu(\rmd y) \eqsp, 
\end{equation}
has at least a solution if and only if 
\begin{equation}
  \label{eq:wstar-defined-check-varphiQ}
\mathbf{w}^*=\frac1{1+b}\,\int \varphi(\ell,y)
\,\ell^T\circ S(y) \;\nu(\rmd 
y)\tau_0 (\rmd \ell)+\frac b{1+b}\,\gradA{\mathrm{w}}_{\tau_0}(\xi_0)
\end{equation}
belongs to $\gradA{\mathrm{w}}_{\tau_0}\circ\changevar(\tilde{\Tset})$, in
which case~(\ref{eq:argrmax-objective:armax-prob-check-varphi-Q }) has a
unique solution $\xi^*$ defined by
\begin{equation}
  \label{eq:wstar-defined-check-varphiQ-sol}
\gradA{\mathrm{w}}_{\tau_0}(\xi^*)=\mathbf{w}^*\;.
\end{equation}
\end{thm}
By setting $\xi^*=\changevar(\vartheta^*)$ and
  $\xi_0=\changevar(\vartheta_0)$ as well as using \eqref{eq:cond-gradA-w-cont-y-wise-expr-}, (\ref{eq:wstar-defined-check-varphiQ-sol}) can be
  interpreted as
  \begin{equation}
  \label{eq:wstar-defined-check-varphiQ-new-sol}
  \PE_{\vartheta^*,\tau_0}\lrb{L^T\circ S(Y)\,\tilde{\varphi}(L)} =   \PE_{\vartheta_0,\tilde{\tau}_0}\lrb{L^T\circ S(\tilde{Y})} \;,
\end{equation}
where $\tilde{\tau}_0$ denotes the probability having density
$\tilde{\varphi}$ with respect to $\tau_0$ and  $\tilde{Y}$ denotes a random
variable valued in $\Yset$ such that,  under $\PP_{\vartheta_0,\tilde{\tau}_0}$, the pair $(\tilde{Y},L)$ is distributed according to the
mixture density
$$
(y,\ell) \mapsto \psi(y,\ell)=\frac1{1+b}\, \varphi(\ell,y)+\frac b{1+b}\,k^{(o)}(\ell\circ
\changevar(\vartheta_0),y)\,\tilde{\varphi}(\ell)
$$
with respect to $\nu \otimes \tau_0$. Interestingly, the second marginals
of these two mixands are $\tilde{\tau}_0$. The distribution of
$(\tilde{Y},L)$ under $\PP_{\vartheta_0,\tilde{\tau}_0}$ can thus equivalently be defined by saying
that $L\sim\tilde{\tau}_0$
and the conditional distribution of $\tilde{Y}$ given $L=\ell$ is
the mixture with density
$$
y\mapsto \psi\argcond{y}{\ell}=\frac1{1+b}\, \frac{\varphi(\ell,y)}{\tilde{\varphi}(\ell)}+\frac b{1+b}\,k^{(o)}(\ell\circ
\changevar(\vartheta_0),y)\;.
$$
Using the definition of $\tilde{\tau}_0$ on the left-hand side
of~(\ref{eq:wstar-defined-check-varphiQ-new-sol}) as well, the latter reads
  \begin{equation}
  \label{eq:wstar-defined-check-varphiQ-new-sol-bis}
  \PE_{\vartheta^*,\tilde{\tau}_0}\lrb{L^T\circ S(Y)} =   \PE_{\vartheta_0,\tilde{\tau}_0}\lrb{L^T\circ S(\tilde{Y})} \;.
\end{equation}
Hence, \Cref{thm:argrmax-objective-equivalence} along with
\eqref{eq:wstar-defined-check-varphiQ-new-sol-bis} can be used to
solve the argmax
problem~(\ref{eq:update-cond-theta-general_mixture_argmax}) in the
case where  \ref{hyp:expo} and
\ref{asump:Upsilon-image}--\ref{item:condB-ident} hold and for $b_0(\ell)=b\,\tilde{\varphi}(\ell)$ with $b$ being a non-negative
constant. From there, it only remains to find ways to check that the condition (\ref{eq:cond-checkVarphi-linear-mixture}) on $\check{\varphi}$ holds. 
This can notably be done when $\check{\varphi}$ takes one of the two
forms listed in \Cref{rem:generic-cond-elm} (see
\Cref{lem:cond-p-varphi} in \Cref{sec:argm-solut-param-expo:proofs}
for details) and we now provide an example of particular interest.

\begin{ex}[Finite mixture of Student's $t$ distributions]
\label{ex:student-mixture}  Let $\nu$ be the Lebesgue measure on $\Yset=\rset^d$, with $d\geq1$.
  Let $\alpha\in[0,1)$ and $p:\rset^d\to\rset_+$ such that
  \begin{equation}
    \label{eq:cond-student-on-p}
0<    \int_{\rset^d} (1+\|y\|^2)^{1/(1-\alpha)}\,p(y)\;\rmd y < \infty \;.
  \end{equation}
  Let $J\geq1$ and consider a  $J$-mixture family $\mathcal{Q}$ as
  in~(\ref{eq:def:MM}), with $k(\theta,\cdot)$ being the density
  of the Student's $t$ distribution with parameter
  $\theta=(\mt,\Sigma,\nut)\in\Tset=\rset^d\times\mathcal{M}_{>0}(d)\times(0,\infty)$,
  given by
  $$
  k(\theta,y)=\frac{\Gamma\lr{(\nut+d)/2}}{\Gamma\lr{\nut/2}(\nut\pi)^{d/2}\lrav{\Sigma}^{1/2}}\,
  \lr{1+\frac1{\nut}(y-\mt)^T\Sigma^{-1}(y-\mt)}^{-(\nut+d)/2}\;.
  $$
  Define a sequence $(\mu_nk)_{n\geq1}$ valued in $\mathcal{Q}$
  where $\mu_n=\mu_{\lbd{n},\Theta_n}$ with $\lbd{n}\in\simplex_J^+$ and
  $\Theta_n=(\mt_{j,n},\Sigma_{j,n},a_{j,n})_{j=1 \dots J}\in\Tset^J$
  such that, for all $n\geq1$,
  \begin{enumerateList}
  \item\label{item:update_weights_student} given $(\eta_n)_{n \geq 1}$ valued in $(0,1]$,
    $(\cte_n)_{n \geq 1}$ valued in $(-\infty,0]$ and
    $\lbd{1}\in\simplex_J^+$, $(\lbd{n})_{n\geq2}$ is defined by  the iterative
    formula~(\ref{eq:aPMC:updateOne});
  \item\label{item:update_param_student} given $(\gamma_{j,n})_{j=1,\dots,J,n\geq1}$ valued in $(0,1]$ and
    $\Theta_{1}\in\Tset^J$, $(\Theta_n)_{n\geq2}$ is defined by 
    \begin{align}
        \label{eq:update-student-m-param}
     & \mt_{j,n+1}= \lr{\int_{\rset^d\times\rset_+}z\,\psi_{j,n}(z,y)\;\rmd
       y\,\check{\tau}_{\nut_{j,n}}(\rmd z)}^{-1}\,\int_{\rset^d\times\rset_+}y\,z\,\psi_{j,n}(z,y)\;\rmd
       y\,\check{\tau}_{\nut_{j,n}}(\rmd z)\;,\\
              \label{eq:update-student-Sigma-param}
      &\Sigma_{j,n+1}=\int_{\rset^d\times\rset_+}(y-\mt_{j,n+1})(y-\mt_{j,n+1})^T\,z\,\psi_{j,n}(z,y)\;\rmd 
        y\,\check{\tau}_{\nut_{j,n}}(\rmd z)\;,\\
        \label{eq:update-student-nut-param}
     & \nut_{j,n+1} = 2\, \kappa^{-1}\lr{\int_{\rset^d\times\rset_+} (z-\ln(z))\,\varphi_{j,n}(z,y)\;\rmd 
       y\,\check{\tau}_{\nut_{j,n}}(\rmd z)}\;,
    \end{align}
where, defining $\normratiot$ as in~(\ref{eq:respat}) and denoting $\kappa^{-1}:\rset\to(0,\infty)$ the inverse
mapping of  $\kappa(x)=\log(x)+\Gamma'(x)/\Gamma(x)$, we set
\begin{align}
  \label{eq:update-student-varphi-def}
     & \varphi_{j,n}(z,y)=\normratiot(y)\frac{\rme^{\frac{-z}{2}(y-\mt_{j,n})^T\Sigma_{j,n}^{-1}(y-\mt_{j,n})}}{(2\pi/z)^{d/2}\lrav{\Sigma_{j,n}}^{1/2}k(\theta_{j,n},y)}\;,\\
  \label{eq:update-student-psi-def}
     & \psi_{j,n}(z,y)=
       \gamma_{j,n}
       \,\varphi_{j,n}(z,y)
       +(1-\gamma_{j,n})\,\,\frac{\rme^{\frac{-z}{2}(y-\mt_{j,n})^T\Sigma_{j,n}^{-1}(y-\mt_{j,n})}}{(2\pi/z)^{d/2}\lrav{\Sigma_{j,n}}^{1/2}}\,\lr{\int\varphi_{j,n}(z,y)\;\rmd y}
       \;,\\
      \label{eq:update-student-tau-def}
      &\check{\tau}_{\nut}(\rmd z)=\mathbbm{1}_{\rset_+}(z)\,
        \frac{(\nut/2)^{\nut/2}}{\Gamma(\nut/2)}z^{\nut/2-1}\rme^{-z\nut}\;\rmd
        z\quad\text{for any $\nut>0$}\;,
\end{align}
  \end{enumerateList}
 Here, $\check{\tau}_{\nut}$
    in~(\ref{eq:update-student-tau-def}) is the $\chi^2$ distribution with
$\nut$ degrees of freedom and $\kappa^{-1}$
in~(\ref{eq:update-student-nut-param}) is well-defined (this follows from \Cref{lem:kappa-def-invertible} in \Cref{sec:elliptical-student-mixture-proof}).
\end{ex}
We then have the following result, whose proof is postponed to \Cref{sec:elliptical-student-mixture-proof}. 
\begin{coro}\label{ex-thm:student-mixture}
  In the setting of \Cref{ex:student-mixture}, if $\Psif(\mu_1 k) <
  \infty$, then, at any
  time $n \geq 1$, we have $\Psif(\mu_{n+1} k) \leq \Psif(\mu_n
  k)$.   
\end{coro}
\kdtxt{
We have stated results that solve our maximisation and gradient-based approaches when the variational family is based on the exponential family. However, all the update formulas we have obtained so far involve intractable integrals. In the next section, we focus on the case of Gaussian Mixture Models optimisation seen in Examples \ref{ex:GMMmax} and \ref{ex:GMM} and we investigate how our framework can be used to derive tractable algorithms in that case.

\section{Gaussian Mixture Models (GMMs) Optimisation}
\label{sec:practicalities}

We consider $d$-dimensional Gaussian mixture densities with $k(\theta_j, y) = \mathcal{N}(y; m_j, \Sigma_j)$, where $\theta_j = (m_j, \Sigma_j) \in \Tset$ denotes the mean and covariance matrix of the $j$-th Gaussian component density. By \Cref{thm:EM:MixtureModel}, a valid update formula at time $n \geq 1$ for $(\lbd{n})_{n \geq 1}$ is
\begin{align}
\lambda_{j,n+1} &=  \frac{\lambda_{j,n} \lrb{\int_\Yset \respat[y] \nu(\rmd y) + (\alpha-1) \cte_n}^{\eta_n}}{\sum_{\ell=1}^J \lambda_{\ell,n} \lrb{\int_\Yset \respat[y][\alpha][\ell] \nu(\rmd y) + (\alpha-1) \cte_n}^{\eta_n}} \eqsp, \quad j = 1 \ldots J \label{eq:aPMC:updateInSec} \eqsp.
\end{align}
Furthermore, 
as underlined in \Cref{ex:GMMmax}, a valid update at time $n \geq 1$ for the means $(m_{j,n+1})_{1 \leq j \leq J}$ and the covariances matrices $(\Sigma_{j,n+1})_{1 \leq j \leq J}$ is given by
\begin{align}
  \forall j = 1 \ldots J, ~ m_{j,n+1}&=  (1 -\gamma_{j,n}) m_{j,n} + \gamma_{j,n}  \widehat{m}_{j,n} \label{eq:updateMeanGaussianTwo}\\
  \Sigma_{j,n+1} &=  (1 -\gamma_{j,n}) \Sigma_{j,n} + \gamma_{j,n} \widehat{\Sigma}_{j,n} + \gamma_{j,n} (1 -\gamma_{j,n})
  \lr{\widehat{m}_{j,n}-m_{j,n}} \lr{\widehat{m}_{j,n}-m_{j,n}}^T \nonumber 
  \end{align}
  where $\widehat{m}_{j,n} =\int_\Yset y\, \normratiot (y)\;\nu(\rmd y)$, $\widehat{\Sigma}_{j,n}= \int_\Yset yy^T\, \normratiot
          (y)\;\nu(\rmd y)
          -\widehat{m}_{j,n}\widehat{m}_{j,n}^T$ with $\normratiot = \ratiogen / \int \ratiogen \nu(\rmd y)$ and $\gamma_{j,n} \in (0,1]$. In addition, another possible update for the means is
\begin{align}
  m_{j, n+1} &= m_{j,n} + \gamma_{j,n} \frac{\int_\Yset \lambda_{j,n} \respat[y] (y-m_{j,n}) \nu(\rmd y)}{\int_\Yset {\mu_n k(y)^\alpha p(y)^{1-\alpha} } \nu(\rmd y)} \eqsp, \quad j = 1 \ldots J \eqsp. \label{eq:updateMeanGaussianOne}
\end{align}
By virtue of \Cref{ex:GMM}, this update can be used in lieu of the means update written in \eqref{eq:updateMeanGaussianTwo} with $\gamma_{j,n}$ still in $(0,1]$. Since the update \eqref{eq:updateMeanGaussianOne} is linked to Gradient Descent steps for R\'{e}nyi's $\alpha$-divergence minimisation (see \Cref{ex:GMM}), we will refer to this update as the \textit{R\'{e}nyi Gradient Descent} (RGD) update in the following. One the other hand, the update on the means written in \eqref{eq:updateMeanGaussianTwo} will be referred to as the \textit{Maximisation Gradient} (MG) update.

Notice that the updates \eqref{eq:aPMC:updateInSec}, \eqref{eq:updateMeanGaussianTwo} and \eqref{eq:updateMeanGaussianOne} all involve intractable integrals. However, we can resort to approximate update rules and take advantage of the fact that the integrals appearing in these update formulas are of the form 
\begin{align}\label{eq:updaterule:to:approx}
\int_\Yset \respat[y] g(y) \nu(\rmd y) \eqsp,
\end{align}
where $g$ is a well-chosen function defined on $\Yset$ and where $\varphi_{j, n}^{(\alpha)}$ is defined in \eqref{eq:respat} by
$
  \respat[y] = k(\theta_{j,n}, y) \lr{{\mu_n k(y)}/{  p(y)}}^{\alpha-1}$ for all $y \in \Yset$. Many choices are indeed possible to approximate \eqref{eq:updaterule:to:approx} and for simplicity we restrict ourselves to typical Importance Sampling estimation. Looking at the definition of $\varphi_{j, n}^{(\alpha)}$, a first idea would be to sample $M$ i.i.d. random variables $(Y_{m, n}^{(j)})_{1 \leq m \leq M}$ according to $k(\theta_{j,n}, \cdot)$ for $j = 1 \ldots J$ and to use the unbiased estimator of \eqref{eq:updaterule:to:approx}
\begin{align*}
  \frac{1}{M} \sum_{m=1}^M \left(\frac{\mu_n k(Y_{m,n}^{(j)})}{p(Y_{m,n}^{(j)})} \right)^{\alpha-1} g(Y_{m,n}^{(j)}) \eqsp.
\end{align*}
As this sampler depends on $j$, this option becomes very computationally heavy as $J$ and $M$ increase due to the \textit{sampling cost} ($J \times M$) and to the \textit{evaluation cost} ($\mu_n k$ and $p$ need to be evaluated $M\times J$ times each, where $\mu_n k$ is a sum over $J$). To reduce this computational bottleneck, it is then preferable to consider sequences of samplers $(q_n)_{n \geq n}$ that do not depend on $j$ at time $n$. More specifically, we discuss two possibilities.

\begin{enumerateList}
  \item \label{sampler:n} \textbf{Best sampler at time $n$ (IS-n).} Approximated update rules can be obtained by sampling $M$ i.i.d. random variables $(Y_{m,n})_{1\leq m\leq M}$ according to the best approximation of $p$ we have a time $n$, that is $q_n = \mu_n k$. Then, the samples are shared throughout the $J$ mixture weights and mixture components parameters updates. This eases the computational burden, with both a sampling and an evaluation costs of $M$, as an unbiased estimator of \eqref{eq:updaterule:to:approx} is
  \begin{align} \label{eq:unbiasedIS}
    \frac{1}{M} \sum_{m=1}^M \respa[Y_{m,n}] g(Y_{m,n}) ~~ \mbox{where} ~~ \respa[y] = \frac{k(\theta_{j,n}, y)}{q_n(y)} \lr{\frac{\mu_n k(y)}{p(y)}}^{\alpha-1}, ~ y \in \Yset \eqsp.
  \end{align}

  \item \label{sampler:unif} \textbf{Uniform sampler (IS-unif).} Approximated update rules can be obtained by sampling according to the uniform sampler $q_n = J^{-1} \sum_{j = 1}^J k(\theta_{j,n}, \cdot)$ and using the unbiased estimator \eqref{eq:unbiasedIS}. That way, the sampling and evaluation costs are also $M$ each. Contrary to IS-n, this sampler ensures a fair sampling among all components, as it entails sampling an index vector $[i_1, i_2, \dots, i_M]$ uniformly from $\{1, \dots, J\}$ whereas IS-n samples an index vector $[i_1, i_2, \dots, i_M]$ according to the mixture weights $\lbd{n}$.
\end{enumerateList}

\begin{rem} Approximated update rules for GMMs can be obtained by sampling according to $q_n = \mathcal{N}(\cdot; \boldsymbol{0}_d, \boldsymbol{I}_d)$, where $\boldsymbol{0}_d$ is the null vector of dimension $d$. These can be deduced by applying the reparameterisation trick: for all $m = 1 \ldots M$, $Y_{m,n}^{(j)} \sim k(\theta_{j,n}, \cdot)$ if $Y_{m,n}^{(j)} = m_{j,n} + \Sigma_{j,n}^{-1/2} \varepsilon_{m,n}$ with $\varepsilon_{m,n} \sim q_n$ and by observing that an unbiased estimator of \eqref{eq:updaterule:to:approx} is
  $
    M^{-1} \sum_{m=1}^M \lr{{\mu_n k(Y_{m,n}^{(j)})} / {p(Y_{m,n}^{(j)})}}^{\alpha-1} g(Y_{m,n}^{(j)})$.
This choice of $q_n$ relies on the existence of a reparameterisation and while it incurs a gain in terms of sampling cost as we only need $M$ samples, the evaluation cost remains $M \times J$.
\end{rem}
We thus have access to tractable algorithms that iteratively update both the weights and components parameters of a GMM by optimising the $\alpha$-divergence between the mixture distribution and the targeted distribution. Our framework also allows the use of learning rates $\gamma_{j,n}$ that are dependent on $j$. This means that the learning rates can differ for each component parameters update and it paves the way for adaptive learning rates per components. This aspect is left for future work as we will focus on the case $\gamma_{j,n} \eqdef \gamma_n$ according to Algorithm \ref{algo:GMMsMAIS} and move on to presenting our numerical experiments.

{\SetAlgoNoLine
\SetInd{0.8em}{-1.4em}
\begin{algorithm}[t]
\caption{GMMs optimisation with the IS-n/IS-unif sampler}
\label{algo:GMMsMAIS}
\textbf{At iteration $n$,}
\begin{enumerate}[wide=0pt, labelindent=\parindent]
\item Draw independently $M$ samples $(Y_{m, n})_{1 \leq m \leq M}$ from the proposal
 $q_n$. 

\item For all $j = 1 \ldots J$, 

\begin{enumerate}
  \item Set $\widehat{m}_{j,n} = \frac{\sum_{m=1}^M \respa[Y_{m,n}] \cdot Y_{m,n} }{ \sum_{m=1}^M  \respa[Y_{m,n}]}$ and $\widehat{\Sigma}_{j,n}=  \frac{\sum_{m=1}^M \respa[Y_{m,n}] \cdot  Y_{m,n}Y_{m,n}^T }{ \sum_{m=1}^M \respa[Y_{m,n}] } -\widehat{m}_{j,n}\widehat{m}_{j,n}^T$.
  \item Choose one between the (MG) and (RGD) approaches
  \begin{align*}
    (MG) \quad m_{j, n+1} &= \lr{1 - \gamma_{n}} m_{j,n} + \gamma_{n} \widehat{m}_{j,n} \\
     (RGD) \quad m_{j, n+1} &= m_{j,n} + \gamma_{n}  \frac{\lambda_{j,n} \sum_{m=1}^M \respa[Y_{m,n}] \cdot (Y_{m,n}- m_{j,n}) }{ \sum_{j = 1}^J \sum_{m=1}^M \lambda_{j,n} \respa[Y_{m,n}]}
  \end{align*}
  and set
  \begin{align*}
    \lambda_{j,n+1} &= \frac{\lambda_{j,n} \lrb{\sum_{m=1}^M \respa[Y_{m, n}]  + (\alpha-1)\cte_n}^{\eta_n}}{\sum_{\ell = 1}^J \lambda_{\ell,n} \lrb{\sum_{m=1}^M \respa[Y_{m, n}][\ell]  + (\alpha-1) \cte_n}^{\eta_n}} \\
    \Sigma_{j,n+1} &=  (1 -\gamma_{n}) \Sigma_{j,n} + \gamma_{n} \widehat{\Sigma}_{j,n} + \gamma_{n} (1 -\gamma_{n})
    \lr{\widehat{m}_{j,n}-m_{j,n}} \lr{\widehat{m}_{j,n}-m_{j,n}}^T. 
  \end{align*}
\end{enumerate}
\end{enumerate} \

\end{algorithm}}


\begin{rem} \label{rem:real:data} Computing the estimator \eqref{eq:unbiasedIS} can be expensive for two reasons: (i) $\mu_n k$ involves a sum over $J$ so that evaluating this function requires heavier computations as the number of components $J$ grows and (ii) bayesian applications with $p = p(\cdot, \data)$ often consider large-scale data sets $\data$. To alleviate this computational burden, one can (i) approximate the summation appearing in $\mu_n k$ using subsampling and (ii) use a mini-batch approach to approximate $p = p(\cdot, \data)$, as detailed in \cite[Section 4.3]{2016arXiv160202311L}.
\end{rem}


}

\section{Numerical Experiments}
\label{subsec:Exp}

In our numerical experiments, we are interested in understanding the behaviour of Algorithm~\ref{algo:GMMsMAIS} in practice. We first investigate the challenging case of multimodal targets, as our framework is designed to handle multimodality thanks to the hyperparameter $\alpha \in [0,1)$ and to the use of mixture models as a variational family.

\subsection{Toy multimodal targets}

Let $\boldsymbol{u_d}$ be the $d$-dimensional vector whose coordinates are all equal to $1$, $\boldsymbol{I_d}$ be the identity matrix and $c$ be a positive constant (we set $c = 2$). We consider three multimodal examples:
\begin{enumerateList}
  \item \label{itemExEWGMM} \textit{Equally-weighted Gaussian Mixture Model.} The target $p$ is a mixture density of two equally-weighted $d$-dimensional Gaussian distributions multiplied by $c$ such that
  $$
  p(y) = c \times \left[ 0.5 \mathcal{N}({y}; -2 \boldsymbol{u_d}, \boldsymbol{I_d}) + 0.5 \mathcal{N}({y}; 2 \boldsymbol{u_d}, \boldsymbol{I_d}) \right] \eqsp.
  $$
  
  \item \label{itemExUWGMM} \textit{Imbalanced Gaussian Mixture Model.} The target $p$ is a mixture density of three $d$-dimensional Gaussian distributions with unequal weights and multiplied by $c$ such that
  $$
  p(y) = c \times \lrb{ 0.35 \mathcal{N}({y}; -2 \boldsymbol{u_d}, \boldsymbol{I_d}) + 0.25 \mathcal{N}({y}; 2 \boldsymbol{u_d}, \boldsymbol{I_d})  + 0.4 \mathcal{N}({y}; \boldsymbol{u_d}, \boldsymbol{I_d}) }
  $$
  
  \item \label{itemExEWSMM} \textit{Equally-weighted Student's $t$ Mixture Model.}  The target $p$ is a mixture density of two equally-weighted $d$-dimensional Student's $t$ distributions with two degree of freedom (i.e. $\nut = 2$) multiplied by $c$ such that
  $$
  p(y) = c \times \lrb{ 0.5 \; \mathrm{t}({y}; -2 \boldsymbol{u_d},
    \boldsymbol{I_d}, \nut) + 0.5 \; \mathrm{t}({y}; 2 \boldsymbol{u_d}, \boldsymbol{I_d}, \nut) } \eqsp.
  $$
\end{enumerateList}


\subsubsection{Comparing the RGD and the MG Approaches with $\eta_n = 0$} 

\label{subsecEtaZero}

Setting $\eta_n =0$ at all times $n \geq 1$ keeps the mixture weights fixed in Algorithm~\ref{algo:GMMsMAIS}. In that case, our RGD approach relates to Gradient Ascent steps on the VR Bound \citep{2016arXiv160202311L} w.r.t. the means of our Gaussian components (see  \Cref{relwork:LiTurner}). We then want to compare the performances of the RGD approach \citep[that can be derived from][]{2016arXiv160202311L} to the novel MG approach we have introduced in our work. 

\textit{Implementation details.} The covariance matrices of the mixture components are fixed and equal to $\sigma^2 \boldsymbol{I_d}$ with $\sigma^2 = 1$, $\alpha = 0.2$, $J \in \lrcb{10, 50}$, $d = 16$, $M = 200$, the total number of iterations $N$ is equal to $100$, $\Theta_1$ is generated by sampling from a centered normal distribution with covariance matrix $10 \boldsymbol{I_d}$, $\lbd{1} = [1/J,  \ldots , 1/J]$ and for all time $n = 1 \ldots N$, $\cte_n = 0$, $\eta_n = 0.$ and $\gamma_n = \gamma$ with $\gamma \in \lrcb{0.1, 0.5, 1.}$. 
The experiments are replicated independently 30 times and the convergence of the RGD and of the MG approaches is monitored for the three multimodal examples \ref{itemExEWGMM}, \ref{itemExUWGMM} and \ref{itemExEWSMM} by computing a Monte Carlo estimator of the VR Bound (since we have sampled $M$ samples from $q_n$ at time $n$, these samples can readily be used to obtain an estimate the VR Bound with no additional computational cost).

Our results are plotted on \Cref{fig:eta0}, where RGD-IS-n($\gamma$) and MG-IS-n($\gamma$) denote the algorithms originating from the RGD and MG approaches (the IS-n and IS-unif samplers are equivalent here), and error bounds plots can be found in Figures \ref{fig:eta0appEBMGISn} and \ref{fig:eta0appEBRGDISn} of \Cref{sec:addNumExp}. As minimising $\Psif$ is equivalent to maximising the VR Bound when $\alpha \in (0,1)$, our theoretical results predict a systematic increase in the VR Bound for the algorithms considered. This is what we observe in \Cref{fig:eta0} and we also see that the choice of $\gamma$ plays the role of a learning rate our algorithms: while increasing $\gamma$ mostly improves the speed of convergence, it may deteriorate the performances if chosen too large, such as when $J = 50$ with the target \ref{itemExEWSMM}.

\begin{figure}[t]
  \centering
  \begin{tabular}{ccc}
    & $J = 10$ & $J = 50$ \\
    \ref{itemExEWGMM} \vspace{-0.5cm} & \\
    & \includegraphics[width=7cm]{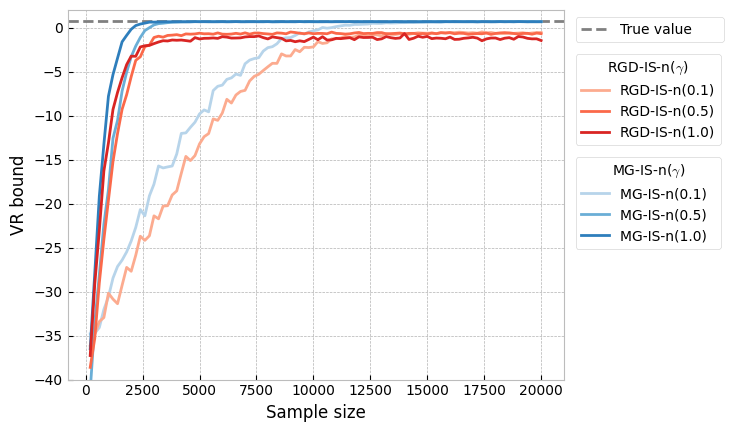} & \includegraphics[width=7cm]{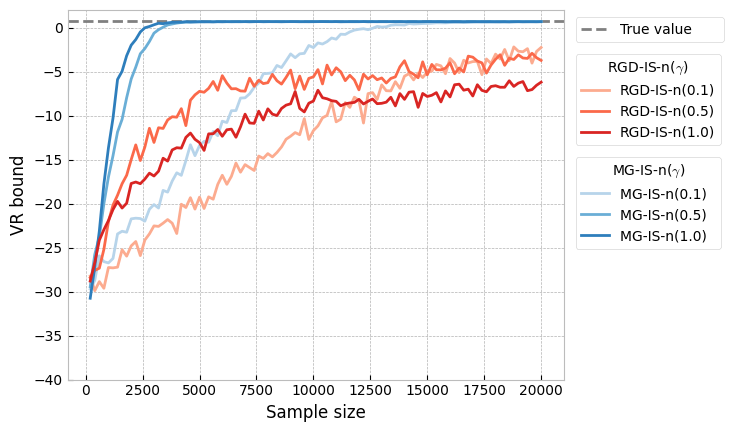} \\
    \ref{itemExUWGMM} \vspace{-0.5cm} & \\ 
    &\includegraphics[width=7cm]{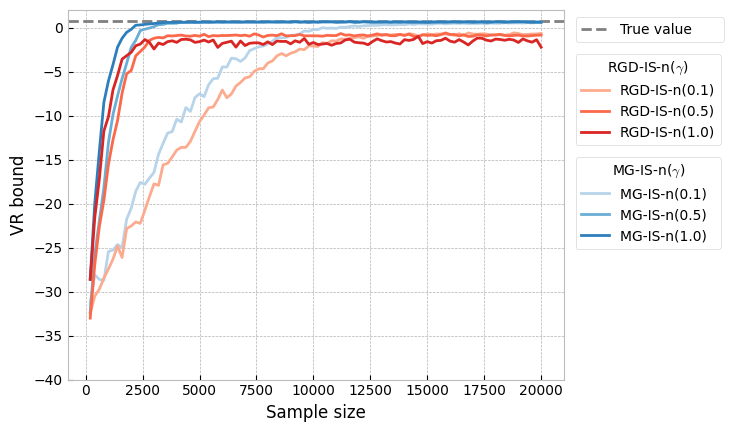} & \includegraphics[width=7cm]{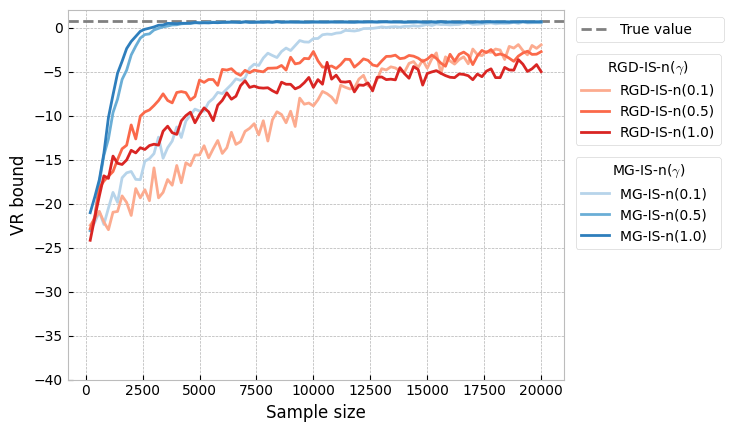} \\
    \ref{itemExEWSMM} \vspace{-0.5cm} & \\
    &\includegraphics[width=7cm]{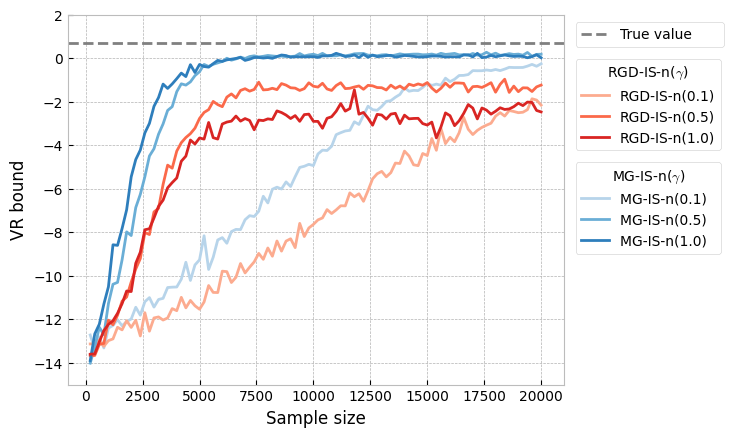} & \includegraphics[width=7cm]{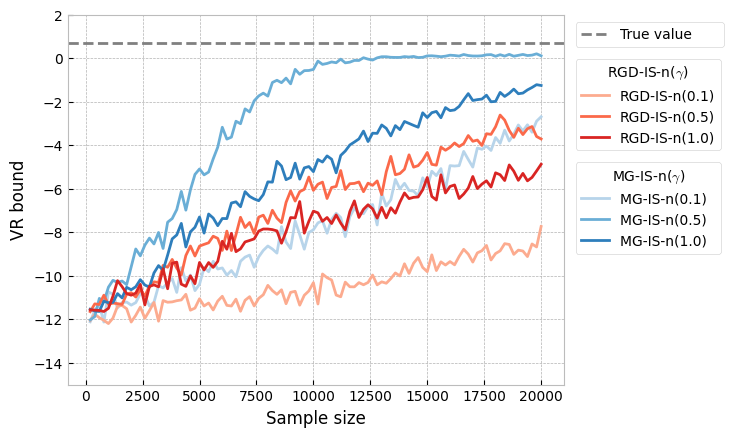} \\
    \end{tabular}

  \caption{Monte Carlo estimate of the VR Bound for the RGD and the MG approaches (fixed mixture weights) when considering each of the target distributions \ref{itemExEWGMM}, \ref{itemExUWGMM} and \ref{itemExEWSMM}.} 
  \label{fig:eta0}
\end{figure}


\begin{table}
  \centering
\begin{tabular}{clccccccc}
      \toprule
    & & & $J = 10$ & & & $J = 50$ & \\  
    & & $\gamma = 0.1$ &  $\gamma = 0.5$ &  $\gamma = 1.0$ & $\gamma = 0.1$ & $\gamma = 0.5$ & $\gamma = 1.0$  \\
       \midrule
    \ref{itemExEWGMM} & RGD-IS-n($\gamma$) & -0.081 & -0.076 & -0.218 & -1.640 & -1.673 & -1.560   \\
    & MG-IS-n($\gamma$) & \textbf{-3.702} & \textbf{-1.875} & \textbf{-2.711} & \textbf{-2.760} & \textbf{-2.771} & \textbf{-2.788}  \\
    \midrule
    \ref{itemExUWGMM} & RGD-IS-n($\gamma$) & -0.211 & -0.072 & -0.015 & -1.401 & -1.437 & -1.515  \\
    & MG-IS-n($\gamma$) & \textbf{-2.581} & \textbf{-2.101} & \textbf{-1.742} & \textbf{-2.611} & \textbf{-2.328} & \textbf{-1.933} \\
    \midrule
    \ref{itemExEWSMM} & RGD-IS-n($\gamma$) & -0.108 & -0.008 & -0.111 & -1.652 & -1.654 & \textbf{-1.634} \\
    & MG-IS-n($\gamma$) &  \textbf{-0.913} & \textbf{-1.489} & \textbf{-1.846} & \textbf{-2.036} & \textbf{-2.530} & {-0.717} \\
       \bottomrule
\end{tabular}
\caption{LogMSE averaged for the RGD and the MG approaches (fixed mixture weights) when considering each of the target distributions \ref{itemExEWGMM}, \ref{itemExUWGMM} and \ref{itemExEWSMM}.}   \label{table:logMSE} 
\end{table}

Furthermore, MG-IS-n($\gamma$) leads to better performances in all but one case compared to RDG-IS-n($\gamma$). This renders the MG-IS-n($\gamma$) algorithm more suitable overall for capturing the multimodality of the various multimodal targets $p$ we considered, a fact that is further supported in \Cref{table:logMSE}, in which we compare the log of the Mean-Squared Error (logMSE) returned by each algorithm. (The MSE is computed as the average of $\|m_{\mathrm{approx}} - m_{\mathrm{true}}\|^2$ over 30 independent runs of each algorithm, where $\| .\|$ stands for the Euclidean norm, $m_{\mathrm{true}} = \PE_\PP[Y]$ with $\PP(A) = \int_A p(y) \nu(\rmd y)/\int_\Yset p(y) \nu(\rmd y)$ for all $A \in \Ysigma$ and where $m_{\mathrm{approx}}$ is the mean of the optimised variational distribution, that is in our setting $m_{\mathrm{approx}} = \sum_{j=1}^{J} \lambda_{j,N} m_{j,N}$.)

Let us next pair up the means optimisation with mixture weights optimisation and investigate the impact of the sequence $(\eta_n)_{n \geq 1}$ on our algorithms.

\subsubsection{Comparing the RGD and the MG Approaches with $\eta_n > 0.$}
\label{subsec:expetachange}

So far, we have demonstrated the viability of our novel MG approach compared to the more usual R\'{e}nyi's $\alpha$-divergence-based RGD approach. Yet, we have not taken advantage of the fact that we can perform mixture weights optimisation on top of means optimisation, which is another novelty of our framework compared to traditional Variational Inference methods. 

\textit{Implementation details.} We work with the same implementation details as those described in \Cref{subsecEtaZero}, except for the fact that we will now let $\eta_n = \eta$ at time $n$ and we will vary the value of $\eta$. In addition to the RGD-IS-n$(\gamma)$ and the MG-IS-n$(\gamma)$ algorithms, we also include the RGD-IS-unif$(\gamma)$ and the MG-IS-unif$(\gamma)$ algorithms in our results (that is, the algorithms resulting from pairing up the RGD and MG approaches with the IS-unif sampler).

The plots for the averaged Monte Carlo estimate of the VR Bound with $\eta = 0.1$ and $\gamma = 0.5$ are then available in \Cref{fig:eta0dot1} (and similar plots for $\eta \in \lrcb{0.05, 0.5}$ are provided in Figures \ref{fig:eta0dot05} and \ref{fig:eta0dot5} of \Cref{sec:addNumExp} respectively). According to \Cref{fig:eta0dot1}, MG-IS-unif$(0.5)$ outperforms most of the time the other methods (this trend for MG-IS-unif$(\gamma)$ is further confirmed by looking at the logMSE in \Cref{table:logMSEEtaO1} of \Cref{sec:addNumExp}, which includes the cases $\gamma \in \lrcb{0.1, 1.}$ for completeness). 

We also plot in \Cref{fig:weights} the final mixture weights $\lbd{N}$ obtained after one typical run of MG-IS-unif($0.5$) for $\eta \in \lrcb{0.05, 0.1, 0.5}$ and $J \in \lrcb{10, 50}$ (and additional LogMSE results can be found  in \Cref{table:logMSEVaryingEta} of \Cref{sec:addNumExp}). This is done in order to delve into how $\eta$ impacts the variational density returned at the end of Algorithm \ref{algo:GMMsMAIS} (MG-IS-unif$(0.5)$ was chosen among all four options since it enjoys good empirical performances according to \Cref{fig:eta0dot1}). \newline

\begin{figure}[t]
  \centering
  \begin{tabular}{ccc}
    & $J = 10$ & $J = 50$ \\
    \ref{itemExEWGMM} \vspace{-0.5cm} & \\
    & \includegraphics[width=7cm]{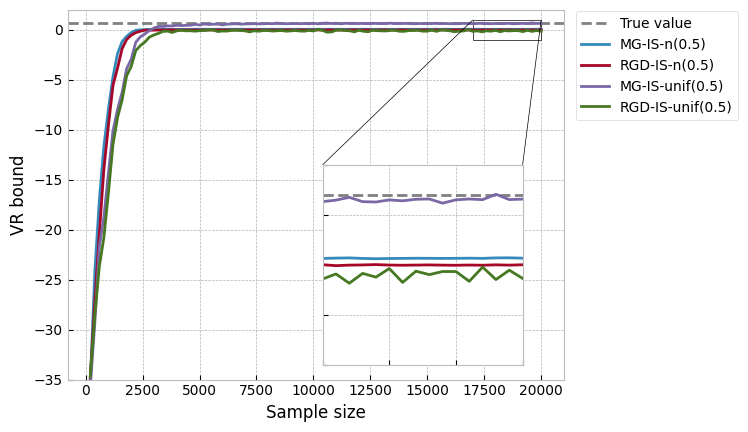} & \includegraphics[width=7cm]{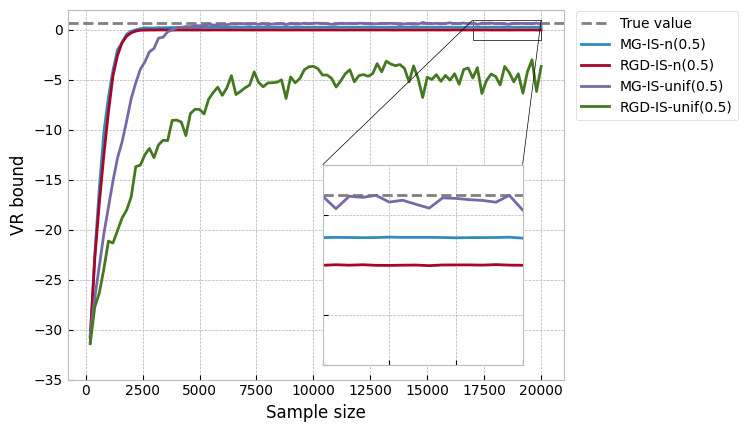} \\
    \ref{itemExUWGMM} \vspace{-0.5cm} & \\
    &\includegraphics[width=7cm]{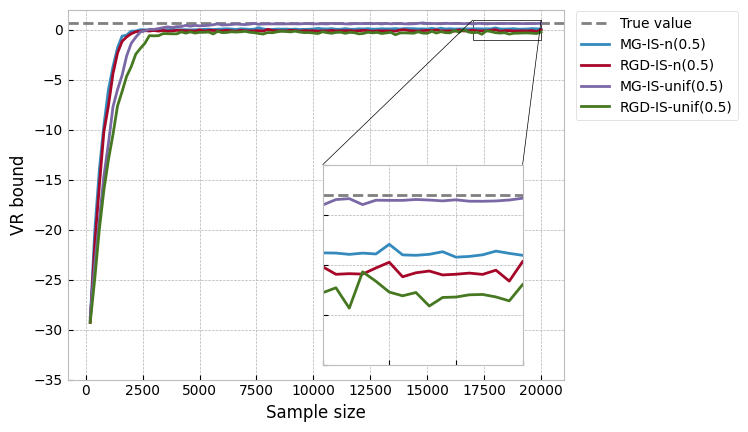} & \includegraphics[width=7cm]{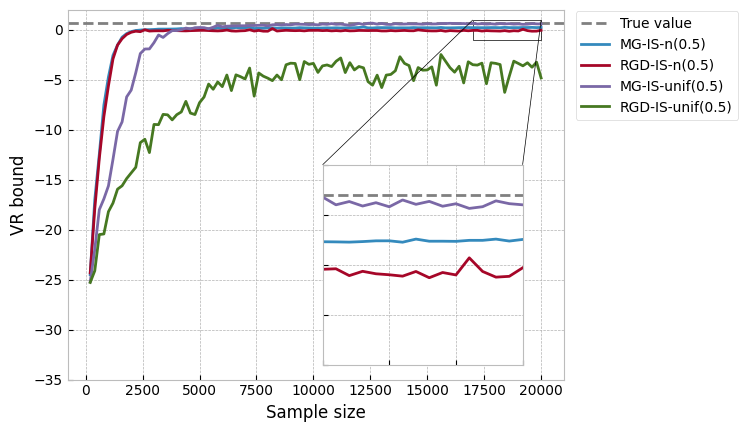} \\
    \ref{itemExEWSMM} \vspace{-0.5cm} & \\
    &\includegraphics[width=7cm]{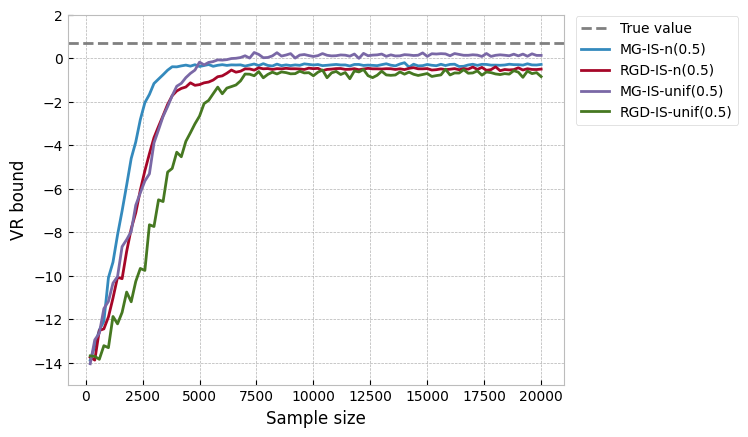} & \includegraphics[width=7cm]{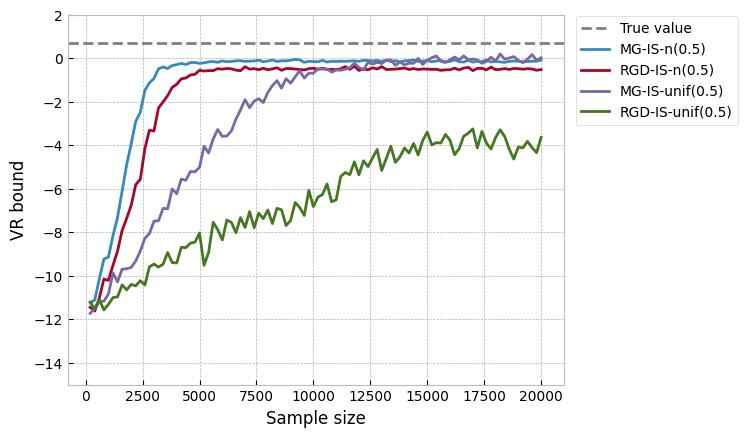} \\    
    \end{tabular}

  \caption{Monte Carlo estimate of the VR Bound for the RGD and the MG approaches ($\eta = 0.1$) when considering each of the target distributions \ref{itemExEWGMM}, \ref{itemExUWGMM} and \ref{itemExEWSMM}.} 
  \label{fig:eta0dot1}
\end{figure}

A key insight from \Cref{fig:weights} is that optimising the mixture weights permits us to select the components appearing in our mixture model according to their overall contribution in constructing a good approximation of the targeted density. 

On the one hand, we enable more flexibility in our variational approximation as we optimise over a set of component parameters $\Theta$ with $J > 1$ instead of just one component parameter $\theta$ ($J = 1$). On the other hand, we avoid complexifying unecessarily the variational distribution returned at the end of the optimisation procedure since we perform mixture weights optimisation alongside components parameters optimisation. That way, we can bypass the limitation of the case $\eta = 0$, which may use more components than needed when approximating the targeted density (this is the case here since the targeted densities considered have three modes at best while $J = \lrcb{10, 50}$). 

However, there is a tradeoff to find between the simplicity of the variational density that is returned and its accuracy at describing the targeted density, which is expressed via the choice of the learning rate $\eta$. Indeed, choosing $\eta$ too large might lead to missing some of the modes while having $\eta$ too small may not discriminate quickly enough between the components. Observe in particular that the number of components $J$ too plays a role in how fast the components selection process occurs, as for the same value of $\eta$, this selection process is slower as $J$ increases. 

\begin{figure}[t]
  \centering
  \begin{tabular}{ccc}
    & $J = 10$ & $J = 50$ \\
    \ref{itemExEWGMM} \vspace{-0.5cm} & \\
    & \includegraphics[width=7cm]{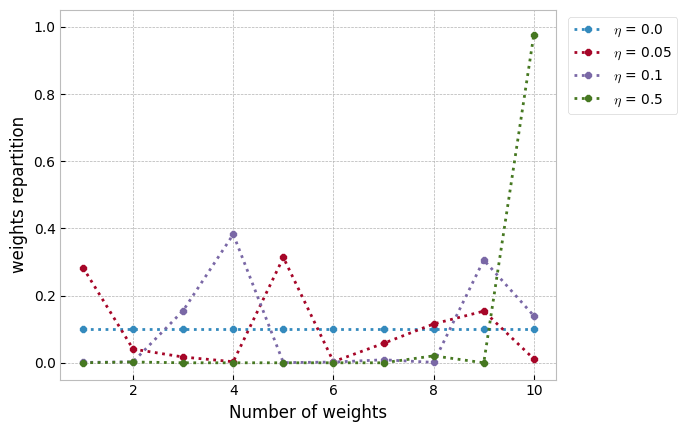} & \includegraphics[width=7cm]{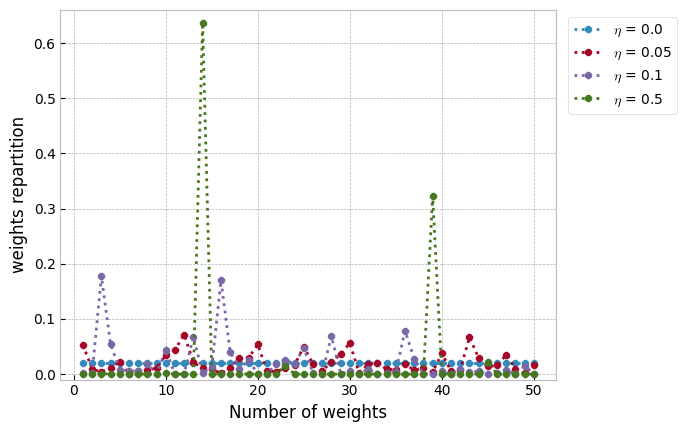} \\
    \ref{itemExUWGMM} \vspace{-0.5cm} & \\
    &\includegraphics[width=7cm]{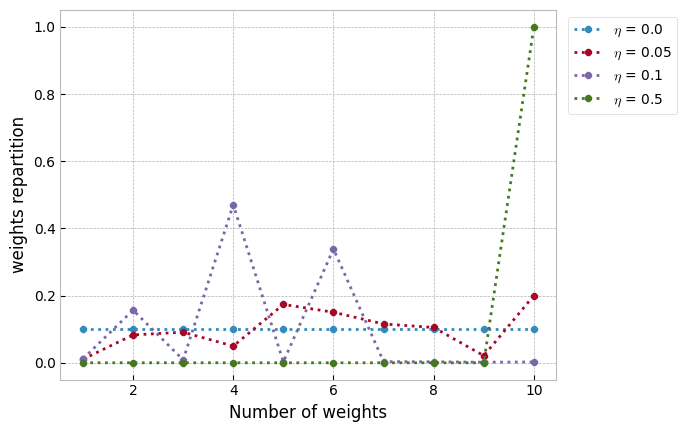} & \includegraphics[width=7cm]{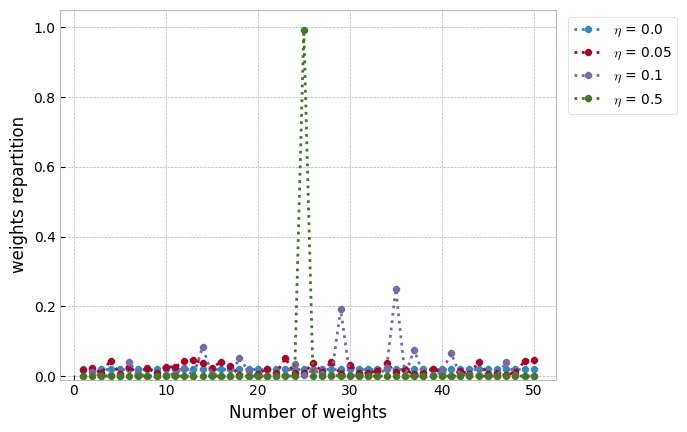} \\
    \ref{itemExEWSMM} \vspace{-0.5cm} & \\
    &\includegraphics[width=7cm]{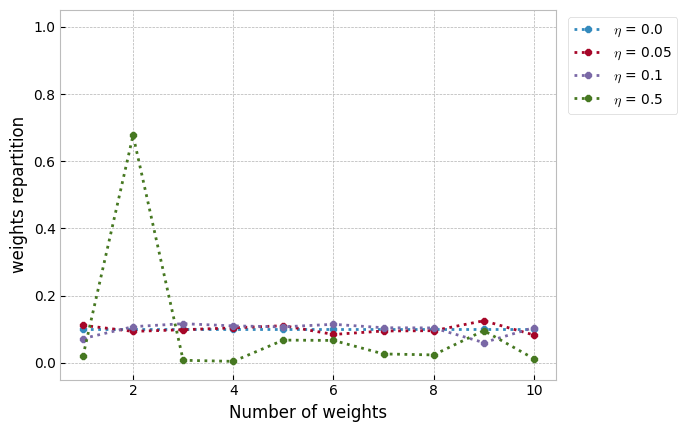} & \includegraphics[width=7cm]{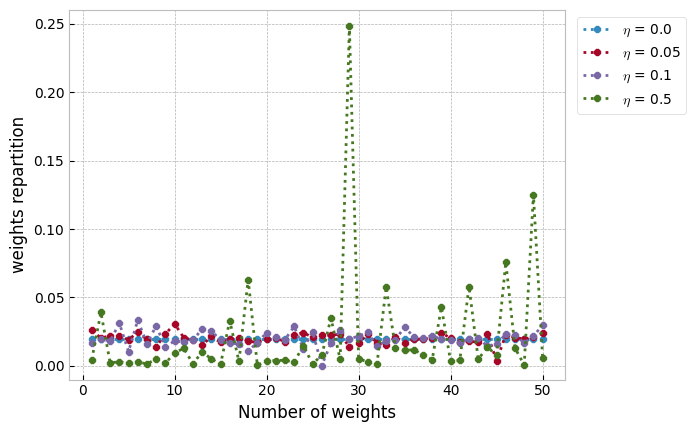} \\    
    \end{tabular}

  \caption{Final mixture weights $\lbd{N}$ for one run of the MG-IS-unif($0.5$) algorithm (varying mixture weights) when considering each of the target distributions \ref{itemExEWGMM}, \ref{itemExUWGMM} and \ref{itemExEWSMM}.} 
  \label{fig:weights}
\end{figure}

\subsection{Bayesian Logistic Regression}

We consider the Bayesian Logistic Regression setting from \cite{gershman2012nonparametric}. Namely, the data $\data = \lrcb{\boldsymbol{c}, \boldsymbol{x}}$ is made of $I$ binary class labels, $c_i \in \lrcb{-1, 1}$, and of $L$ covariates for each datapoint, $\boldsymbol{x}_i \in \Rset^L$. The latent variable $y = \lrcb{\boldsymbol{w}, \beta}$ consists of $L$ regression coefficients $w_l \in \Rset$, and a precision parameter $\beta \in \Rset^+$.  Furthermore, the following model is assumed
\begin{align*}
  &p_0(\beta) = \mathrm{Gamma}(\beta; a, b), \quad p_0(\boldsymbol{w}| \beta) = \mathcal{N}(\boldsymbol{w}; 0, \beta^{-1} \boldsymbol{I}_L), \quad p(c_i = 1 | \boldsymbol{x}_i, \boldsymbol{w}) = \frac{1}{1 + e^{- \boldsymbol{w}^T \boldsymbol{x}_i}} \eqsp,
\end{align*}
where $a = 1$ and $b = 0.01$, so that $p(y, \data) = p_0(\beta) p_0(\boldsymbol{w}|\beta) \prod_{i=1}^I p(c_i | \boldsymbol{x}_i , y)$. We now want to demonstrate the practicability of our framework in a real data scenario. To this end, we select the \emph{Covertype} data set ($581,012$ data points and $54$ features, available at \href{https://www.csie.ntu.edu.tw/~cjlin/libsvmtools/datasets/binary.html}{\nolinkurl{https://www.csie.ntu.edu.tw/
~cjlin/libsvmtools/datasets/binary.html}}) and we compare the RGD and the MG approaches for this choice of target density.

\textit{Implementation details.} The covariance matrices of the mixture components are fixed and equal to $\sigma^2 \boldsymbol{I}_d$ with $\sigma^2 = 1$, $\alpha = 0.2$, $J \in \lrcb{50, 100}$, $d = 56$, $M = 200$, the total number of iterations $N$ is equal to $200$, $\Theta_1$ is generated by sampling from a centered normal distribution with covariance matrix $5\boldsymbol{I}_d$, $\lbd{1} = [1/J, \ldots, 1/J]$ and for all time $n = 1 \ldots N$, $\cte_n = 0$, $\eta_n = \eta$, $\gamma_n = \gamma$ where $\eta \in \lrcb{0.1, 0.5}$ and $\gamma \in \lrcb{0.05, 0.1, 0.2, 0.3}$. Computing $p(y, \data)$ constitutes the major computation bottleneck here since $I = 581,012$. To address this problem, $p(y, \data)$ is approximated according to \Cref{rem:real:data} with subsampled mini-batches of size $100$. The experiments are replicated independently $30$ times and the convergence of the RGD and of the MG approaches is monitored by computing a Monte Carlo estimator of the VR Bound. 

\begin{figure}[t]
  \centering
  \begin{tabular}{ccc}
    $\eta$ \vspace{1.5cm} & $J = 50$ & $J = 100$ \\
    $0.1$ \vspace{-1.7cm} & & \\
    & \includegraphics[width=7cm]{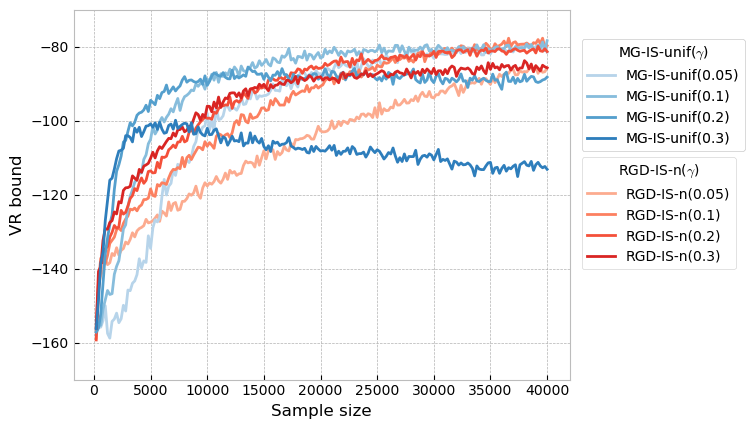} & \includegraphics[width=7cm]{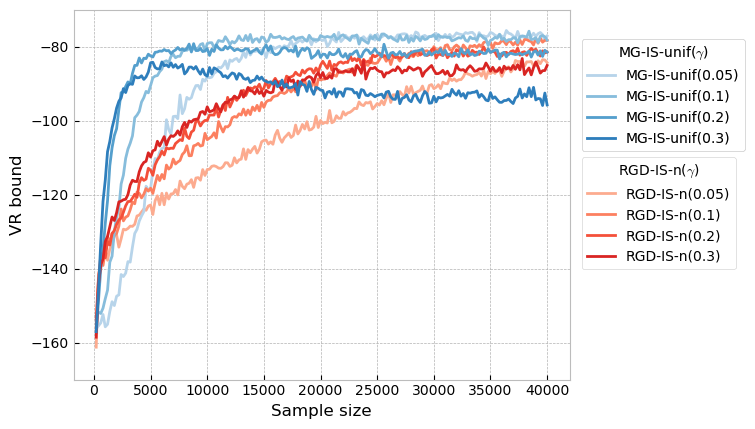} \vspace{1.1cm} \\
    $0.5$ \vspace{-1.7cm} & & \\
    & \includegraphics[width=7cm]{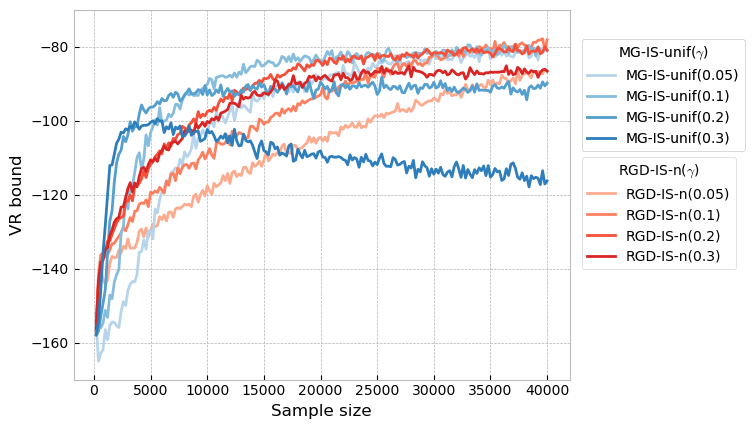} & \includegraphics[width=7cm]{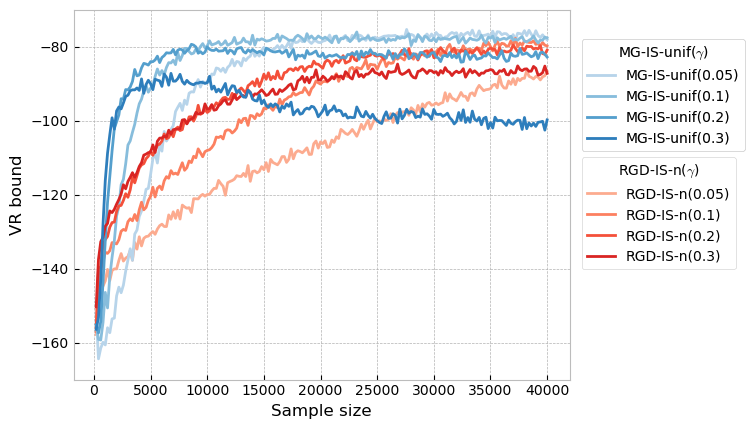} \\
    \end{tabular}

  \caption{Monte Carlo estimate of the VR Bound for the RGD and the MG approaches when considering the Bayesian Logistic Regression on the Covertype data set.} 
  \label{fig:blr}
\end{figure}

Our results are plotted in \Cref{fig:blr}, in which we focus on RGD-IS-n($\gamma$) and MG-IS-unif($\gamma$) (those two versions of Algorithm \ref{algo:GMMsMAIS} are the most interesting to compare in this particular setting since MG-IS-n($\gamma$) enjoys similar performances to RGD-IS-n($\gamma$) and RGD-IS-unif($\gamma$) underperforms, see \Cref{fig:blr:app} of \Cref{sec:addNumExp} for details). We observe that both algorithms are able to learn in a real data scenario for a proper tuning of $(\eta, \gamma)$ and that selecting either $\eta$ or $\gamma$ too small/large can deteriorate the performance (as per the learning rate behaviour associated to those parameters). Furthermore, MG-IS-unif($\gamma$) can be tuned to outperform RGD-IS-n($\gamma$), which confirms our previous empirical findings regarding the relevance of the novel MG approach compared to the more traditional RGD one.


\section{Conclusion}

We introduced a novel methodology to build algorithms ensuring a monotonic decrease in the $\alpha$-divergence at each step. Our methodology enabled simultaneous updates for both the weights and components parameters of a given mixture model, making it suitable for capturing complex multimodal target densities. Our work also connected and improved on different approaches: Gradient Descent for $\alpha$-divergence minimisation, Power Descent and an Integrated EM algorithm. By investigating variational families based on the exponential family, we applied our framework to important classes of models such as Gaussian Mixture Models and Student's $t$ Mixture Models. Finally, we provided empirical evidence that our methodology can be used to enhance the aformentioned existing algorithms.

To conclude, we state several directions to extend our work. Now that we have established a systematic decrease for our iterative schemes, the next step could be to derive convergence results and to compare them with those obtained using typical Gradient Descent schemes. Based on \Cref{thm:admiss}, another direction is to generalise the monotonicity property from \Cref{thm:WeightsMixture} beyond the case $\alpha \in [0,1)$. Lastly, much like it is already the case in traditional Black-Box Variational Inference, we expect that fine-tuning our hyperparameters and resorting to more advanced Monte Carlo methods in the estimation of the intractable integrals appearing in our ideal algorithms will lead to further improved numerical results.


\section*{Acknowledgments and Disclosure of Funding}
We would like to thank the action editor and the reviewers for helpful comments and suggestions on the paper. Kam\'{e}lia Daudel acknowledges support of the UK Defence Science and Technology Laboratory (Dstl) and and Engineering and Physical Research Council (EPSRC) under grant EP/R013616/1. This is part of the collaboration between US DOD, UK MOD and UK EPSRC under the Multidisciplinary University Research Initiative.

\appendix
\section{Deferred Proofs and Results for \Cref{sec:Theta}}

\subsection{Quantifying the Improvement in one Step of Gradient Descent}
\label{appendix:smooth}

\kdtxt{Let $\Tset \subseteq \rset^d$ be a convex set.} Here, $\langle \cdot ,\cdot \rangle $ is the standard inner product on $\rset^d$ and $\| . \|$ is the Euclidean norm.


\begin{defi} A continuously differentiable function $g$ defined on $\Tset$ is said to be $\beta$-smooth if for all $\theta, \theta' \in \Tset$,
$$
\|\nabla g(\theta) - \nabla g(\theta') \| \leq \beta \| \theta - \theta' \| \eqsp.
$$
\end{defi}

\begin{lem} \label{lemma:inegThetaGD} Let $\gamma \in (0,1]$, let $g$
  be a 
  $\beta$-smooth function defined on $\Tset$. Then, for all $\kdtxt{\theta} \in \Tset$ it holds that
\begin{align*}
g(\kdtxt{\theta}) - g \lr{\kdtxt{\theta}- \frac{\gamma}{\beta} \nabla g(\theta)} \geq \frac{\gamma}{2\beta} \| \nabla g(\theta) \|^2 \eqsp.
\end{align*}
\end{lem}

\begin{proof} By assumption on $g$, we have that for all $\theta, \theta' \in \Tset$
\kdtxt{$$
g(\theta ') - g(\theta ) - \langle \nabla g(\theta), \theta' - \theta  \rangle \leq \frac{\beta}{2} \| \theta - \theta ' \|^2 \eqsp.
$$}
In particular, setting \kdtxt{$\theta' = \theta - \frac{\gamma}{\beta} \nabla g(\theta)$} yields
\kdtxt{
\begin{align*}
g(\theta) - g \lr{\theta- \frac{\gamma}{\beta} \nabla g(\theta)} &\geq \frac{\gamma}{\beta} \| \nabla g(\theta) \|^2 - \frac{\gamma^2}{2 \beta} \| \nabla g (\theta) \|^2 \\
&\geq \frac{\gamma}{\beta} \lr{1 -\frac{\gamma}{2}} \| \nabla g(\theta) \|^2 \eqsp.
\end{align*}}
Since $\gamma \in (0,1]$, we can deduce the desired result.
\end{proof}

\kdtxt{
\subsection{Gradient Descent for $\alpha$- / R\'{e}nyi's $\alpha$-divergence Minimisation in \Cref{sec:Theta}}
\label{subsec:GDsteps}  

Let us first write the definition of the $\alpha$-divergence (resp. of R\'{e}nyi's $\alpha$-divergence) between the two absolutely continuous probability measures $K(\theta, \cdot)$ and $\PP$
\begin{align*}
&\diverg \couple[K(\theta, \cdot)] = \int_\Yset \frac{1}{\alpha(\alpha-1)} \lrb{\lr{ \frac{k(\theta, y)}{\posterior}}^{\alpha} -1} \posterior \nu(\rmd y) \\
&\divergR \couple[K(\theta, \cdot)] = \frac{1}{\alpha(\alpha-1)} \log \lr{\int_\Yset {k(\theta, y)}^\alpha {\posterior}^{1-\alpha} \nu(\rmd y)} \eqsp.
\end{align*}
Here, the R\'{e}nyi divergence is defined following the convention from \cite{alpha-beta-gamma}, alternative definitions may use a different scaling factor.

\subsubsection{Gradient Descent for $\alpha$-divergence Minimisation}

As underlined in the introduction, minimising the $\alpha$-divergence $\diverg \couple[\PQ]$ w.r.t $q$ is equivalent to minimising $\Psif(q; p)$ with $p = p(\cdot, \data)$ w.r.t $q$, where we have gotten rid of $p(\data)$ in the optimisation problem as written in \eqref{eq:objective}. Letting $\mathcal{Q}$ be of the form 
$\mathcal{Q} = \set{q: y \mapsto k(\theta,y)}{\theta \in \Tset}$, 
the traditional Variational Inference way to optimise $\Psif(k(\theta, \cdot); p)$ with $p = p(\cdot, \data)$ w.r.t $\theta$ corresponds to performing Gradient Descent steps on $\theta$ to construct a sequence $(\theta_n)_{n \geq 1}$ that converges towards a local optimum of the function $\theta \mapsto \Psif(k(\theta, \cdot); p)$. This procedure involves a well-chosen learning rate policy $(r_n)_{n \geq 1}$ and sets:
\begin{align}\label{eq:GDalphaOneapp}
  \theta_{n+1} = \theta_n - r_n \nabla \Psif(k(\theta_n, \cdot); p) \eqsp, \quad n \geq 1 \eqsp.
\end{align}
Under common differentiability assumptions,
\begin{align*}
\nabla \Psif(k(\theta_n, \cdot); p) & = \nabla \lr{ \int_\Yset \falpha \lr{ \frac{k(\theta_n, y)}{p(y)}} p(y) \nu(\rmd y) } \\
& = \int_\Yset \frac{\partial}{\partial \theta} \lr{  \falpha \lr{ \frac{k(\theta, y)}{p(y)}} } \bigg|_{(\theta,y) = (\theta_n,y)} p(y) \nu(\rmd y) \\
& = \int_\Yset  \falpha' \lr{ \frac{k(\theta_n, y)}{p(y)}} \left.\frac{\partial k(\theta, y)}{\partial \theta} \right|_{(\theta,y) = (\theta_n,y)} \nu(\rmd y)  \\
& = \int_\Yset \frac{k(\theta_n, y)^\alpha p(y)^{1-\alpha} }{\alpha-1} \left.\frac{\partial \log k(\theta, y)}{\partial \theta} \right|_{(\theta,y) = (\theta_n,y)} \nu(\rmd y) 
\end{align*} 
so that \eqref{eq:GDalphaOneapp} becomes
$$
\theta_{n+1} = \theta_n - r_n \int_\Yset \frac{\ratio(y)}{\alpha-1} \left.\frac{\partial \log k(\theta, y)}{\partial \theta} \right|_{(\theta,y) = (\theta_n,y)} \nu(\rmd y) \eqsp, \quad n \geq 1 \eqsp.
$$

\subsubsection{Gradient Descent for R\'{e}nyi's $\alpha$-divergence Minimisation}

Considering yet again $\mathcal{Q} = \set{q: y \mapsto k(\theta,y)}{\theta \in \Tset}$,
minimising R\kdtxt{\'{e}}nyi's $\alpha$-divergence
$$
\divergR \couple[K(\theta, \cdot)] = \frac{1}{\alpha(\alpha-1)} \log \lr{\int_\Yset {k(\theta, y)}^\alpha {\posterior}^{1-\alpha} \nu(\rmd y)} \eqsp
$$
w.r.t $\theta$ can be done by performing Gradient Descent steps on $\theta$ to construct a sequence $(\theta_n)_{n \geq 1}$ that converges towards a local optimum of the function $\theta \mapsto \divergR \couple[K(\theta, \cdot)]$. This procedure involves a well-chosen learning rate policy $(r_n)_{n \geq 1}$ and sets:
\begin{align}\label{eq:GDalphaOneappTwo}
  \theta_{n+1} = \theta_n - r_n \nabla \divergR \couple[K(\theta_n, \cdot)] \eqsp, \quad n \geq 1 \eqsp.
\end{align}
Under common differentiability assumptions and setting $p = p(\cdot, \data)$,
\begin{align*}
\nabla \divergR \couple[K(\theta_n, \cdot)] & = \nabla \lr{ \frac{1}{\alpha(\alpha-1)} \log \lr{\int_\Yset {k(\theta_n, y)}^\alpha {\posterior}^{1-\alpha} \nu(\rmd y)} } \\
& = \frac{1}{\alpha(\alpha -1)} \frac{\nabla  \lr{\int_\Yset {k(\theta_n, y)}^\alpha {p(y,\data)}^{1-\alpha} \nu(\rmd y)}  }{\int_\Yset {k(\theta_n, y)}^\alpha {p(y,\data)}^{1-\alpha} \nu(\rmd y)} \\
& = \int_\Yset \frac{\normratio(y)}{\alpha -1} \left.\frac{\partial \log k(\theta, y)}{\partial \theta} \right|_{(\theta,y) = (\theta_n,y)} \nu(\rmd y) 
\end{align*} 
so that \eqref{eq:GDalphaOneappTwo} becomes
$$
\theta_{n+1} = \theta_n - r_n \int_\Yset \frac{\normratio(y)}{\alpha -1} \left.\frac{\partial \log k(\theta, y)}{\partial \theta} \right|_{(\theta,y) = (\theta_n,y)} \nu(\rmd y)  \eqsp, \quad n \geq 1 \eqsp.
$$

\section{Deferred Proofs and Results for \Cref{sec:MM}}

\subsection{Proof of \Cref{coro:argminMixtureModel}}
\label{subsec:coro:argminMixtureModel:proof}

\begin{proof}[Proof of \Cref{coro:argminMixtureModel}]
Following the proof of \Cref{coro:argmax}, we obtain that \eqref{eq:posMixtureP} holds by using the definition of $\Theta_{n+1}$ combined with the fact that $\alpha \in [0,1)$, \kdtxt{$\diverg[1]\couple[K(\theta_{j,n}, \cdot)][K(\theta_{j,n+1}, \cdot)] \geq 0$} and $\lambda_{j,n} > 0$ for all $j = 1 \ldots J$. \eqref{eq:posMixtureW} holds by \Cref{thm:WeightsMixture} and we can thus apply \Cref{thm:EM:MixtureModel}.
  \end{proof}

\subsection{Extension of \Cref{lem:MF} to Mixture Models}
\label{subsec:extension}

\begin{lem}[Generalised maximisation approach for the mean-field family]\label{lem:MFMM} Let \\ 
  each component of the mixture be a member of the same mean-field variational family so that $k(\theta_j,y) = \prod_{\ell =1}^{L} k^{(\ell)}(\theta_j^{(\ell)}, y^{(\ell)})$ with $\theta_j = (\theta_j^{(1)}, \ldots, \theta_j^{(L)}) \in \Tset$. Then, starting from $\Theta_1 \in \Tset$ \kdtxt{and denoting $\Theta_n =(\theta_{1, n}, \ldots \theta_{j,n})$ with $\theta_{j,n} = (\theta_{j,n}^{(1)}, \ldots, \theta_{j,n}^{(L)})$ for all $n \geq 1$ and all $j = 1 \ldots J$}, solving \eqref{eq:aPMC:updateP} yields the following update formulas: for all $n \geq 1$ and all $j = 1 \ldots J$,
\begin{align*}
\theta^{(\ell)}_{j ,n+1}&= \argmax_{\theta^{(\ell)}} \int_\Yset \lrb{\respat[y] + b_{j,n} k(\theta_n, y)} \log \lr{\frac{k^{(\ell)}(\theta^{(\ell)}, y^{(\ell)})}{k^{(\ell)}(\theta_{j,n}^{(\ell)}, y^{(\ell)})}} \nu(\rmd y), \quad \ell = 1 \ldots L \eqsp.
\end{align*}
\end{lem}

\subsection{Proof of \Cref{coro:GDMixtureModel}}
\label{subsec:coro:GDMixtureModel:proof}

\begin{proof}[Proof of \Cref{coro:GDMixtureModel}] \kdtxt{Following the proof of \Cref{coro:gradientDescent}, we will use that $\gamma_{j,n} \in (0,1]$ and that $g_{j,n}$ is a $\beta_{j,n}$-smooth function. Indeed, we can thus apply \Cref{lemma:inegThetaGD} and we obtain by definition of $\theta_{j, n+1}$ in \eqref{eq:updateThetaGDMixtureModels} that $0 = g_{j,n}(\theta_{j,n}) \geq g_{j,n}(\theta_{j,n+1})$ for all $n \geq 1$ and all $j = 1 \ldots J$,}
  which in turn implies \eqref{eq:posMixtureP}.  In addition, \eqref{eq:posMixtureW} holds by \Cref{thm:WeightsMixture}, hence we can apply \Cref{thm:EM:MixtureModel}.
  \end{proof}

\subsection{Gradient Descent for $\alpha$- / R\'{e}nyi's $\alpha$-divergence Minimisation in \Cref{sec:MM}}
\label{subsec:GDstepsMM}

We obtain the desired results by adapting the reasoning from \Cref{subsec:GDsteps} to the case 
$$
\mathcal{Q} = \set{q: y \mapsto \mu_{\lbd{}, \Theta}k(y) = \sum_{j=1}^{J} \lambda_j k(\theta_j,y)}{\Theta \in \Tset^J} \eqsp,
$$
with $\lbd{} \in \simplex_J$, meaning that we consider the $\alpha$-divergence (resp. R\'{e}nyi's $\alpha$-divergence) between the two absolutely continuous probability measures $\mu_{\lbd{}, \Theta} K$ and $\PP$:
\begin{align*}
& \diverg \couple[\mu_{\lbd{}, \Theta} K] = \int_\Yset \frac{1}{\alpha(\alpha-1)} \lrb{\lr{ \frac{\mu_{\lbd{}, \Theta} k(y)}{\posterior}}^{\alpha} -1} \posterior \nu(\rmd y) \\
& \divergR \couple[\mu_{\lbd{}, \Theta} K] = \frac{1}{\alpha(\alpha-1)} \log \lr{\int_\Yset (\mu_{\lbd{}, \Theta} k(y))^\alpha {\posterior}^{1-\alpha} \nu(\rmd y)},
\end{align*}
(we use the convention from \cite{alpha-beta-gamma} for R\'{e}nyi's $\alpha$-divergence).
}

\subsection{Monotonicity Property for the Power Descent}
\label{sec:power}

\subsubsection{Preliminary Remarks} 

For convenience, we redefine in this section the function $\bmuf$ for all $\mu \in \meas{1}(\Tset)$ by
\begin{align*}
\bmuf(\theta) &= \mathlarger{\int_\Yset} k(\theta,y)  \frac{1}{\alpha-1} \left(\dfrac{\mu k(y)}{p(y)}\right)^{\alpha-1} \nu(\rmd y), \quad \theta \in \Tset \eqsp.
\end{align*}
Then, for all $\eta > 0$, the iteration $\mu \mapsto \iteration (\mu)$ is well-defined if we have
\begin{align}\label{eq:admiss}
0 <\mu(|\bmuf + \cte|^{\frac{\eta}{1-\alpha}})<\infty \eqsp.
\end{align}
Furthermore, \cite{daudel2020infinitedimensional} already established that one transition of the Power Descent algorithm ensures a monotonic decrease in the $\alpha$-divergence at each step for all $\eta \in (0,1]$ and all $\cte$ such that $(\alpha-1)\cte \geq 0$ under the assumption of \Cref{thm:admiss}, which settles the case \ref{item:admiss-alpha-c}. Finally, while we establish our results for \ref{item:admiss-alpha-a} and \ref{item:admiss-alpha-b} in the general case where $\mu \in \meas{1}(\Tset)$, the particular case of mixture models follows immediately by choosing $\mu$ as a weighted sum of dirac measures.

\subsubsection{Extending the Monotonicity}

Let $(\zeta,\mu)$ be a couple of probability measures where $\zeta$ is dominated by $\mu$, which we denote by $\zeta \preceq \mu$. A first lower-bound for the difference $\Psif(\mu k) - \Psif(\zeta k)$ was derived in \cite{daudel2020infinitedimensional} and was used to establish that the Power Descent algorithm diminishes $\Psif$ for all $\eta \in (0,1]$. 
We now prove a novel lower-bound for the difference $\Psif(\mu k) - \Psif(\zeta k)$ which will allow us to extend the monotonicity results from \cite{daudel2020infinitedimensional} beyond the case $\eta \in (0,1]$ when $\alpha < 0$. This result relies on the existence of an exponent $\varrho$ satisfying condition \ref{hyp:f} below, which will later on be used to specify a range of values for $\eta$ ensuring that $\Psif$ is decreasing after having applied one transition $\mu \mapsto \iteration (\mu)$ 

\begin{hyp}{A}
  \item \label{hyp:f} We have $\varrho \in \rset\setminus[0,1]$ and the function
  ${\frho}: u\mapsto \falpha(u^{1/\varrho})$ is non-decreasing and concave on
  $\rset_{>0}$.
\end{hyp}

\begin{prop}\label{prop:fondam}
Assume \ref{hyp:positive}. Let $\alpha \in \rset \setminus \lrcb{1}$, assume that $\varrho$ satisfies \ref{hyp:f} and let $\cte$ be such that $(\alpha-1)\cte \geq 0$. Then, for all $\mu,\zeta\in\meas{1}(\Tset)$ such that $\mu(|\bmuf|)<\infty$  and $\zeta \preceq \mu$,
\begin{equation} \label{eq:bound:fondam}
|\varrho|^{-1}\lrcb{\mu(|\bmuf+ \cte|) - \mu\lr{|\bmuf + \cte| g^\varrho}} \leq \Psif(\mu k) - \Psif(\zeta k) \eqsp,
\end{equation}
where $g$ is the density of $\zeta$ wrt $\mu$, i.e. $\zeta(\rmd
\theta)=\mu(\rmd \theta) g(\theta)$. Moreover, equality holds in \eqref{eq:bound:fondam} if and only if $\zeta=\mu$.
\end{prop} 
\begin{proof} First note that for all $\alpha \in \rset \setminus \lrcb{1}$, we have by \ref{hyp:f} that $f_{\alpha,\varrho}'(u) \geq 0$ for all $u>0$, and thus that $sg(\varrho)=sg(\alpha-1)$ where $sg(v)=1$ if $v\geq 0$ and $-1$ otherwise. Since $sg(f'_\alpha(u))=sg(\alpha-1)=sg(\kappa)$ for all $u>0$, this implies that $\varrho^{-1}\falpha'(u)=|\varrho|^{-1}|\falpha'(u)|$, $\varrho^{-1} \cte = |\varrho^{-1} \cte|$ and finally that $\varrho^{-1} (\bmuf(\theta) + \cte) = |\varrho^{-1}| |\bmuf(\theta) + \cte|$ for all $\theta \in \Tset$, which will be used later in the proof.

Write by definition of ${\frho}$ in \ref{hyp:f} and $\zeta$,
\begin{align}
\label{eq:first}
\Psif(\zeta k) &= \int_\Yset \falpha\lr{\frac{\zeta k(y)}{p(y)}}p(y)\nu(\rmd y) \nonumber \\ &=\int_\Yset {\frho}\lr{\lrb{\frac{\zeta k(y)}{p(y)}}^{\varrho}} p(y)\nu(\rmd y) \nonumber \\
&=\int_\Yset {\frho}\lr{\lrb{\int_\Tset \mu(\rmd \theta) \frac{k(\theta,y)}{\mu k(y)} \lr{\frac{g(\theta)\mu k(y)}{ p(y)}}}^{\varrho}} p(y)\nu(\rmd y)\nonumber\\
& \leq \int_\Yset {\frho} \lr{\int_\Tset \mu(\rmd \theta) \frac{ k(\theta,y)}{\mu k(y)} \lr{\frac{g(\theta)\mu k(y)}{ p(y)}}^{\varrho}} p(y)\nu(\rmd y)
\end{align}
where the last inequality follows from Jensen's inequality applied to the convex function $u\mapsto u^{\varrho}$ (since $\varrho\in \rset\setminus[0,1]$) and the fact that ${\frho}$ is non-decreasing. Now set
\begin{align*}
& u_y=\int_\Tset \mu(\rmd \theta) \frac{ k(\theta,y)}{\mu k(y)} \lr{\frac{g(\theta)\mu k(y)}{ p(y)}}^{\varrho}\\
& v_y=\lr{\frac{\mu k(y)}{p(y)}}^{\varrho}
\end{align*}
and note that
\begin{align}\label{eq:diff}
u_y-v_y=\lr{\frac{\mu k(y)}{p(y)}}^{\varrho} \lr{\int_\Tset \mu(\rmd \theta) \frac{k(\theta,y)}{\mu k(y)}g^\varrho(\theta)-1} \eqsp.
\end{align}
Since ${\frho}$ is concave, ${\frho}(u_y) \leq {\frho}(v_y) +f_{\alpha,\varrho}'(v_y)(u_y-v_y)$. Combining with \eqref{eq:first}, we get
\begin{align}\label{eq:split}
\Psif(\zeta k) &\leq \int_\Yset {\frho}(u_y)p(y)\nu(\rmd y)\\
& \leq \int_\Yset {\frho}(v_y)p(y)\nu(\rmd y)+ \int_\Yset f_{\alpha,\varrho}'(v_y)(u_y-v_y)p(y) \nu(\rmd y)  \nonumber
\end{align}
Note that the first term of the rhs can be written as
\begin{align}\label{eq:rhs1}
\int_\Yset {\frho}(v_y)p(y)\nu(\rmd y)= \int_\Yset \falpha\lr{\frac{\mu k(y)}{p(y)}} p(y) \nu(\rmd y) = \Psif(\mu k) \eqsp.
\end{align}
Using now $f_{\alpha,\varrho}'(v_y)=\varrho^{-1}v_y^{1/\varrho-1}\falpha'(v_y^{1/\varrho})$ and \eqref{eq:diff}, the second term of the rhs of \eqref{eq:split} may be expressed as
\begin{align*}
\int_\Yset & f_{\alpha,\varrho}'(v_y)(u_y-v_y)p(y) \nu(\rmd y) \\
& = \varrho^{-1}\int_\Yset \lr{\frac{\mu k(y)}{p(y)}}^{1-\varrho} \falpha'\lr{\frac{\mu k(y)}{p(y)}} \\
&\hspace{2cm} \lr{\frac{\mu k(y)}{p(y)}}^{\varrho} \lr{\int_\Tset \mu(\rmd \theta) \frac{k(\theta,y)}{\mu k(y)}g^\varrho(\theta)-1} p(y) \nu(\rmd y)\\
& =\varrho^{-1}\int_\Tset \mu(\rmd \theta) \lr{\int_\Yset k(\theta,y) \falpha'\lr{\frac{\mu k(y)}{p(y)}}\nu(\rmd y)}g^\varrho(\theta) \\
&\hspace{2cm} -\varrho^{-1}\int_\Yset \mu k(y) \falpha'\lr{\frac{\mu k(y)}{p(y)}} \nu(\rmd y)\\
&  = \varrho^{-1}\lrcb{\mu\lr{\bmuf \cdot g^\varrho}-\mu(\bmuf)}  \\
& =|\varrho|^{-1}\lrcb{\mu\lr{|\bmuf + \cte| g^\varrho}-\mu(|\bmuf + \cte|)} + |\varrho^{-1}\cte|(1-\mu(g^\varrho)) \eqsp,
\end{align*}
where we have used that $\varrho^{-1} (\bmuf(\theta) + \cte) = |\varrho^{-1}| |\bmuf(\theta) + \cte|$ for all $\theta \in \Tset$ and that $\varrho^{-1} \cte = |\varrho^{-1} \cte|$.
In addition, Jensen's inequality applied to the convex function $u \mapsto u^\varrho$ implies that $\mu(g^\varrho) \geq 1$ and thus
\begin{align}\label{eq:jensen2}
\int_\Yset f_{\alpha,\varrho}'(v_y)(u_y-v_y)p(y) \nu(\rmd y) \leq |\varrho|^{-1}\lrcb{\mu\lr{|\bmuf + \cte| g^\varrho}-\mu(|\bmuf + \cte|)} \eqsp.
\end{align}
Combining this inequality with \eqref{eq:split} and \eqref{eq:rhs1} finishes the
proof of the inequality.
Furthermore, if the equality holds in \eqref{eq:bound:fondam}, then
the equality in Jensen's inequality \eqref{eq:jensen2} shows that $g$ is constant
$\mu$-a.e. so that $\zeta=\mu$, and the proof is completed.
\end{proof}

\begin{rem}The proof of \Cref{prop:fondam} relies on $\falpha'$ being of constant sign. Notice however that the definition of the $\alpha$-divergence in \eqref{eq:gen:divQ} is invariant with respect to the transformation $\tilde \falpha (u) = \falpha(u) + \cte(u-1)$ for any arbitrary constant $\cte$, that is $\falpha$ can be equivalently replaced by $\tilde \falpha$ in \eqref{eq:gen:divQ}. This aspect is in fact expressed through the constant $\cte$ appearing in the update formula, that we however need to assume to satisfy $(\alpha-1)\cte \geq 0$ in our proofs.
\end{rem}

We now plan on setting $\zeta= \iteration(\mu)$ in \Cref{prop:fondam} and obtain that one iteration of the Power Descent yields $\Psif(\iteration(\mu) k) \leq \Psif(\mu k)$. For this purpose and based on the upper bound obtained in \Cref{prop:fondam}, we strengthen the condition \eqref{eq:admiss} as follows to take into account the
exponent $\varrho$
\begin{multline}\label{eq:admiss:varrho}
0<\mu(|\bmuf + \cte|^{\frac{\eta}{1-\alpha}})<\infty \mbox{   and   } \mu(|\bmuf + \cte|g^{\varrho}) \leq \mu(|\bmuf + \cte|) \\ \mbox{ with $g=\frac{|\bmuf + \cte|^{\frac{\eta}{1-\alpha}}}{\mu(|\bmuf+ \cte|^{\frac{\eta}{1-\alpha}})}$} \eqsp.
\end{multline}
This leads to the following result.
\begin{prop}\label{thm:monotone}
  Assume \ref{hyp:positive}. Let $\alpha \in \rset \setminus \lrcb{1}$, assume that $\varrho$ satisfies \ref{hyp:f} and let $\cte$ be such that $(\alpha-1)\cte \geq 0$. Let $\mu\in\meas{1}(\Tset)$ be such that $\mu(|\bmuf|)<\infty$ and assume that $\eta$ satisfies \eqref{eq:admiss:varrho}. Then, the two following assertions hold.
\begin{enumerate}[label=(\roman*)]
\item \label{item:mono1} We have  $\Psif(\iteration(\mu) k) \leq \Psif(\mu k)$.
\item \label{item:mono2} We have $\Psif(\iteration(\mu) k) =\Psif(\mu k)$ if and only if $\mu=\iteration (\mu)$.
\end{enumerate}
\end{prop}

\begin{proof} We treat the case $\cte = 0$ in the proof below (the case $\cte \neq 0$ unfolds similarly). We apply \Cref{prop:fondam} with $\zeta=\iteration (\mu)$ so that $\zeta(\rmd \theta)=\mu(\rmd \theta)g(\theta)$ with $g=|\bmuf|^{\eta /(1-\alpha)}/\mu(|\bmuf|^{\eta /(1-\alpha)})$. Then,
\begin{equation}
\label{eq:monotone:one}
\Psif(\iteration(\mu) k) \leq \Psif(\mu k)+ |\varrho|^{-1}\lrcb{\mu\lr{|\bmuf| g^\varrho}-\mu(|\bmuf|)} \leq \Psif(\mu k) \eqsp,
\end{equation}
where the last inequality follows from condition \eqref{eq:admiss:varrho}. Let us now show \ref{item:mono2}. The {\em if} part is obvious. As for the {\em only if} part,  $\Psif(\iteration(\mu) k) =\Psif(\mu k)$  combined with \eqref{eq:monotone:one} yields
$$
\Psif(\iteration(\mu) k) = \Psif(\mu k)+ |\varrho|^{-1}\lrcb{\mu\lr{|\bmuf| g^\varrho}-\mu(|\bmuf|)} \eqsp,
$$
which is the case of equality in \Cref{prop:fondam}. Therefore, $\iteration (\mu) = \mu$.
\end{proof}

While \Cref{thm:monotone} resembles \cite[Theorem 1]{daudel2020infinitedimensional} in its formulation and in the properties on the iteration $\mu \mapsto \iteration(\mu)$ it establishes, it is important to note that the proof techniques used, and thus the conditions on $\eta$ obtained, are different. This brings us to the proof of \Cref{thm:admiss}. The proof of this theorem requires intermediate results, which are derived in \Cref{sec:proof-mon} alongside the proof of \Cref{thm:admiss}.

\subsubsection{Proof of \Cref{thm:admiss}}
\label{sec:proof-mon}

For the sake of readability, we only treat the case $\cte = 0$ in the proofs below (the case $\cte \neq 0$ unfolds similarly). In \Cref{prop:fondam}, the difference $\Psif(\zeta k)-\Psif(\mu k)$ is split into two terms
$$
\Psif(\zeta k)-\Psif(\mu k)=A(\mu,\zeta)+|\varrho|^{-1}\lrcb{\mu\lr{|\bmuf| g^\varrho}-\mu(|\bmuf|)} \eqsp,
$$
where $g=\rmd \zeta/\rmd \mu$. Moreover, \Cref{prop:fondam} states that $A(\mu,\zeta)$ is
always non-positive. It turns out that the second term is minimal over all positive probability densities $g$ when it is proportional to $|\bmuf|^{1/(1-\varrho)}$, as we show in \Cref{lem:optim:forall:g} below.

\begin{lem}\label{lem:optim:forall:g}
Let $\varrho \in \Rset \setminus [0,1]$. Then, for any positive probability density $g$ w.r.t $\mu$, we have
$$
\mu\lr{|\bmuf| g^\varrho} \geq \lrb{\mu\lr{|\bmuf|^{1/(1-\varrho)}}}^{1-\varrho} \eqsp,
$$
with equality if and only if $g \propto  |\bmuf|^{1/(1-\varrho)}$.
\end{lem}

\begin{proof}
The function $x\mapsto x^{1-\varrho}$ is  strictly convex for
$\varrho \in \rset\setminus[0,1]$. Thus Jensen's inequality yields, for any
positive probability  density $g$ w.r.t. $\mu$,
\begin{align}\label{eq:optimirhs:bound}
\mu\lr{|\bmuf| g^\varrho} &= \int_\Tset \mu(\rmd \theta) \left(\frac{|\bmuf(\theta)|^{1/(1-\varrho)}}{g(\theta)}\right)^{1-\varrho} g(\theta)  \geq \lrb{\mu\lr{|\bmuf|^{1/(1-\varrho)}}}^{1-\varrho}
\end{align}
which finishes the proof of the inequality. The next statement follows from the
case of equality in Jensen's inequality: $g$ must be proportional to
$|\bmuf|^{1/(1-\varrho)}$.
\end{proof}

The next lemma shows that this
choice leads to a non-positive second term, thus implying that
$\Psif(\zeta k)\leq\Psif(\mu k)$.

\begin{lem} \label{lem:optimrhs}
Assume \ref{hyp:positive}. Let $\alpha \in \Rset \setminus \lrcb{1}$ and assume that $\varrho$ satisfies \ref{hyp:f}. Then $\eta= (1-\alpha) / (1-\varrho)$ satisfies \eqref{eq:admiss:varrho} for any $\mu \in \meas{1}(\Tset)$ such
that $\mu(|\bmuf|)<\infty$.
\end{lem}
\begin{proof}
We apply \eqref{eq:optimirhs:bound} with $g=1$ and get that
\begin{align}\label{eq:optimirhs:bound:two}
\lrb{\mu\lr{|\bmuf|^{1/(1-\varrho)}}}^{1-\varrho}\leq \mu(|\bmuf|)<\infty\eqsp.
\end{align}
Then \eqref{eq:admiss:varrho} can be readily checked with  $\eta = (1-\alpha)/(1-\varrho)$. Set $\phi = \eta /(1-\alpha)$. Using that $\mu(|\bmuf|)<\infty$ when $\phi < 0$ and \ref{hyp:positive} for $\phi > 0$, we obtain $\mu(|\bmuf|^\phi) > 0$, which concludes the proof.
\end{proof}

While \Cref{lem:optimrhs} seems to advocate for $g=\rmd \zeta/\rmd \mu$ to be proportional to $|\bmuf|^{1/(1-\varrho)}$, notice that this choice of $g$ might not be optimal to minimize $\Psif(\zeta k)-\Psif(\mu k)$, as $A(\mu,\zeta)$ also depends on $g$ through $\zeta$. In the next lemma, we thus propose another choice of the tuning parameter $\eta$, which also satisfies \eqref{eq:admiss:varrho} for any $\mu \in \meas{1}(\Tset)$ such
that $\mu(|\bmuf|)<\infty$.

\begin{lem} \label{lem:unSurRho}
Assume \ref{hyp:positive}. Let $\alpha \in \Rset \setminus \lrcb{1}$ and assume that $\varrho$ satisfies \ref{hyp:f}. Let $\mu \in \meas{1}(\Tset)$ be such that $\mu(|\bmuf|)<\infty$. Further assume that $|\varrho| \geq 1$, then $\eta=(\alpha-1)/\varrho$ satisfies \eqref{eq:admiss:varrho}.
\end{lem}
\begin{proof}
Setting $g\propto |\bmuf|^{-1/\varrho}$, we get
$$
\mu(|\bmuf|g^\varrho)=\mu(|\bmuf|^{1-\varrho/\varrho}) [\mu(|\bmuf|^{-1/\varrho})]^{-\varrho}=[\mu(|\bmuf|^{-1/\varrho})]^{-\varrho}\leq \mu(|\bmuf|) 
$$
where the last inequality follows from Jensen's inequality applied to the convex function $u \mapsto u^{-\varrho}$ (since $|\varrho| \geq 1$). Since $\mu(|\bmuf|)<\infty$, the parameter $\eta=(\alpha-1)/\varrho$ satisfies \eqref{eq:admiss:varrho}. Set $\phi = \eta/(1-\alpha)$. Using that $\mu(|\bmuf|)<\infty$ when $\phi < 0$ and \ref{hyp:positive} for $\phi > 0$, we obtain $\mu(|\bmuf|^\phi) > 0$, which concludes the proof.
\end{proof}

\Cref{lem:optimrhs} and \Cref{lem:unSurRho} allow us to define a range of values for $\eta$ that decreases $\Psif$ after one transition of the Power Descent, under the assumption that $\varrho$ satisfies \ref{hyp:f}. Now, in
order to prove \Cref{thm:admiss} and given $\alpha \in \rset \setminus \lrcb{1}$, we need to check which values of $\varrho$ satisfy the conditions expressed in \ref{hyp:f}. \newline


\begin{proof}[Proof of \Cref{thm:admiss}] The proof consists in verifying that we can apply \Cref{thm:monotone}, that
  is, given $\alpha \in \rset \setminus \lrcb{1}$, we must find a range of constants $\varrho$ which satisfy \ref{hyp:f}. We then use \Cref{lem:optimrhs} or \Cref{lem:unSurRho} to deduce that, for the provided constants $\eta$, \eqref{eq:admiss:varrho} holds. 

  \begin{enumerateList}
\item Assumption \ref{hyp:f} holds for
  all $\varrho<0$, with $\frho(u)= -\log(u)/\varrho$. Moreover, by definition of
  $\bmuf[\mu]$, we get for all $n \geq 1$,
$$
\mu(|\bmuf[\mu]|)=\int_\Yset  \mu k(y) \frac{p(y)}{\mu k(y)} \nu(\rmd y)=\int_\Yset p(y) \nu(\rmd y)<\infty\eqsp.
$$
Combining with \Cref{lem:optimrhs} and \Cref{lem:unSurRho}, \eqref{eq:admiss:varrho} holds for all $\mu \in \meas{1}(\Tset)$ and for any $\eta \in (0,1]$.
\item  Observing that for $\alpha \notin \{0,1\}$,
$$
    \frho(u)=\frac{1}{\alpha(\alpha-1)} \lr{u^{\alpha/\varrho}-1}\;,
    $$
    we get that  \ref{hyp:f} holds for $\varrho\leq\alpha$ if $\alpha<0$
    Lemmas \ref{lem:optimrhs} and~\ref{lem:unSurRho} provide the corresponding
    ranges for $\eta$ in
    Cases~\ref{item:admiss-alpha-a} and~\ref{item:admiss-alpha-b}. To finish the proof, we now show that for all $\mu \in \meas{1}(\Tset)$, $\mu(|\bmuf[\mu]|)$ is finite, so that Lemmas \ref{lem:optimrhs}
    and~\ref{lem:unSurRho} can indeed be applied.

    Since $u\falpha'(u)=\alpha \falpha(u)+1/(\alpha-1)$, we have, for all $n \geq 1$,
\begin{align} \label{eq:majoMuN}
\mu(|\bmuf[\mu]|)&=\int_\Yset \left| \lr{\frac{\mu k(y)}{p(y)}} \falpha'\lr{\frac{\mu k(y)}{p(y)}}\right| p(y)\nu(\rmd y)\\
&\leq |\alpha| \int_\Yset \left|\falpha\lr{\frac{\mu k(y)}{p(y)}}\right| p(y)\nu(\rmd y) + \frac{1}{|\alpha-1|} \nonumber
\end{align}
Using that $\Psif(\mu k) > -\infty$ (which is a consequence of \ref{hyp:positive} and of Jensen's inequality applied to the convex function $u \mapsto u \falpha(1/u)$), the r.h.s is finite if and only if $\Psif(\mu k)$ is
finite, which is what we have assumed and thus the proof is finished.
\end{enumerateList}

\end{proof}

\section{Deferred Proofs and Results of \Cref{sec:expo-family}}
\label{sec:proofs:section5}

We start with a useful lemma.

\subsection{A Useful Lemma}

The following lemma will be used for the proofs of Theorems
\ref{thm:generic-update-canonic-expo-family} and \ref{thm:argrmax-objective-equivalence} and \Cref{thm:gradiant-g-canonic-expo-family}.

\begin{lem}\label{lem:useful-1}
  Let  $k^{(o)}$ satisfy \ref{hyp:expo}. We have, for all
  $\zeta,\zeta'\in\intE_0$,
  \begin{align}
    \label{eq:basin-expo-family-convex-ineq}
    &        \log \lr{ \frac{k^{(o)}(\zeta',y)}{k^{(o)}(\zeta,y)}}
      \geq \pscal[E]{S(y)-\nabla
      A(\zeta')}{\zeta'-\zeta}\;,
      \quad\text{with equality if and only if $\zeta'=\zeta$,} \\
    \label{eq:kullback-exponential-canonical}
    &   0\geq \int_\Yset
      \log \lr{\frac{k^{(o)}(\zeta',y)}{k^{(o)}(\zeta,y)}}\,k^{(o)}(\zeta,y)\;\nu(\rmd
      y)\geq  \pscal[E]{\nabla A(\zeta)-\nabla A(\zeta')}{\zeta'-\zeta}\;.
  \end{align}
\end{lem}
\begin{proof}
  Let $\zeta,\zeta'\in\intE_0$. We have, since $\intE_0$ is convx, 
  \begin{align*} 
    \log \lr{ \frac{k^{(o)}(\zeta',y)}{k^{(o)}(\zeta,y)}}
    &=\pscal[E]{S(y)}{\zeta'-\zeta}
      +A(\zeta)-
      A(\zeta')\\
        &=\pscal[E]{S(y)}{\zeta'-\zeta}+ \pscal[E]{\nabla
          A(\zeta')}{\zeta-\zeta'}\\
          &\phantom{=}+(\zeta-\zeta')^T \lr{\int_0^1(1-t)\,
          \lrb{\nabla\nabla^TA(t \,\zeta+(1-t)\,\zeta')}\rmd t}(\zeta-\zeta')
           \;.
  \end{align*}
  The inequality and the equality case
  of~(\ref{eq:basin-expo-family-convex-ineq}) then follow from the
  fact that the hessian $\nabla\nabla^TA$ is positive definite in
  $\intE_0$.  Using~(\ref{eq:basin-expo-family-convex-ineq})
  and~(\ref{eq:partialA:expo-family}), we further obtain the lower
  bound displayed in~(\ref{eq:kullback-exponential-canonical}) on the negated
  exclusive Kullback-Leibler divergence between the two (absolutely
  continuous w.r.t. $\nu$) probability distributions with probability
  density functions $k^{(o)}(\zeta,\cdot)$ and $k^{(o)}(\zeta',\cdot)$
  respectively.
\end{proof}

\subsection{Proofs of \Cref{sec:argm-solut-param-expo}}
\label{sec:argm-solut-param-expo:proofs}
The following lemma will be useful.
\begin{lem}\label{lem:cond-p-varphi}
  Let let $p$ and $g$ be functions measurable from $(\Yset,\Ysigma)$ to
  the Borel sets of $\rset_+$ such that
  \begin{equation}
    \label{eq:cond-p-varphi-hyp}
  0<\int_{\Yset}p(y)\,(1+g(y))^{1/(1-\alpha)}\;\nu(\rmd y) <\infty\;.    
  \end{equation}
  Let $\alpha\in[0,1)$ and $k:\Tset\times\Yset\to\rset_+$ be a positive
  kernel density with respect to $\nu$. Let $\theta_0\in\Tset$ and $\mu$ be
  a probability on $\Tset$ such that
  \begin{equation}
    \label{eq:cond-p-varphi-hyp2}
    \inf_{y\in\Yset}\frac{k(\theta_0,y)}{\mu k(y)} >0\;. 
  \end{equation}
  Then, setting
  $$
  \varphi(y)=k(\theta_0,y)^\alpha\,\lr{\frac{p(y)}{\mu k(y)}}^{1-\alpha}\;,
  $$
  we have $\int \varphi\;\rmd\nu>0$ and the probability density
  function $\check{\varphi}=\varphi/\lr{\int \varphi\;\rmd\nu}$
  satisfies
  \begin{equation}
    \label{eq:cond-p-varphi-conc}
  \int \check{\varphi}(y)\,g(y)\,\nu(\rmd y) <\infty \;.
\end{equation}
\end{lem}
\begin{proof}
  We have $\mu k>0$ by assumption on $k$ and $\mu$. Thus $\varphi$ is
  well defined and $\int
  \varphi\;\rmd\nu>0$ as a consequence of $p\geq0$ with the left-hand
  side of~(\ref{eq:cond-p-varphi-hyp}). We now
  check~(\ref{eq:cond-p-varphi-conc}) in which $\check{\varphi}$ can
  equivalently be replaced by $\varphi$. We have
  \begin{align*}
\int \varphi(y)\,g(y)\,\nu(\rmd y)&=    \int 
                                    k(\theta_0,y)\,\lr{\frac{p(y)}{\mu k(y)}}^{1-\alpha}\,g(y)\,\nu(\rmd
                                    y) \\
&    \leq \lr{\int k(\theta_0,y)\,\frac{p(y)}{\mu k(y)}\,g(y)^{1/(1-\alpha)}\,\nu(\rmd
    y) }^\alpha\;,
  \end{align*}
by Jensen's inequality and concavity of $x\mapsto x^{1-\alpha}$
for $\alpha\in[0,1)$. Hence we obtain~(\ref{eq:cond-p-varphi-conc}) as
a consequence of~(\ref{eq:cond-p-varphi-hyp}) and~(\ref{eq:cond-p-varphi-hyp2}).
\end{proof}

  \begin{proof}[Proof of \Cref{thm:generic-update-canonic-expo-family}]
  Let  $\zeta_0\in\mathcal{O}$ and  $b\geq0$.
  In the following, we denote by $\mathcal{C}_0(\zeta)$ the value of the integral in
  the argmax in~(\ref{generic-update-canonic-expo-family-argmax}).
  Using that $\check{\varphi}$ is a probability density
  function paired up with \eqref{eq:partialA:expo-family} and \eqref{generic-update-canonic-expo-family-hyp}, we have
  \begin{align*}
       \int_\Yset \lrb{ \check{\varphi}(y) + b \, k^{(o)}(\zeta_0 , y)} 
   \nu(\rmd y)  & = 1+b >0 \;, \\
   \int_\Yset \lrb{ \check{\varphi}(y) + b \, k^{(o)}(\zeta_0 , y)} S(y)
   \nu(\rmd y)  & = \int_\Yset \check{\varphi}\,S \;\rmd\nu +b\,
   \nabla A(\zeta_0) \in E \;.
  \end{align*}
  By~(\ref{eq:basin-expo-family-convex-ineq}) in \Cref{lem:useful-1}, for all $\zeta'\neq\zeta$
  in $\mathcal{O}$ and all $y\in\Yset$, we have
  $$
  \log\left(\frac{k^{(o)}(\zeta', y)}{k^{(o)}(\zeta, y)}\right) >
  \pscal[E]{S(y)-\nabla A(\zeta')}{\zeta'-\zeta} \;.
  $$
  Hence, combining with the three previous identities, we get that, for all $\zeta'\neq\zeta$
  in $\mathcal{O}$,
  \begin{align}
    \mathcal{C}_0(\zeta')-\mathcal{C}_0(\zeta) & =
                                               \int_\Yset \lrb{\check{\varphi}(y) + b \, k^{(o)}(\zeta_0 , y)}
                                               \log\left(\frac{k^{(o)}(\zeta', y)}{k^{(o)}(\zeta, y)}\right) \nu(\rmd y) \nonumber \\
                                             & > \pscal[E]{ \int_\Yset \check{\varphi}\,S \;\rmd\nu +b\,
                                               \nabla A(\zeta_0) -\lr{1+b}\nabla
                                               A(\zeta')}{\zeta'-\zeta} \nonumber\\
                                             & =\lr{1+b}
                                               \pscal[E]{ \mathbf{s}^* -\nabla A(\zeta')}{\zeta'-\zeta}\;, \label{eq:generic-expo-canon-inter}
  \end{align}
  where we used the definition of $\mathbf{s}^*$
  in~(\ref{generic-update-canonic-expo-family-cond}). Now, since
  $\nabla A$ is a $\mathcal{C}^\infty$-diffeomorphism on $\intE_0$, if
  $\mathbf{s}^*$ belongs to $\nabla A(\mathcal{O})$ then there exists
  a unique $\zeta^* \in\mathcal{O}$ such that
  $\nabla A(\zeta^*)=\mathbf{s}^*$. As a result, plugging this into
  \eqref{eq:generic-expo-canon-inter} and setting $\zeta' = \zeta^*$,
  we have that $\mathcal{C}_0(\zeta^*) > \mathcal{C}_0(\zeta)$ for all
  $\zeta \neq \zeta^*$. This shows that $\zeta^*$ is the (unique)
  solution to the argmax
  problem~(\ref{generic-update-canonic-expo-family-argmax}).
  
  We conclude this proof with the proof of the reciprocal
  implication. Suppose that the argmax
  problem~(\ref{generic-update-canonic-expo-family-argmax}) has a
  solution $\zeta\in\mathcal{O}$ so that, for all
  $\zeta'\in\mathcal{O}$,
  $\mathcal{C}_0(\zeta')-\mathcal{C}_0(\zeta)\leq0$. Using
  \eqref{eq:generic-expo-canon-inter}, this would imply
  $\pscal{\mathbf{s}^* -\nabla A(\zeta')}{\zeta'-\zeta}< 0$ for all
  $\zeta'\in\mathcal{O}\setminus\{\zeta\}$. Since $\mathcal{O}$ is an
  open set, we can take $\zeta'=\zeta+\epsilon\lr{\mathbf{s}^* -\nabla
    A(\zeta)}$ with $\epsilon > 0$ small enough, which gives
    \begin{equation*}
    \pscal{\mathbf{s}^* -\nabla A(\zeta')}{\mathbf{s}^* -\nabla A(\zeta)}< 0 \eqsp.
    \end{equation*}
  Now letting $\epsilon\downarrow0$ and since $\nabla A$ is continuous, we get that
  $\mathbf{s}^* -\nabla A(\zeta)=0$ and thus that $\mathbf{s}^*$
  belongs to $\nabla A(\mathcal{O})$.
  \end{proof}
  
 \begin{proof}[Proof of \Cref{coro:GaussianTemp}]
  First note that letting $\nu$ be the $d$-dimensional Lebesgue measure and setting $S(y)=(y,-yy^T/2)$, the condition \eqref{generic-update-canonic-expo-family-hyp} written in \Cref{thm:generic-update-canonic-expo-family} takes the form  
  \begin{equation}
    \label{generic-update-canonic-gaussian-hyp-proof} \int_\Yset \|y\|^2\,\check{\varphi}(y) \;\rmd y<\infty\;. 
  \end{equation}
  Furthermore, setting $\changevarcomp(\theta)=(\Sigma^{-1}m,\Sigma^{-1})$ with $\theta = (m,\Sigma)$, we obtain that $\changevarcomp(\Tset)=\rset^d\times\mathcal{M}_{>0}(d)$, which is
    the whole interior set $\intE_0$, seen as an open set of
    $E=\rset^d\times\rset^{d\times d}$. Hence \Cref{thm:generic-update-canonic-expo-family} applies and we can
    use the interpretation \eqref{eq:generic-update-canonic-expo-family-solution-theta} of the equation \eqref{generic-update-canonic-expo-family-solution} with $k(\theta^*,y) = \mathcal{N}(y; m^*, \Sigma^*)$ and $k(\theta_0,y) = \mathcal{N}(y; m_0, \Sigma_0)$, that is
    \begin{equation}    \label{eq:generic-update-canonic-expo-family-solution-theta-inter}
      \int_\Yset S(y) ~ \mathcal{N}(y; m^*, \Sigma^*)\nu(\rmd y) = \int_\Yset S(y) \psi(y)\nu(\rmd y) \eqsp,
    \end{equation}
    where $\psi$ is
    the mixture
    \begin{equation*}
      \psi(y)= \frac1{1+b}  \check{\varphi}(y) + \frac
      b{1+b} \mathcal{N}(y; m_0, \Sigma_0) \eqsp, \quad y \in \Yset \;.
    \end{equation*}
    Using now that $S(y)=(y,-yy^T/2)$, \eqref{eq:generic-update-canonic-expo-family-solution-theta-inter} is equivalent to saying that the distributions on both sides of this equation have the same mean and covariance matrix. Since the
    left-hand side distribution is a Gaussian distribution, we
    directly obtain that the solution
    of \eqref{eq:generic-update-canonic-expo-family-solution-theta-inter} is
    given by $\theta^*=(m^*,\Sigma^*)$ with $m^*$ and $\Sigma^*$ being
    the mean and covariance matrix of the random variable $Y$ with density $\psi$. 
    
    Observe that these mean and covariance matrix are well-defined due to the condition~(\ref{generic-update-canonic-gaussian-hyp-proof}). This means that we only have to check that the covariance matrix is positive definite in order to have a solution in
    $\Tset=\rset^d\times\mathcal{M}_{>0}(d)$. If $b>0$, then
    $1/(1+b)<1$ and the fact that $\theta\in\Tset$ guaranties that the
    covariance matrix of $Y$ is positive definite. If $b=0$,
    then $Y$ has density $\check\varphi$ and the non-degenerate
    condition on $\check{\varphi}$ guaranties that the covariance is
    also positive definite in this case. This concludes the proof.
 \end{proof}

 \begin{proof}[Proof of \Cref{cor:iterative-updates-expo-family}]
  Recall from \eqref{eq:respat} that $\respat =  k(\theta_{j,n}, y)
  (p(y) / \mu_n k(y))^{1-\alpha}$ with $k(\theta_{j,n}, y) > 0$. Then
  \Cref{lem:cond-p-varphi} gives that
  \begin{align}
  \label{eq:proofmMaxApproachGen}
  0 < \int_\Yset \lr{1+ \norm[E]{S(y)}}\; \respat \;\nu(\rmd y) < \infty \;.
  \end{align}
  We can thus apply \Cref{thm:generic-update-canonic-expo-family} with
  $\zeta_{0}=\changevarcomp(\theta_{j,n})$,
  $\mathcal{O}=\changevarcomp(\Tset)$,
  $b=b_{j,n}/\int_\Yset \ratiogen \;\rmd \nu$ and
  $\check{\varphi}=\ratiogen /\int \ratiogen \;\rmd \nu$, which
  satisfies~\eqref{generic-update-canonic-expo-family-hyp}.
  \end{proof}

\subsection{Proofs of \Cref{sec:gradient-smoothness-expo-family} }
\label{sec:gradient-smoothness-expo-family:proofs}
\begin{proof}[Proof of \Cref{thm:gradiant-g-canonic-expo-family}]
  Following the proof of \Cref{thm:generic-update-canonic-expo-family} but with
  $b=0$, we have using \eqref{eq:generic-expo-canon-inter} that: for all $\zeta',\zeta\in\intE_0$,
  $$
  \pscal[E]{ \nabla
  A(\zeta)-\int_\Yset S\,\check{\varphi}\;\rmd\nu}{\zeta'-\zeta}
\leq g^{(o)}(\zeta')-g^{(o)}(\zeta) \leq
\pscal[E]{\nabla A(\zeta')-\int_\Yset S\,\check{\varphi}\;\rmd\nu }{\zeta'-\zeta}\;.
$$
Since $\nabla A$ is continuous in $\intE_0$, we
get that $g^{(o)}$ is $\mathcal{C}^1$ over $\intE_0$ and, for all $\zeta\in\intE_0$,
$$
\nabla g^{(o)}(\zeta)=\nabla A(\zeta)-\int_\Yset S\,\check{\varphi}\;\rmd\nu \;.
$$
The $\beta_0$-smoothness condition follows from the mean value theorem
applied to $\nabla g^{(o)}$ whose Jacobian is $\nabla\nabla^T A$. 
\end{proof}

\begin{proof}[Proof of \Cref{coro:GaussianTempGrad}]
  The context is similar to the proof of \Cref{coro:GaussianTemp}
  except that the covariance matrix is known and thus the sufficient
  statistic reduces to 
  $S(y)=y$ and the canonical parameter to 
  $\changevarcomp(\theta)=\Sigma^{-1}\theta$. The given expressions of
  $h$ and $A$ follow. Then we have, for all $\zeta\in\rset^d$, $\nabla
  A(\zeta)=\Sigma\zeta$. Applying
  \Cref{thm:gradiant-g-canonic-expo-family}, we get that 
    $$
    \nabla g^{(o)}(\zeta) = \Sigma\zeta-\int_\Yset y  \check{\varphi}(y) \;\rmd y \eqsp.
    $$
    This gives that $g^{(o)}$ is $\beta_0$-smooth over $\rset^d$ with
    $\beta_0$ equal to the maximal eigenvalue of $\Sigma$ and
    leads to the gradient
    step
    $$
    \zeta=\zeta_0-\frac{\gamma}{\beta_0}\lr{\Sigma\zeta_0-\int_\Yset y  \check{\varphi}(y)
      \;\rmd y}\;.
    $$
    Multiplying on both sides by $\Sigma$ leads to the
    gradient-based update~(\ref{eq:gradient-step-canonical-gaussian}).
    
    We now derive~(\ref{eq:gradient-step-noncanonical-gaussian}). Using that $g=g^{(o)}\circ\changevarcomp$, we get that
    $\nabla g(\theta) = \Sigma^{-1}\,\nabla g^{(o)}(\Sigma^{-1}\theta)$
    and thus
  $$
  \nabla g(\theta) = \Sigma^{-1} \lr{\theta- \int_\Yset y~\check{\varphi}(y)} \eqsp.
  $$
  We obtain that $g$ is $\beta$-smooth over $\rset^d$ with $\beta$
  equal to the maximal eigenvalue of $\Sigma^{-1}$, which leads to the
  gradient update~(\ref{eq:gradient-step-noncanonical-gaussian}).
\end{proof}

\subsection{Proofs of \Cref{sec:partial-mixture-expo-family}}
\label{sec:proof:mix:exp}

\subsubsection{Proofs of Preliminary Results}
\label{sec:proof:mix:exp:prelim:res}

  It will be convenient to use the shorthand
  notation: for all $(\ell,y,\xi) \in\Zset \times \Yset\times F$,
  \begin{align*}
    &  k_{\xi}(\ell,y) \eqdef k^{(o)}(\ell(\xi),y)\;,\\
    & \check{k}_{\xi,\tau}(\ell,y) \eqdef \frac{k_{\xi}(\ell,y)}{\tau k_\xi(y)}\;,
  \end{align*}
  where $\tau k_\xi$ denotes the
  mapping $y\mapsto\int k_{\xi}(\ell,y)\;\tau (\rmd \ell)$. 
  
\begin{proof}[Proof of \Cref{prop:argmax-linear-mixture-exponential}]
  We have,  for all $\tau_0,\tau\in\parammixproba$, $\xi_0,\xi\in\ F$ and
  $y\in\Yset$, using the definitions above and that of $\mathrm{D}_{\tau,\tau_0}$, 
  \begin{align*}
  \log\frac{\tau k_{\xi}(y)}{\tau_0 k_{\xi_0}(y)}
    & = \log \int_\Zset
      \left(\frac{k_{\xi}(\ell,y)}{k_{\xi_0}(\ell,y)}\right)\,\check{k}_{\xi_0,\tau_0}(\ell,y)\;\tau (\rmd \ell)\\
    & = \log \int_\Zset
      \lr{\frac{k_{\xi}(\ell,y)}{k_{\xi_0}(\ell,y)}\;\mathrm{D}_{\tau,\tau_0}(\ell)}\,\check{k}_{\xi_0,\tau_0}(\ell,y)\;\tau_0
      (\rmd \ell)\\
    &\geq
   \int_\Zset\,\check{k}_{\xi_0,\tau_0}(\ell,y)\,
      \log\lr{\frac{k_{\xi}(\ell,y)}{k_{\xi_0}(\ell,y)}\;\mathrm{D}_{\tau,\tau_0}(\ell)}\;\tau_0 (\rmd \ell) \;,
  \end{align*}
  where we used the concavity of the log function and the fact that
  $\check{k}_{\xi_0,\tau_0}(\cdot,y)$ is a probability density function with
  respect to $\tau_0$. Let  $\xi=\changevar(\vartheta)$ and
  $\xi_0=\changevar(\vartheta_{0})$ and set $\theta=(\vartheta,\tau)$
  and  $\theta_0=(\vartheta_0,\tau_0)$, so that, by
  \Cref{defi:LMEFfam}, we have $\tau
  k_{\xi}(y)=k^{(1)}(\theta,y)$ and $\tau_0 k_{\xi_0}(y)=k^{(1)}(\theta_0,y)$.
  Applying the previous bound  and integrating
  w.r.t. $\check{\varphi}(y)\nu(\rmd y)$ thus yields 
  \begin{multline}
  \int_\Yset \check{\varphi}(y) \, \log
  \lr{\frac{k^{(1)}(\theta, y)}{k^{(1)}(\theta_0, y)}} \nu(\rmd y)\\ \geq
   \int_{\Yset}\lr{\int_{\Zset} \varphi(\ell, y)
      \log\lr{\frac{k_{\xi}(\ell,y)}{k_{\xi_0}(\ell,y)}\;\mathrm{D}_{\tau,\tau_{0}}(\ell)}\;\tau_0 (\rmd \ell)}\nu(\rmd y)  \;, \label{eq:mixt:exp:one}
  \end{multline}
  where we used that
  $\check{\varphi}(y)\check{k}_{\xi_0,\tau_0}(\ell,y)=\varphi(\ell,
  y)$ as a consequence of~(\ref{eq:check-varphi-varphiL-carphiQ}). We
  conclude by noting
  that~(\ref{eq:update-cond-tau-general_mixture}),~(\ref{eq:update-cond-theta-general_mixture})
  and~(\ref{eq:mixt:exp:one}) imply~(\ref{eq:generic-decrease-theta-cond}).
\end{proof}
\begin{proof}[Proof of \Cref{cor:monot-decr-general-LM}]
  Suppose that $\changevar(\vartheta)$ belongs to the
  argmax~(\ref{eq:update-cond-theta-general_mixture_argmax}), and 
  let us show that~(\ref{eq:update-cond-theta-general_mixture}) holds.
  Note that the objective function in the argmax is zero for
  $\xi=\xi_0$ and thus must be non-negative for
  $\xi=\changevar(\vartheta)$ in the argmax. Note also
  that the left-hand side of~(\ref{eq:update-cond-theta-general_mixture}) is the sum of this
  objective function at $\xi=\changevar(\vartheta)$ with the term
  \begin{equation}
    \label{eq:last-todo-positif-to-get-update-theta_LM}
    -\int_\Yset\lr{\int_\Zset
      b_0(\ell) \, k_{\xi_0}(\ell,y) \log\lr{\frac{k_{\changevar(\vartheta)}(\ell,y)}{k_{\xi_0}(\ell,y)}}\;\tau_0 (\rmd \ell)}\nu(\rmd y)\;,
\end{equation}
and it only remains to show that this term is also non-negative. To
conclude, we observe that, for $\tau_0$-a.e. $\ell\in\Zset$,
Condition~\ref{asump:gradA-well-def} of \Cref{def:expo-family} implies
that $k_{\xi_{0}}(\ell,\cdot)$ and $k_{\changevar(\vartheta)}(\ell,\cdot)$ are
probability densities, so that
    $$
     - \int_\Yset  k_{\xi_0}(\ell,y) \log\lr{\frac{k_{\changevar(\vartheta)}(\ell,y)}{k_{\xi_0}(\ell,y)}}\;\nu(\rmd y)\geq0 \;,
    $$
  and thus, $b_0$ being non-negative,
  $$
  -\int_\Zset\lr{\int_\Yset
  b_0(\ell) \, k_{\xi_0}(\ell,y) \log\lr{\frac{k_{\changevar(\vartheta)}(\ell,y)}{k_{\xi_0}(\ell,y)}}\;\nu(\rmd y)}\tau_0 (\rmd \ell)\geq0  \;.
  $$
  Using that $\log_+(x)\leq x$ for all $x\geq0$, the positive part
  inside this double integral is at most $b_0(\ell) \,
  k_{\changevar(\vartheta)}(\ell,y)$, whose integral w.r.t. $\nu(\rmd
  y)\tau_{0}(\rmd \ell)$ is $\int b_0\;\rmd\tau_0<\infty$ by
  assumption.  Therefore, the same inequality holds when reversing the
  order of integration in the previous display. We get that the expression
  in~(\ref{eq:last-todo-positif-to-get-update-theta_LM}) is
  non-negative, which concludes the proof. 
\end{proof}
The following lemma is used to illustrate~\ref{asump:Upsilon-image}.
\begin{lem}
  \label{lem:C1cond}
  Consider an ELM family satisfying~\ref{asump:Upsilon-image}. Then, for all
  $\vartheta\in\tilde{\Tset}$, 
  $\tau\in\parammixproba$ and  $\tau$-a.e. $\ell$,
  $\ell\circ\changevar(\tilde{\Tset})$ is an open subset included in  $\intE_0$.
\end{lem}
\begin{proof}
   In~\ref{asump:Upsilon-image}, the linear operator $\ell\in\Zset$ being full rank implies that
the image of an open set is an open subset of $E$. Indeed, take 
  $F'\subset F$ a linear supplementary space of the kernel of $\ell$. Then $\ell_{|F'}$ is a bijective
  linear mapping from
  $F'$ to $E$ and thus, for any open set $O$ included in $F$,
  $\ell(O)=\ell(O\cap F')=\ell_{|F'}(O\cap F')$ is an open subset of
  $E$ since $O\cap F'$ is an open subset of $F'$ and $\ell_{|F'}$ is
  bicontinuous. Hence \ref{asump:Upsilon-image} implies that
$\ell\circ\changevar(\tilde{\Tset})$ is an open subset of $E$ for all
$\tau\in\parammixproba$ and  $\tau$-almost all $\ell\in \Zset$, which, by
Assumption~\ref{asump:gradA-well-def} from \Cref{defi:LMEFfam}, must
be included in $\intE_0$.
\end{proof}
We conclude with the following lemma, which relates
Assumption~\ref{item:condB-ident} to a kind of identifiability
property. 
\begin{lem}\label{lem:ident-elm}
  Consider an ELM family and let $\PP_{\vartheta,\tau}$ denote the
  probability of $(Y,L)$ defined on $\Yset\times\Zset$ as
  above. Then~\ref{item:condB-ident} is equivalent to have that for
  all
  $(\vartheta,\tau),(\vartheta',\tau')\in\tilde{\Tset}\times\parammixproba$,
  $\PP_{\vartheta,\tau}=\PP_{\vartheta',\tau'}$  if and only if
  $\changevar(\vartheta)=\changevar(\vartheta')$ and $\tau=\tau'$.
\end{lem}
\begin{proof}
  Let
  $(\vartheta,\tau),(\vartheta',\tau')\in\tilde{\Tset}\times\parammixproba$.  Note that $\changevar(\vartheta)=\changevar(\vartheta')$ and
  $\tau=\tau'$ is always sufficient to have $\PP_{\vartheta,\tau}=\PP_{\vartheta',\tau'}$.
  Also, since
  $\tau$ is the marginal over $\Zset$ of $\PP_{\vartheta,\tau}$ and $k^{(o)}(\ell\circ
  \changevar(\vartheta),\cdot)$ is the conditional density of $Y$
  given $L=\ell$ with respect to $\nu$, it is
  clear that $\PP_{\vartheta,\tau}=\PP_{\vartheta',\tau'}$ is
  equivalent to have $\tau=\tau'$ and
  \begin{equation}
    \label{eq:ident-ineterm-cond-density}
   k^{(o)}(\ell\circ
  \changevar(\vartheta),y)=k^{(o)}(\ell\circ
  \changevar(\vartheta'),y)\quad\text{for $\nu$-a.e. $y$ and
    $\tau$-a.e. $\ell$}\;.    
\end{equation}
By~\ref{hyp:expo}, which is assumed for an ELM family in
\Cref{defi:LMEFfam}, for all $\zeta,\zeta'\in E_0$, we have
$k^{(o)}(\zeta,y)=k^{(o)}(\zeta',y)$ for $\nu$-a.e. $y$ if and only if
$\zeta=\zeta'$. By~\ref{asump:gradA-well-def} in \Cref{defi:LMEFfam},
we have $\ell\circ \changevar(\vartheta)\in E_0$ for all
$\vartheta\in\tilde{\Tset}$, $\tau\in\parammixproba$ and
$\tau$-a.e. $\ell$. Thus, for all
$(\vartheta,\vartheta',\tau)\in\tilde{\Tset}^2\times\parammixproba$,
Eq.~(\ref{eq:ident-ineterm-cond-density}) is equivalent to have
  \begin{equation}
    \label{eq:ident-ineterm-cond-density-2}
\ell\circ
  \changevar(\vartheta)=\ell\circ
  \changevar(\vartheta')\quad\text{for
    $\tau$-a.e. $\ell$}\;.    
\end{equation}
  Hence, to prove the result, we only need to show
  that~\ref{item:condB-ident} is equivalent to have that, for
  all
  $(\vartheta,\vartheta',\tau)\in\tilde{\Tset}^2\times\parammixproba$,
  Eq.~(\ref{eq:ident-ineterm-cond-density-2}) implies
  $\changevar(\vartheta)=\changevar(\vartheta')$. In
  fact~\ref{item:condB-ident} exactly says that for all
  $\vartheta,\vartheta'\in\tilde{\Tset}$,
  $\changevar(\vartheta)\neq\changevar(\vartheta')$ implies the
  opposite of~(\ref{eq:ident-ineterm-cond-density-2}). 
\end{proof}

\subsubsection{Proof of \Cref{thm:argrmax-objective-equivalence}}
\label{sec:proof-crefthm:-objective-equiv}
The proof of \Cref{thm:argrmax-objective-equivalence} relies on a
general proposition (\Cref{prop:argrmax-objective-equivalence}) and
two technical lemmas 
(Lemmas \ref{lem:technical-exp-moment-sup-bound} and \ref{lem:gradAw-continuity-fubini}). The
following definition will be useful to state the proposition.
\begin{defi}[Directional continuity]
  Let $g$ be defined on an
  open subset $\mathcal{O}$ of the Euclidean space $F$ and valued in a topological
  space. Let
  $\xi\in\mathcal{O}$ and $x\in F$ such that $\normev[F]{x}=1$. We say
  that $g$ is continuous at $\xi$ in  direction $x$ if
  $\displaystyle
  \lim_{t\downarrow0} g(\xi+t x)=g(\xi)$.
\end{defi}

\begin{prop}\label{prop:argrmax-objective-equivalence}
    Consider an ELM family as in \Cref{defi:LMEFfam} and $\changevar$
  satisfying~\ref{asump:Upsilon-image}. Let $\check{\varphi}$ be a probability density function w.r.t. $\nu$. Let $\tau_0 \in\parammixproba$ and
  $\varphi:\Zset\times\Yset\to\rset_+$ be measurable function such
  that
\begin{align}
\label{eq:varphiQ:condS} &\int_{\Yset \times \Zset}
\varphi(\ell,y)\,\normop{\ell}\,\normev[E]{S(y)}\;\nu(\rmd
                           y)\tau_0 (\rmd \ell)<\infty \,\\
  \label{eq:varphiQ:cond_ell}
&\int_\Zset
\tilde{\varphi}(\ell)\,\normev[F]{\gradA{\ell}(\xi)}\;\tau_0 (\rmd
                                \ell)<\infty\qquad\text{for all $\xi\in\changevar(\tilde{\Tset})$}\;,
\end{align}
where we set $\varphi(\ell,y)$ and
$\tilde{\varphi}(\ell)$ as in~(\ref{eq:check-varphi-varphiL-carphiQ}) and~(\ref{eq:check-WW-well-defined-carphiQ}). Define
$\gradA{\mathrm{w}_{\tau_0}}:\changevar(\tilde{\Tset})\to F$ as in~(\ref{eq:gradA-WW-well-defined-carphiQ}).
Let $b\in\rset_+$ and $\xi_0$ in $\changevar(\tilde{\Tset})$ and set
$\mathbf{w}^*$ as in~(\ref{eq:wstar-defined-check-varphiQ}).
Then the  reciprocal set
$$
\Xi^*=\set{\xi\in\changevar(\tilde{\Tset})}{\gradA{\mathrm{w}_{\tau_0}}(\xi)=\mathbf{w}^*}
$$
is included in the set of solutions of the argmax problem
\begin{equation}
  \label{eq:argrmax-objective-equivalence:armax-prob}
  \argmax_{\xi\in\changevar(\tilde{\Tset})}\int_\Yset\lr{\int_\Zset
  \lrb{\varphi(\ell,y) + b\,\tilde{\varphi}(\ell)\,k^{(o)}(\ell(\xi_0),y)}\,
  \log
  \lr{\frac{k^{(o)}(\ell(\xi),y)}{k^{(o)}(\ell(\xi_0),y)
    }}\; \tau_0 (\rmd \ell)}\nu(\rmd y) \;.
\end{equation}
Moreover the two following assertions hold.
\begin{enumerate}[label=(\roman*)]
\item\label{item:argrmax-objective-equivalence:1} If
  $\gradA{\mathrm{w}}_{\tau_0}$ is continuous on
$\changevar(\tilde{\Tset})$ in every direction, the opposite inclusion is true, namely, the
solution set of the argmax
problem~(\ref{eq:argrmax-objective-equivalence:armax-prob}) is
included in  $\Xi^*$.
\item\label{item:argrmax-objective-equivalence:2} If we have, for all  $\xi\neq\xi'\in\changevar(\tilde{\Tset})$,
  \begin{equation}
  \label{eq:cond-strict-ineq-varphi-L}
\int_{\set{\ell\in\Zset}{\ell(\xi-\xi')\neq0}}
\tilde{\varphi}(\ell) \;\tau_0 (\rmd \ell)>0 \;,
\end{equation}
then $\xi\mapsto \gradA{\mathrm{w}}_{\tau_0}(\xi)$ is one-to-one on
$\changevar(\tilde{\Tset})$ and thus $\Xi^*$ is either
a singleton or empty. 
\end{enumerate}
\end{prop}
\begin{proof}
  In this proof, we take $b\geq0$ and $\xi_0\in\changevar(\tilde{\Tset})$, and
  we  denote the integral defining the objective function in the
  argmax problem~(\ref{eq:argrmax-objective-equivalence:armax-prob}) by
  $\mathcal{C}(\xi)$.

  A first remark is that, by~(\ref{eq:varphiQ:condS}),
  $\int \varphi(\ell,y)\,\ell^T\circ S(y) \;\nu(\rmd y)\tau_0 (\rmd \ell)$ is well-defined in $F$ and $\gradA{\mathrm{w}}_{\tau_0}(\xi)$   in~(\ref{eq:gradA-WW-well-defined-carphiQ}) is well-defined in $F$ for all $\xi\in\changevar(\tilde{\Tset})$ as a consequence of~(\ref{eq:varphiQ:cond_ell}). In particular we get that $\mathbf{w}^*$ in~(\ref{eq:wstar-defined-check-varphiQ}) is well-defined
  in $F$ (but does not necessarily belong to $w_{\tau_0}^\nabla \circ\changevar(\tilde{\Tset})$). 

  Next we show that the objective function $\mathcal{C}(\xi)$ is well
  defined and finite for all $\xi\in\changevar(\tilde{\Tset})$ and
  that Fubini's theorem applies (so that the double integral
  in~(\ref{eq:argrmax-objective-equivalence:armax-prob}) can be
  computed in the opposite order or as a single integral over
  $\Yset\times\Zset$ w.r.t. the product measure $\nu\otimes\tau_0$). To
  this end, we show successively that
\begin{align}
  \label{eq:argrmax-objective-well-defined-carphiQ}
&  \int_{\Yset\times\Zset}
\varphi(\ell,y)\,
  \lrav{\log
    \lr{\frac{k^{(o)}(\ell(\xi),y)}{k^{(o)}(\ell(\xi_0),y)
                                                      }}}\; \nu(\rmd y)\tau_0 (\rmd \ell) <\infty\;,\\
  \label{eq:argrmax-objective-well-defined-carphiQ-2bis}
& \int_{\Yset\times\Zset} \tilde{\varphi}(\ell)\,k^{(o)}(\ell(\xi_0),y)\,
  \log_+
  \lr{\frac{k^{(o)}(\ell(\xi),y)}{k^{(o)}(\ell(\xi_0),y)
    }}\; \nu(\rmd y)\tau_0 (\rmd \ell) < \infty \;,\\
  \label{eq:argrmax-objective-well-defined-carphiQ-2}
& -\infty < \int_\Zset \tilde{\varphi}(\ell)\,\lr{\int_\Yset\,k^{(o)}(\ell(\xi_0),y)\,
  \log
  \lr{\frac{k^{(o)}(\ell(\xi),y)}{k^{(o)}(\ell(\xi_0),y)
                                                        }}\; \nu(\rmd y)}\tau_0 (\rmd \ell) \leq 0\;,\\
  &
  \label{eq:argrmax-objective-well-defined-carphiQ-2ter}
\int_{\Yset\times\Zset} \tilde{\varphi}(\ell)\,k^{(o)}(\ell(\xi_0),y)\,
  \lrav{\log
  \lr{\frac{k^{(o)}(\ell(\xi),y)}{k^{(o)}(\ell(\xi_0),y)
    }}}\; \nu(\rmd y)\tau_0 (\rmd \ell) < \infty
\;.
\end{align}
  Let us now show
  that~(\ref{eq:argrmax-objective-well-defined-carphiQ}) holds. Let
  $\ell\in\Zset$ such that
  $\ell\circ\changevar(\tilde{\Tset})\subseteq\intE_0$, which by
  Condition~\ref{asump:gradA-well-def} and~\ref{asump:Upsilon-image},
  is true for $\tau_0$-a.e. $\ell\in\Zset$.
  Applying~(\ref{eq:basin-expo-family-convex-ineq}) in \Cref{lem:useful-1} with
  $(\zeta,\zeta')=(\ell(\xi_0),\ell(\xi))$ as well as $(\zeta,\zeta')=(\ell(\xi),\ell(\xi_0))$ and using the definition of $\ell^T$
  in~(\ref{eq:defajoint}) and that of
  $\gradA{\ell}$ in~(\ref{eq:Q:gradA-def}), we get that
  $$
\pscal[F]{\ell^T\circ S(y)-\gradA{\ell}(\xi)}{\xi-\xi_0} \leq \log
\lr{\frac{k^{(o)}(\ell(\xi),y)}{k^{(o)}(\ell(\xi_0),y)
  }}\leq \pscal[F]{\ell^T\circ S(y)-\gradA{\ell}(\xi_0)}{\xi-\xi_0}\;.
$$
It follows that, for all $\xi_0,\xi\in\changevar(\tilde{\Tset})$,
\begin{equation}
  \label{eq:argrmax-objective-well-defined-carphiQ-0}
\lrav{\log\lr{\frac{k^{(o)}(\ell(\xi),y)}{k^{(o)}(\ell(\xi_0),y) }}}\leq
\lr{\normev[F]{\ell^T\circ S(y)}+\max (\normev[F]{\gradA{\ell}(\xi_0)},\normev[F]{\gradA{\ell}(\xi)})}\,\normev[F]{\xi-\xi_0}\;.
\end{equation}
Using $\normev[F]{\ell^T\circ
  S(y)}\leq\normop{\ell^T}\normev[E]{S(y)}=\normop{\ell}\normev[E]{S(y)}$,~(\ref{eq:varphiQ:condS})
and~(\ref{eq:varphiQ:cond_ell}), we
obtain~(\ref{eq:argrmax-objective-well-defined-carphiQ}).

To show~(\ref{eq:argrmax-objective-well-defined-carphiQ-2bis}), we use
the bound $\log_+(x)\leq x$ for all $x\geq0$ so that the left-hand
side of~(\ref{eq:argrmax-objective-well-defined-carphiQ-2bis}) is upper-bounded by 
$\int_{\Yset\times\Zset} \tilde{\varphi}(\ell)\,k^{(o)}(\ell(\xi),y)\;
\nu(\rmd y)\tau_0 (\rmd \ell)=\int_{\Zset}
\tilde{\varphi}(\ell)\;\tau_0 (\rmd \ell)=1$.

Next, to show~(\ref{eq:argrmax-objective-well-defined-carphiQ-2}), we observe that, as
a consequence
of~(\ref{eq:kullback-exponential-canonical}),~(\ref{eq:defajoint})
and~(\ref{eq:Q:gradA-def}), we have 
\begin{align*}
  0\geq \int_\Yset
  k^{(o)}(\ell(\xi),y)\,
  \log
  \lr{\frac{k^{(o)}(\ell(\xi),y)}{k^{(o)}(\ell(\xi_0),y)
    }}\; \nu(\rmd y)
  &\geq\pscal[F]{\gradA{\ell}(\xi_0)-\gradA{\ell}(\xi)}{\xi-\xi_0}\\
  &\geq
  -\normev[F]{\gradA{\ell}(\xi_0)-\gradA{\ell}(\xi)}\,\normev[F]{\xi-\xi_0}\;.
\end{align*}
Hence, using~(\ref{eq:varphiQ:cond_ell}), we
get~(\ref{eq:argrmax-objective-well-defined-carphiQ-2}). 

Finally, Tonelli's theorem gives us that the same positive part as
in~(\ref{eq:argrmax-objective-well-defined-carphiQ-2bis}) has a finite
integral when integrating in the same order as
in~(\ref{eq:argrmax-objective-well-defined-carphiQ-2}), so
that~(\ref{eq:argrmax-objective-well-defined-carphiQ-2bis})
and~(\ref{eq:argrmax-objective-well-defined-carphiQ-2})
imply~(\ref{eq:argrmax-objective-well-defined-carphiQ-2ter}).

We now prove the inclusion of $\Xi^*$ in the set of argmax solutions. 
Since Fubini's theorem applies in the definition of the objective
function, we can write, for all $\xi,\xi'\in\changevar(\tilde{\Tset})$,
$$
\mathcal{C}(\xi')-\mathcal{C}(\xi)
=
\int_{\Yset\times\Zset}
  \lrb{\varphi(\ell,y) + b\,\tilde{\varphi}(\ell)\,k^{(o)}(\ell(\xi_0),y)}\,
  \log
  \lr{\frac{k^{(o)}(\ell(\xi'),y)}{k^{(o)}(\ell(\xi),y)
    }}\; \nu(\rmd y)\tau_0 (\rmd \ell)\;.
$$
Using~(\ref{eq:basin-expo-family-convex-ineq}) in \Cref{lem:useful-1} this time with $\zeta=\ell(\xi)$
and $\zeta'=\ell(\xi')$, we obtain that
$$
\mathcal{C}(\xi')-\mathcal{C}(\xi)
\geq\int_{\Yset\times\Zset}
  \lrb{\varphi(\ell,y) + b\,\tilde{\varphi}(\ell) k^{(o)}(\ell(\xi_0),y)}
\pscal[E]{\ell^T\circ S(y)-\gradA{\ell}(\xi')}{\xi'-\xi}
  \nu(\rmd y)\tau_0 (\rmd \ell).
$$
  Using~(\ref{eq:varphiQ:condS}),~(\ref{eq:varphiQ:cond_ell}) to
  integrate in any convenient order, we have 
\begin{multline*}
\int_{\Yset\times\Zset}
\varphi(\ell,y)\,
\pscal[E]{\ell^T\circ S(y)-\gradA{\ell}(\xi')}{\xi'-\xi}  \; \nu(\rmd y)\tau_0 (\rmd \ell)\\ =
\pscal[E]{\int_{\Yset\times \Zset}
\varphi(\ell,y)\,
\ell^T\circ S(y)\; \nu(\rmd y)\tau_0 (\rmd \ell)-\gradA{\mathrm{w}}_{\tau_0}(\xi')}{\xi'-\xi}
  \;,
\end{multline*}
where we have used the definition of $\gradA{\mathrm{w}}_{\tau_0}$ in (\ref{eq:gradA-WW-well-defined-carphiQ}), and
\begin{multline*}
  \int_{\Yset\times\Zset}
\tilde{\varphi}(\ell)\,k^{(o)}(\ell(\xi_0),y)\,
\pscal[E]{\ell^T\circ S(y)-\gradA{\ell}(\xi')}{\xi'-\xi}  \; \nu(\rmd
    y)\tau_0 (\rmd \ell) \\
 = \pscal[E]{\gradA{\mathrm{w}}_{\tau_0}(\xi_0)-\gradA{\mathrm{w}}_{\tau_0}(\xi')}{\xi'-\xi}\;.
\end{multline*}
where we have used once again the definition of
$\gradA{\mathrm{w}}_{\tau_0}$ in
(\ref{eq:gradA-WW-well-defined-carphiQ}) as well as the expression of
$\nabla A$ given in (\ref{eq:partialA:expo-family}).  Plugging these
two equalities in the previous inequality, we finally get that, for
all $\xi,\xi'\in\changevar(\tilde{\Tset})$,
\begin{align}
  \nonumber
\mathcal{C}(\xi')-\mathcal{C}(\xi)
&\geq
\pscal[E]{\int_{\Yset \times \Zset}
\varphi(\ell,y)\,
\ell^T\circ S(y)\; \nu(\rmd y)\tau_0 (\rmd
                                     \ell)-(1+b)\,\gradA{\mathrm{w}}_{\tau_0}(\xi')+b\,\gradA{\mathrm{w}}_{\tau_0} (\xi_0)}{\xi'-\xi}\\
  \label{eq:charac-argmax-varphi-Q}
&  =(1+b)\;\pscal[E]{\mathbf{w}^*-\gradA{\mathrm{w}}_{\tau_0}(\xi')}{\xi'-\xi}\;,
\end{align}
Hence any $\xi'\in\Xi^*$ must be a a solution of the
argmax~(\ref{eq:argrmax-objective-equivalence:armax-prob}). 

Let us now prove
Assertion~\ref{item:argrmax-objective-equivalence:1}. We use
again~(\ref{eq:charac-argmax-varphi-Q}) this time with $\xi$ in the
set of the argmax solutions and $\xi'$ in the neighborhood of $\xi$
(remember that $\changevar(\tilde{\Tset})$ is assumed to be open) such that
$\xi'-\xi=\epsilon\lr{\mathbf{w}^*-\gradA{\mathrm{w}}_{\tau_0}(\xi)}$ with $\epsilon>0$
small enough. We obtain that
$$
\pscal[E]{\mathbf{w}^*-\gradA{\mathrm{w}}_{\tau_0}(\xi')}{\mathbf{w}^*-\gradA{\mathrm{w}}_{\tau_0}(\xi)}\leq 0\;.
$$
The
assumption that $\gradA{\mathrm{w}}_{\tau_0}$ is continuous at $\xi$ in
direction $\lr{\mathbf{w}^*-\gradA{\mathrm{w}}_{\tau_0}(\xi)}$ gives us that
$\mathbf{w}^*=\gradA{\mathrm{w}}_{\tau_0}(\xi)$, that is, $\xi\in\Xi^*$ and the
opposite inclusion is shown.

We conclude with the proof of
Assertion~\ref{item:argrmax-objective-equivalence:2}. By~(\ref{eq:gradA-WW-well-defined-carphiQ}),
for all
$\xi,\xi'\in\changevar(\tilde{\Tset})$, we can write
$$
\gradA{\mathrm{w}}_{\tau_0}(\xi)-\gradA{\mathrm{w}}_{\tau_0}(\xi')=\int_\Zset \tilde{\varphi}(\ell)\,\lr{\gradA{\ell}(\xi)-\gradA{\ell}(\xi')}\;\tau_0 (\rmd \ell)\;.
$$
Since $\nabla A$ has a positive definite Jacobian on the convex set $\intE_0$, we
have, for all $\zeta\neq\zeta'\in\intE_0$,
$\pscal[E]{\nabla A(\zeta)-\nabla A(\zeta')}{\zeta-\zeta'}>0$.

Letting
$\xi \neq \xi'\in\changevar(\tilde{\Tset})$ and $\ell\in\Zset$ be such
that $\ell(\xi),\ell(\xi')\in\intE_0$ and $\ell(\xi-\xi')\neq0$, we
can apply the previous inequality with $\zeta=\ell(\xi)$ and
$\zeta'=\ell(\xi')$ leading to
$$
\pscal[F]{\gradA{\ell}(\xi)-\gradA{\ell}(\xi')}{\xi-\xi'}>0\;.
$$
Combining this result with the condition~(\ref{eq:cond-strict-ineq-varphi-L}) gives us that for all $\xi\neq\xi'\in\changevar(\tilde{\Tset})$
$$\pscal[F]{\gradA{\mathrm{w}}_{\tau_0}(\xi)-\gradA{\mathrm{w}}_{\tau_0}(\xi')}{\xi-\xi'}>0\;.$$ In
particular, $\gradA{\mathrm{w}}_{\tau_0}$ is one-to-one and the proof is concluded.
\end{proof}
Let $d_F$ denote the dimension of $F$ and let us now introduce and prove two key lemmas.

\begin{lem}\label{lem:technical-exp-moment-sup-bound}
  Let $\xi\in\changevar(\tilde{\Tset})$ and let $x\in F$ be such that
  $\normev[F]{x}=1$. Let us denote, for any $\epsilon>0$,
  \begin{align*}
  &B_F(\xi,\epsilon)=\set{\xi'\in F}{\normev[F]{\xi-\xi'}<\epsilon}\;,\\
  &  D_x(\xi,\epsilon)=\set{\xi+t\,\epsilon\,x}{0\leq t<1}\;.
  \end{align*}
  Let $\ell\in\Zset$ be such that $\ell\circ\changevar(\tilde{\Tset})\subset\intE_0$.
  Then the two following assertions hold.

  \begin{enumerateList}
  \item \label{item:technical-exp-moment-sup-bound1}   For any $0<\epsilon\leq\epsilon'<\epsilon''$ such that
  $D_x(\xi,\epsilon'')\subset\changevar(\tilde{\Tset})$, we have
  \begin{align}\label{eq:diffAL}
    \sup_{\xi'\in D_x(\xi,\epsilon)}\lrav{A\circ \ell(\xi')-A\circ
      \ell(\xi)}\leq\epsilon\,\max\lr{\normev[F]{\gradA{\ell}(\xi)},\normev[F]{\gradA{\ell}(\xi+\epsilon'
      x)}}\;.
  \end{align}
  \item \label{item:technical-exp-moment-sup-bound2} For any $\epsilon>0$ such that
  $B_F\lr{\xi,2\lr{1+\sqrt{d_F}}\epsilon}\subset\changevar(\tilde{\Tset})$,
  we have, for all $\epsilon_0\in(0,\epsilon]$,
  \begin{align}
    \label{eq:sup-k-is-sup-gradA}
    &  \int_\Yset \rme^{\epsilon_0\normev[F]{\ell^T\circ S(y)}}\,
  \lr{\sup_{\xi'\in D_x(\xi,\epsilon_0)}k^{(o)}(\ell(\xi'),y)}\;\nu(\rmd y)\leq
  2\,d_F\,\max_{
    \xi'\in \Delta_{x,\epsilon}(\xi)}\lr{\rme^{\epsilon_0'\normev[F]{\gradA{\ell}(\xi')}}} \;,
\end{align}
where $\epsilon_0'=\lr{1+2\sqrt{d_F}}\epsilon_0$ and
$\Delta_{x,\epsilon}(\xi)$ is a set of $2(d_F+1)$ points in
$B_x(\xi,2\lr{1+\sqrt{d_F}}\epsilon)$
only depending on $\xi$, $x$ and $\epsilon$.
\end{enumerateList}
\end{lem}
\begin{proof} We prove
  Assertions~\ref{item:technical-exp-moment-sup-bound1}
  and~\ref{item:technical-exp-moment-sup-bound2} successively.

  \noindent\textbf{Proof of Assertion~\ref{item:technical-exp-moment-sup-bound1}} Let $0<\epsilon\leq\epsilon'<\epsilon''$ be such that
  $D_x(\xi,\epsilon'')\subset\changevar(\tilde{\Tset})$. First note that $t\mapsto A\circ \ell(\xi+t\,x)$ has its
  first and second derivative on $t\in[0,\epsilon']$ given respectively by
  \begin{align*}
  &h_1 : t\mapsto\pscal[E]{\nabla A\circ \ell(\xi+t\,x)}{\ell(x)} = \pscal[F]{\gradA{\ell} (\xi+t\,x)}{x} \\
  &h_2 : t\mapsto\pscal[E]{\lrb{\nabla\nabla^T A\circ
      \ell(\xi+t\,x)}\ell(x)}{\ell(x)} \eqsp.
  \end{align*}
  Since $\nabla\nabla^T A$ is positive definite on $\intE_0$, $h_2$ is non-negative, thus $h_1$ is monotonous
  on $[0,\epsilon']$ and its maximal absolute value is obtained at $t=0$ or $t=\epsilon'$, meaning that $|h_1(t)| \leq \max(|h_1(0)|, |h_1(\epsilon')|)$ for all $t \in [0,\epsilon']$. Therefore, the bound
  in~(\ref{eq:diffAL}) follows from the mean value theorem        
  combined with the fact that: for all $t \in [0,\epsilon']$ and $\xi'\in D_x(\xi,\epsilon)$,
      \begin{align*}
      & |h_1(t)| \leq \normev[F]{\gradA{\ell}(\xi+t
      x)} \normev[F]{x} = \normev[F]{\gradA{\ell}(\xi+t
      x)}\\
      & \normev[F]{\xi' - \xi} \leq \epsilon \normev[F]{x} \leq \epsilon
      \end{align*}
     since $\normev[F]{x}=1$ by assumption.

       \noindent\textbf{Proof of
           Assertion~\ref{item:technical-exp-moment-sup-bound2}} Take $0<\epsilon_0\leq\epsilon$
  such that
  $B_x(\xi,2\lr{1+\sqrt{d_F}}\epsilon)\subset\changevar(\tilde{\Tset})$, let
   $\epsilon_0'$ be as stated and set $\epsilon'=\lr{1+2\sqrt{d_F}}\epsilon$.  Let
  $x_1,\dots,x_{d_F}$ be an orthonormal basis of $F$ so that, for any
  $w\in E$, we have
  \begin{align}
    \nonumber
    \normev[F]{\ell^T(w)}=\lr{\sum_{k=1}^{d_F}\lrav{\pscal[F]{\ell^T(w)}{x_k}}^2}^{1/2}
    &\leq\sqrt{d_F}\max_{k,s}\lr{s\,\pscal[F]{\ell^T(w)}{x_k}}\\
     \label{eq:boundOrthonormalBasis}
    &=\max_{k,s}\pscal[E]{w}{\ell(s\,\sqrt{d_F}\,x_k)}\;,
  \end{align}
  where, here and in the following, $\max_{k,s}$ is the maximum over $s\in \{\pm1 \}$ and $k=1 \dots d_F$. 
    By definition of $k^{(o)}$ in \Cref{def:expo-family}, we also have: for all $\xi'\in
    D_x(\xi,\epsilon_0)$,
    \begin{align*}
      \rme^{\epsilon_0\normev[F]{\ell^T\circ S(y)}}\,k^{(o)}(\ell(\xi'),y) &= h(y) ~
       \rme^{\epsilon_0\normev[F]{\ell^T\circ S(y)}+\pscal[E]{S(y)}{\ell(\xi')} - A\circ \ell(\xi')}
      \\
      &= h(y) ~
       \rme^{\epsilon_0\normev[F]{\ell^T\circ S(y)}+\pscal[E]{S(y)}{\ell(\xi'-\xi)}+\pscal[E]{S(y)}{\ell(\xi)}- A\circ \ell(\xi')}
    \end{align*}
    Now using that 
    $$
    |\pscal[E]{S(y)}{\ell(\xi'-\xi)}| = |\pscal[F]{\ell^T \circ S(y)}{\xi'-\xi}| \leq \epsilon_0\normev[F]{\ell^T\circ S(y)}
    $$
    (since $\normev[F]{x} = 1$ and $\xi'\in
    D_x(\xi,\epsilon_0)$) paired up with \eqref{eq:boundOrthonormalBasis}, we can deduce that 
    \begin{align}
      \rme^{\epsilon_0\normev[F]{\ell^T\circ S(y)}}\,k^{(o)}(\ell(\xi'),y) &\leq h(y) ~
       \rme^{2\epsilon_0\normev[F]{\ell^T\circ S(y)}+\pscal[E]{S(y)}{\ell(\xi)}- A\circ \ell(\xi')}
        \nonumber \\                                             
   & \leq                   \max_{k,s}h(y) \, \exp\lr{\pscal[E]{S(y)}{\ell(\xi+2\,s\,\sqrt{d_F}\, \epsilon_0\,x_k)}-A\circ
                                                                  \ell(\xi')} \nonumber \\
      &\leq \max_{k,s}k^{(o)}(\ell(\xi''_{s,k}),y)\,
        \rme^{\lrav{A\circ \ell\lr{\xi''_{s,k}}- A\circ \ell(\xi')}}\;, \label{eq:interMaxks}
    \end{align}
    where $\xi''_{s,k}=\xi+2\,s\,\sqrt{d_F}\,\epsilon_0\,x_k$.
Note that we took $\epsilon$ small enough so that for all $s,k$, 
$$
D_x(\xi,\epsilon'+\epsilon)\cup
D_{s\,x_k}\lr{\xi,\epsilon'+\epsilon}\subset B_F\lr{\xi,2\lr{1+\sqrt{d_F}}\epsilon}
\subset\changevar(\tilde{\Tset})\;.
  $$
 $\xi'$ and $\xi''_{s,k}$ are in
$D_x(\xi,\epsilon_0)$ and $D_{s\,x_k}\lr{\xi,2\sqrt{d_F}\epsilon_0}$,
respectively. Using~(\ref{eq:diffAL}) of
Assertion~\ref{item:technical-exp-moment-sup-bound1} twice with $\epsilon\leq
\epsilon'<\epsilon''$ successively  replaced by
$2\sqrt{d_F}\epsilon_0\leq\epsilon'\leq 2\lr{1+\sqrt{d_F}}\epsilon$
and $\epsilon_0\leq\epsilon'\leq
2\lr{1+\sqrt{d_F}}\epsilon$, we get that
    \begin{align*}
    \lrav{A\circ \ell\lr{\xi''_{s,k}}- A\circ \ell(\xi')}&
                                                           \leq
                                                           \lrav{A\circ
                                                           \ell\lr{\xi''_{s,k}}- A\circ \ell(\xi)}+\lrav{A\circ
                                                           \ell\lr{\xi}- A\circ \ell(\xi')}\\
                                                         &
                                                           \leq
                                                           2\sqrt{d_F}\epsilon_0\,
                                                           \max\lr{\normev[F]{\gradA{\ell}(\xi)},\normev[F]{\gradA{\ell}(\xi+\epsilon'\,s\,x_k)}}\\
    &\phantom{\leq}+
                                                           \epsilon_0\max\lr{\normev[F]{\gradA{\ell}(\xi)},\normev[F]{\gradA{\ell}(\xi+\epsilon'\,x)}}\\
                                                         &\leq\epsilon_0'\, C_{\xi,x}(\ell)\;,
  \end{align*}
  where the last line follows by setting
  $$
  C_{\xi,x}(\ell)=
  \max\set{\normev[F]{\gradA{\ell}(\xi+\epsilon'\,x')}}{x'\in \{0,\,x\,,\,s\,x_{k}\}\,,\,\text{with}\,s\in \{\pm1\}\,,\,k=1\dots d_F}\;.
  $$
  Combining this inequality with \eqref{eq:interMaxks} then yields : for all $\xi'\in
  D_x(\xi,\epsilon_0)$,
  \begin{align*}
    \rme^{\epsilon_0\normev[F]{\ell^T\circ S(y)}}\,k^{(o)}(\ell(\xi'),y) & \leq \rme^{\epsilon_0'\,C_{\xi,x}(\ell)} \max_{k,s}k^{(o)}(\ell(\xi''_{s,k}),y) \\
    & \leq \rme^{\epsilon_0'\,C_{\xi,x}(\ell)}
    \,\sum_{k,s} k^{(o)}(\ell(\xi''_{s,k}),y)\;.
  \end{align*}
Observe that the right-hand side of the previous display does not
depend on $\xi'$ and that the summands in the sum (which has $2d_F$
terms) all have integrals equal w.r.t. $\nu$ equal to 1. The
bound~(\ref{eq:sup-k-is-sup-gradA}) follows and
Assertion~\ref{item:technical-exp-moment-sup-bound2} is proved.
\end{proof}
\begin{lem}\label{lem:gradAw-continuity-fubini}
  Consider an ELM family as in \Cref{defi:LMEFfam}. Assume~\ref{asump:Upsilon-image} and~\ref{item:condB-stilde-ratio}.  Let
  $\vartheta_0\in\tilde{\Tset}$, $\tau_0\in\parammixproba$ and
  $\check{\varphi}:\Yset\to\rset_+$ be a density function such that  
    \begin{equation}
    \label{eq:cond-checkVarphi-linear-mixture-left}
    \int_\Yset \check{\varphi}(y)\,\rme^{\tilde{\mathrm{s}}_{\vartheta_0,\tau_0}(y)} \;\nu(\rmd y) <\infty
    \;. 
  \end{equation}
  Define $\varphi(\ell,y)$, $\tilde{\varphi}(\ell)$ and
  $\gradA{\mathrm{w}}_{\tau_0}$ as
  in~(\ref{eq:check-varphi-varphiL-carphiQ}),~(\ref{eq:check-WW-well-defined-carphiQ})
  and~(\ref{eq:gradA-WW-well-defined-carphiQ}).  Then for all
  $\xi\in\changevar(\tilde{\Tset})$ and $x\in F$ such that
  $\normev[F]{x}=1$, there exists $\epsilon_0>0$ such that
\begin{align}\label{eq:cond-gradA-w-cont-y-wise}
  \int_{\Yset \times \Zset} \normev[F]{\ell^T\circ S(y)}\sup_{\xi'\in D_x(\xi,\epsilon_0)} \lr{k^{(o)}(\ell(\xi'),y)}
  \,\tilde{\varphi}(\ell)\;\nu(\rmd y)\tau_0 (\rmd \ell) < \infty\;.
\end{align}
Moreover, $\gradA{\mathrm{w}}_{\tau_0}$ is continuous on
$\changevar(\tilde{\Tset})$ in every direction, and, for all
$\vartheta\in\tilde{\Tset}$, it holds that
\begin{align}\label{eq:cond-gradA-w-cont-y-wise-expr}
  \gradA{\mathrm{w}}_{\tau_0}\circ\changevar(\vartheta)
  = \PE_{\vartheta,\tau_0}\lrb{L^T\circ S(Y)\,\tilde{\varphi}(L)}\;.
\end{align}
\end{lem}
\begin{proof}
Let  $\ell \in \Zset$ is such that $\ell \circ \changevar(\tilde{\Tset}) \subset\intE_0$, $\xi\in\changevar(\tilde{\Tset})$ and $x\in F$ such that
$\normev[F]{x}=1$. Then for $\epsilon>0$ small enough, we have
$B_x(\xi,2\lr{1+\sqrt{d_F}}\epsilon)\subset\changevar(\tilde{\Tset})$,
and for all $\epsilon_0\in(0,\epsilon]$,
\begin{align*}
                 \int_\Yset
                 \rme^{\epsilon_0\normev[F]{\ell^T\circ S(y)}}
                 \sup_{\xi'\in D_x(\xi,\epsilon_0)} \lr{k^{(o)}(\ell(\xi'),y)}
                 \;\nu(\rmd y) 
  &\leq   
    2\,d_F\,\sum_{
    \xi'\in \Delta}\rme^{\epsilon_0'\normev[F]{\gradA{\ell}(\xi')}} \;,
\end{align*}
where we used Assertion~\ref{item:technical-exp-moment-sup-bound2} of \Cref{lem:technical-exp-moment-sup-bound} with $\epsilon'=\lr{1+2\sqrt{d_F}}\epsilon$, $\epsilon_0'=\lr{1+2\sqrt{d_F}}\epsilon_0$ and
$\Delta$ is a set of $2(d_F+1)$ points in
$B_x(\xi,2\lr{1+\sqrt{d_F}}\epsilon)$
only depending on $\xi$, $x$ and $\epsilon$. 
Since $\normev[F]{\ell^T\circ S(y)}\leq\frac1{\epsilon_0}\rme^{\epsilon_0\normev[F]{\ell^T\circ S(y)}}$,
we get that: 
\begin{align*}
G_{\xi,\epsilon_0}(\ell) &\eqdef
  \int_\Yset \normev[F]{\ell^T\circ S(y)}\sup_{\xi'\in D_x(\xi,\epsilon_0)} \lr{k^{(o)}(\ell(\xi'),y)}
                        \;\nu(\rmd y)\\
  &\leq \frac{2\,d_F}{\epsilon_0}\,\sum_{
    \xi'\in \Delta} \rme^{\epsilon_0'\normev[F]{\gradA{\ell}(\xi')}}\;.
\end{align*}
In particular by~\ref{asump:Upsilon-image}, with \Cref{lem:C1cond}, this
bounds hold for $\tau_0$-almost all $\ell \in \Zset$. 
By~\ref{item:condB-stilde-ratio}, since $\Delta$ is a finite set only
depending on $\xi$, $x$ and $\epsilon$ and 
satisfying $\Delta\subset
B_x(\xi,2\lr{1+\sqrt{d_F}}\epsilon)\subset\changevar(\tilde{\Tset})$,
there exist  $\epsilon_\Delta,C_\Delta>0$ only
depending on $\xi$, $x$ and $\epsilon$ such that, for all $\xi'\in\Delta$,
$$
\PE_{\vartheta_0,\tau_0}\argcond{\rme^{\epsilon_\Delta\normev[F]{\gradA{L}(\xi')}}}{Y}\leq
C_\Delta\,\rme^{\tilde{\mathrm{s}}_{\vartheta_0,\tau_0}(Y)}
\quad\PP_{\vartheta_0,\tau_0}-\as
$$
Since $\Delta$ has cardinal $2(d_F+1)$, taking
$\epsilon_0\in(0,\epsilon]$ small enough to have
$\epsilon'_0\leq\epsilon_\Delta$, this gives us that 
$$
g_{\xi,\epsilon_0}(Y)\eqdef \PE_{\vartheta_0,\tau_0}\argcond{G_{\xi,\epsilon_0}(L)}{Y}\leq \frac{4\,C_\Delta\,d_F(d_F+1)}{\epsilon_0}\rme^{\tilde{\mathrm{s}}_{\vartheta_0,\tau_0}(Y)}
\quad\PP_{\vartheta_0,\tau_0}-\as
$$
Now observing that the ratio in~(\ref{eq:check-varphi-varphiL-carphiQ})
is the conditional density of $L$ applied to $\ell$ given $Y=y$ under $\PP_{\vartheta_0,\tau_0}$
and that the marginal distribution of  $\PP_{\vartheta_0,\tau_0}$ on $\Yset$ is equivalent to $\nu$, we obtain that
\begin{align*}
\int_\Zset G_{\xi,\epsilon_0}(\ell)\,\tilde{\varphi}(\ell)\;\tau_0 (\rmd
  \ell)&=\int_{\Zset\times\Yset} G_{\xi,\epsilon_0}(\ell)\,\varphi(\ell,y)\;\tau_0 (\rmd
  \ell)\nu(\rmd y)\\
    &=\int_\Yset g_{\xi,\epsilon_0}(y)\,\check{\varphi}(y)\;\nu(\rmd y)\\
    &\leq  \frac{4\,C_\Delta\,d_F(d_F+1)}{\epsilon_0}\int_\Yset \rme^{\tilde{\mathrm{s}}_{\vartheta_0,\tau_0}(y)}\,\check{\varphi}(y)\;\nu(\rmd y)
      \;,
\end{align*}
Hence, using condition~(\ref{eq:cond-checkVarphi-linear-mixture-left}), we deduce that 
\begin{align*}
  \int_\Zset G_{\xi,\epsilon_0}(\ell)\,\tilde{\varphi}(\ell)\;\tau_0 (\rmd
    \ell) < \infty \eqsp,
  \end{align*}
which is exactly the claimed bound \eqref{eq:cond-gradA-w-cont-y-wise}. This bound gives in particular that
$$
\xi'\mapsto \int_{\Yset \times \Zset} \ell^T\circ S(y) \, k^{(o)}(\ell(\xi'),y)
\,\tilde{\varphi}(\ell)\;\nu(\rmd y)\tau_0 (\rmd \ell)
$$
is well-defined in $F$ for $\xi'\in D_x(\xi,\epsilon_0)$, and
continuous at $\xi$ in direction $x$. Moreover, the Fubini theorem
applies so that (i) integrating first with respect to $y$ and (ii) using the expression of $\nabla A$ given in \eqref{eq:partialA:expo-family}, we obtain that this mapping is in fact $\xi'\mapsto\gradA{\mathrm{w}}_{\tau_0}(\xi')$. This being true for all
$\xi\in\changevar(\tilde{\Tset})$ and $x\in F$ such that
$\normev[F]{x}=1$, we get the last claim by observing that,
by~(\ref{eq:DefiKernelLinMixt-PE}), the expectation
in~(\ref{eq:cond-gradA-w-cont-y-wise-expr}) is another way to express
the same integral in the case where $\xi=\changevar(\vartheta)$.
\end{proof}

We can now prove \Cref{thm:argrmax-objective-equivalence}. \newline

\begin{proof}[Proof of \Cref{thm:argrmax-objective-equivalence}]
  We use \Cref{prop:argrmax-objective-equivalence} with $\varphi(\ell,y)$
  given by~(\ref{eq:check-varphi-varphiL-carphiQ}). We first
  prove that Conditions~(\ref{eq:varphiQ:condS}) and~(\ref{eq:varphiQ:cond_ell})
  hold in this setting, under the assumptions of \Cref{thm:argrmax-objective-equivalence}. 

  We use again that the ratio in~(\ref{eq:check-varphi-varphiL-carphiQ}) is
  the conditional density of $L$ applied to $\ell$ given $Y=y$ under $\PP_{\vartheta_0,\tau_0}$,
  and that the marginal $\PP_{\vartheta_0,\tau_0}$ on $\Yset$ is equivalent to
  the dominating measure $\nu$.
  Hence we can write the integral
  in~(\ref{eq:varphiQ:condS}) as
  $$
  \int_\Yset \check{\varphi}(y)\normev[E]{S(y)}g(y)\;\nu(\rmd
  y)\quad\text{with}\quad g(Y)=\PE_{\vartheta_0,\tau_0}\argcond{\normop{L}}{Y}\;.
  $$
By~\ref{item:condB-sup-exp-moment}
  and~(\ref{eq:cond-checkVarphi-linear-mixture}), we obtain
  Condition~(\ref{eq:varphiQ:condS}). Similarly the integral
  in~(\ref{eq:varphiQ:cond_ell}) reads, for all $\xi\in\changevar(\tilde{\Tset})$,
  $$
  \int_\Yset \check{\varphi}(y)g(y)\;\nu(\rmd
  y)\quad\text{with}\quad g(Y)=\PE_{\vartheta_0,\tau_0}\argcond{\normev[F]{\gradA{L}(\xi)}}{Y}\;.  
  $$
  Clearly, for this $g$, \ref{item:condB-stilde-ratio} gives that,
  for some $\epsilon,C>0$ 
  $g(y)\leq C\epsilon^{-1}\,\rme^{\tilde{\mathrm{s}}_{\vartheta_0,\tau_0}(y)}$
  for $\nu$-a.e. $y$, and~(\ref{eq:cond-checkVarphi-linear-mixture})
  implies~(\ref{eq:varphiQ:cond_ell}).

  We can thus apply \Cref{prop:argrmax-objective-equivalence} and to
  conclude, we only need to check that the conditions of
  Assertions~\ref{item:argrmax-objective-equivalence:1}
  and~\ref{item:argrmax-objective-equivalence:2} of
  \Cref{prop:argrmax-objective-equivalence} also applies, that is, we
  need to show that $\gradA{\mathrm{w}}_{\tau_0}$ is continuous on
  $\changevar(\tilde{\Tset})$ in every direction and
  that~(\ref{eq:cond-strict-ineq-varphi-L}) holds for all
  $\xi\neq\xi'\in\changevar(\tilde{\Tset})$. The continuity of
  $\gradA{\mathrm{w}}_{\tau_0}$ follows from
  \Cref{lem:gradAw-continuity-fubini} since
  \eqref{eq:cond-checkVarphi-linear-mixture-left} holds by
  \eqref{eq:cond-checkVarphi-linear-mixture} (and we obtain in passing \eqref{eq:cond-gradA-w-cont-y-wise-expr-}, which is exactly \eqref{eq:cond-gradA-w-cont-y-wise-expr}). Finally,
  Condition~(\ref{eq:cond-strict-ineq-varphi-L}) follows
  from~\ref{item:condB-ident}, since in the setting of
  \Cref{thm:argrmax-objective-equivalence}, $\tilde{\varphi}(\ell)>0$
  for all $\ell\in\Zset$.
\end{proof}
\subsubsection{Proof of \Cref{ex-thm:student-mixture}}
  \label{sec:elliptical-student-mixture-proof}

  The proof of \Cref{ex-thm:student-mixture} relies on some preliminary results. We first show in \Cref{prop:elleptical} below that \Cref{thm:argrmax-objective-equivalence} applies to a general mixture of Gaussian distributions (which notably includes the Student's $t$ distribution family considered in \Cref{ex:student-mixture}).
  \begin{prop}\label{prop:elleptical}
    Let $\check{\parammixproba}$ be a class of probability distributions
    on $(0,\infty)$ and suppose that for all  
  $\check{\tau}_0\in\check{\parammixproba}$, there exists      $\epsilon_0>0$ such that
    \begin{equation}
      \label{eq:gaussian-expo-setting-gq-cond}
      \int_0^\infty \rme^{\epsilon_0\,z}\;\check{\tau}_0(\rmd z)<\infty\;.
    \end{equation}
    Let $d\geq1$ and define for any
    $\vartheta=(m,\Sigma)\in\tilde{\Tset}=\rset^d\times\mathcal{M}_{>0}(d)$
    \begin{equation}
      \label{eq:eliptical-family}
      \check{k}(\vartheta,y) = \int_0^\infty \mathcal{N}\lr{y;m,z^{-1}\Sigma} \check{\tau}_0(\rmd z)\;,
    \end{equation}
    seen as a kernel density with respect to $\nu$ being the Lebesgue
    measure on $\rset^d$. Then the following assertions hold.
  \begin{enumerate}[label=(\roman*)]
    \item\label{item:prop:elleptical:assertion1} Let $k^{(1)}$ be  ELM
      family given by
    \Cref{defi:LMEFfam} with $F=E=\rset^d\times\rset^{d\times d}$,
    $E_0=\tilde{\Tset}=\rset^d\times\mathcal{M}_{>0}(d)$,
    $\changevar(m,\Sigma)=(\Sigma^{-1}m,\Sigma^{-1})$,
\begin{equation}
  \label{eq:gaussian-expo-setting--q-def}
k^{(o)}(\xi,y) 
=\mathcal{N}\lr{y;m,\Sigma}
\qquad\text{for all $\xi=\changevar(\vartheta)$ with $\vartheta=(m,\Sigma)$}\;,  
\end{equation}
and $\parammixproba$ the class of all push-forward probabilities $\tau_0$ of $\check{\tau}_0\in\check{\parammixproba}$ through
the mapping $z\mapsto z\mathrm{Id}$ defined from $(0,\infty)$ to
$\mathcal{L}(E)$, where $\mathrm{Id}$ is the identity operator on $E$.
Then, $\parammixproba$ is a class of distributions on positive scalar
operators on $E$ and,
for all
    $\vartheta=(m,\Sigma)\in\tilde{\Tset}$, we have
    $\check{k}(\vartheta,y)=k^{(1)}((\vartheta,\tau_0),y)$. 
  \item\label{item:prop:elleptical:assertion2} Setting, for all $y\in\rset^d$ and
    $\vartheta_0=(m_0,\Sigma_0) \in \tilde{\Tset}$,
    \begin{align}
      \label{eq:gaussian-expo-setting-quad-form-def}
      &q_{\vartheta_0}(y)=\frac12(y-m_0)^T\Sigma_0^{-1}(y-m_0)\;,\\
      \label{eq:gaussian-expo-setting-gq-def}
&    g_{\check{\tau}_0} (u,v)= \int_0^\infty z^u\rme^{-v\,z}\;\check{\tau}_0(\rmd z)\quad \text{for all  $u>-1$ and
    $v>-\epsilon_0$}\;,
    \end{align}
    this ELM family satisfies
    \ref{asump:Upsilon-image}--\ref{item:condB-ident}  with
    \begin{align}
      \label{eq:gaussian-expo-setting-m-cond-def}
      &
        \tilde{\mathrm{m}}_{\vartheta_0,\tau_0}(y)=\frac{g_{\check{\tau}_0}
        (d/2+1,q_{\vartheta_0}(y))}{g_{\check{\tau}_0}
        (d/2,q_{\vartheta_0}(y))}\;,\\
      \label{eq:gaussian-expo-setting-s-cond-def}
      &    \tilde{\mathrm{s}}_{\vartheta_0,\tau_0}(y)=\log\lr{\frac{g_{\check{\tau}_0} (d/2,q_{\vartheta_0}(y)-\epsilon_0)}{g_{\check{\tau}_0} (d/2,q_{\vartheta_0}(y))}}\;.
    \end{align}
  \item\label{item:prop:elleptical:assertion3} \Cref{thm:argrmax-objective-equivalence} applies for any
    $\theta_0 = (\vartheta_0, \tau_0) \in \tilde{\Tset} \times
    \parammixproba$ and any probability density function
    $\check{\varphi}:\rset^d\to\rset_+$
    satisfying
  \begin{equation}
    \label{eq:cond-checkVarphi-linear-mixture-elliptical}
    \int_\Yset \check{\varphi}(y)\,\lr{\lr{\|y\|+\|y\|^2}\,\tilde{\mathrm{m}}_{\vartheta_0,\tau_0}(y)
      +\rme^{\tilde{\mathrm{s}}_{\vartheta_0,\tau_0}(y)}}\;\nu(\rmd y) <\infty
    \;, 
  \end{equation}
in which case the
    argmax
    problem~(\ref{eq:argrmax-objective:armax-prob-check-varphi-Q })
    has a unique solution $\xi^*=\changevar(\vartheta^*)$ with
    $\vartheta^*=(m^*,\Sigma^*)$ given by
  \begin{align}\label{eq:update-mean-elliptical}
    m^* & =
          \frac{\PE_{\vartheta_0,\tilde{\tau}_0}\lrb{Z\,\tilde{Y}}}{\PE_{\vartheta_0,\tilde{\tau}_0}\lrb{Z}}\;,\\
    \label{eq:update-cov-elliptical}
    \Sigma^* &= \PE_{\vartheta_0,\tilde{\tau}_0}\lrb{Z\,\tilde{Y}\tilde{Y}^T}-\frac{\PE_{\vartheta_0,\tilde{\tau}_0}\lrb{Z\,\tilde{Y}}\PE_{\vartheta_0,\tilde{\tau}_0}\lrb{Z\,\tilde{Y}}^T}{\frtxt{\PE_{\vartheta_0,\tilde{\tau}_0}\lrb{Z}}}\;,
  \end{align}
  where, under $\PP_{\vartheta_0,\tilde{\tau}_0}$, $Z$ has density
  with respect to $\check{\tau}_0$ given by
  \begin{align}\label{eq:Zdens-elliptical}
  \tilde{\varphi}(z) & =\int_{\Yset} \check{\varphi}(y)\frac{z^{d/2}\rme^{-z\,q_{\vartheta_0}(y)}}{g_{\check{\tau}_0} (d/2,q_{\vartheta_0}(y))}\,
      \rmd y
  \end{align}
 and the conditional distribution of $\tilde{Y}$ given $Z=z$ has density
  \begin{align}\label{eq:tildeY-cond-dens-elliptical}
  \psi\argcond yz=\frac1{1+b}\,\frac{z^{d/2}\rme^{-z\,q_{\vartheta_0}(y)}\,\check{\varphi}(y)}
  {g_{\check{\tau}_0} (d/2,q_{\vartheta_0}(y))\,\tilde{\varphi}(z)}
  +\frac b{1+b}\,\mathcal{N}\lr{y;m_0,z^{-1}\Sigma_0}\;.
  \end{align}
\end{enumerate}
\end{prop}
\begin{proof}
  We prove the claimed assertions successively.
  \begin{enumerate}[label=\textbf{Proof of
      Assertion~(\roman*)},wide=0pt, labelindent=0pt]
  \item The spaces $F=E=\rset^d\times\rset^{d\times d}$ are endowed with the
  usual inner product
$$
\pscal[E]{(x,M)}{(x',M')}=x^Tx'+\mathrm{Tr}\lr{M^TM'}\;.
$$
We have $\intE_0=E_0=\tilde{\Tset}=\rset^d\times\mathcal{M}_{>0}(d)$,
and set, for all $y\in\rset^d$, $S(y)=\lr{y,-yy^T/2}$ and, for all
$(x,M)\in\intE_0$, $A(x,M)=\frac12\lr{x^TM^{-1}x-\log\lrav{M}}$. In
this setting, $\ell^T=\ell$, $\ell=\normop{\ell}\,\mathrm{Id}$ for
$\tau_0$-almost all $\ell\in \Zset$, and $\normop{L}$ has distribution
$\check{\tau}_0$ for $L\sim \tau_0$.  Then, for all
$\vartheta=(m,\Sigma)\in\tilde{\Tset}$ and $\tau_0$-almost all $\ell\in \Zset$, we have
$\ell\circ\changevar(m,\Sigma)=\changevar(m,\Sigma/\normop{\ell})$. Consequently,
for all $\vartheta \in \tilde{\Tset}$, $y \in \rset^d$ and
$\tau_0$-almost all $\ell\in\Zset$, we have
 \begin{align}
   \nonumber
k^{(o)}( \ell \circ \changevar(\vartheta), y) & = \lr{ \frac{1}{2\pi} \lrav{\normop{\ell}^{-1}\Sigma}}^{1/2}\rme^{-(y-m)^T(\normop{\ell}^{-1} \Sigma)^{-1}(y-m)/2}
 \\
  \label{eq:gaussian-expo-setting--q-def-bis}
                                              & =\mathcal{N}\lr{y;m,\normop{\ell}^{-1}\Sigma}\;.
\end{align}
Thus the obtained ELM
family $k^{(1)}$ of \Cref{defi:LMEFfam} satisfies
$\check{k}(\vartheta,y)=k^{(1)}((\vartheta,\tau_0),y)$ for all
$\vartheta=(m,\Sigma)\in\tilde{\Tset}$.
\item Assumption \label{item:1}~\ref{asump:Upsilon-image} is obvious
  in the above setting. We then check \ref{item:condB-ident},~\ref{item:condB-sup-exp-moment}
    and~\ref{item:condB-stilde-ratio}, successively. For all $\xi\neq\xi' \in \changevar(\tilde{\Tset})$ and $\tau_0\in\parammixproba$, we easily see that
    $$
    \tau_0 (\set{\ell\in\Zset}{\ell(\xi-\xi')=0})=\check{\tau}_0 (\{0\})=0\;,
    $$
    and \ref{item:condB-ident} follows. To
    check~\ref{item:condB-sup-exp-moment}, we compute, for all
    $\vartheta_0=(m_0,\Sigma_0)$ and $\tau_0\in\parammixproba$,
\begin{align*}
    \PE_{\vartheta_0,\tau_0}\argcond{\normop{L}}{Y}& = \int_0^\infty z \;\frac{\mathcal{N}\lr{Y;m_0,z^{-1}\Sigma_0}}{\check{k}(\vartheta_0,Y)}\;\check{\tau}_0(\rmd z)\\
&=    \frac{\int z^{d/2+1}\rme^{-z\,q_{\vartheta_0}(Y)}\;\check{\tau}_0(\rmd z)}
    {\int z^{d/2}\rme^{-z\,q_{\vartheta_0}(Y)}\;\check{\tau}_0(\rmd z)}
\end{align*}
where we successively used~(\ref{eq:DefiKernelLinMixt-PE})
with~(\ref{eq:gaussian-expo-setting--q-def-bis}) and~(\ref{eq:gaussian-expo-setting-quad-form-def}).
Hence \ref{item:condB-sup-exp-moment} follows
from~(\ref{eq:gaussian-expo-setting-gq-def}) and~(\ref{eq:gaussian-expo-setting-m-cond-def}).
Let us now check~\ref{item:condB-stilde-ratio}. Straightforward computations yield, for all
    $\xi=(x,M)\in \changevar(\tilde{\Tset})$ and $\tau_0$-almost all $\ell\in \Zset$, 
    \begin{align*}
      \nabla A\circ \ell(\xi)
        &=(M^{-1}x,-(M^{-1}xx^TM^{-1}+\normop{\ell}^{-1}\,M^{-1})/2)\\
      \gradA{\ell}(\xi)&= (\normop{\ell}\, M^{-1}x,-(\normop{\ell}\,M^{-1}xx^TM^{-1}+ M^{-1})/2)\\
      \normev[F]{\gradA{\ell}(\xi)}&\leq(1+\normop{\ell})\,\normev[F]{(M^{-1}x,M^{-1}xx^TM^{-1}+M^{-1})}\;.
    \end{align*}
    It follows, that, for all $\vartheta_0=(m_0,\Sigma_0) \in \tilde{\Tset}$, $\tau_0\in\parammixproba$, $\xi=(x,M)\in \changevar(\tilde{\Tset})$ and $\epsilon>0$,
    \begin{align*}
    \PE_{\vartheta_0,\tau_0}\argcond{\rme^{\epsilon\,\normev[F]{\gradA{L}(\xi)}}}{Y}&\leq
    C_1\,\PE_{\vartheta_0,\tau_0}\argcond{\rme^{C_2\,\epsilon\,\normop{L}}}{Y}\quad\PP_{\vartheta_0,\tau_0}-\as\\
&=  C_1\,\log\lr{\frac{g_{\check{\tau}_0} (d/2,q_{\vartheta_0}(y)- C_2 \epsilon)}{g_{\check{\tau}_0} (d/2,q_{\vartheta_0}(y))}}\quad\PP_{\vartheta_0,\tau_0}-\as\,,      
    \end{align*}
    where $C_1,C_2>0$ only depend on
    $\xi$. Hence, since we assumed~(\ref{eq:gaussian-expo-setting-gq-cond}), the definition of
    $g_{\check{\tau}_0}$ 
    in~(\ref{eq:gaussian-expo-setting-gq-def}) holds for all $v>-\epsilon_0$, and
    \ref{item:condB-stilde-ratio} holds with
    $\tilde{\mathrm{s}}_{\vartheta_0,\tau_0}$ as
    in~(\ref{eq:gaussian-expo-setting-s-cond-def}) by taking $\epsilon<\epsilon_0/C_2$.
    \item  From the previous assertion, we can apply
  \Cref{thm:argrmax-objective-equivalence} and use
  \eqref{eq:wstar-defined-check-varphiQ-new-sol-bis} to solve the argmax
  problem~(\ref{eq:argrmax-objective:armax-prob-check-varphi-Q }) for a density
  $\check{\varphi}$ with respect to the Lebesgue measure on
  $\rset^d$ satisfying~(\ref{eq:cond-checkVarphi-linear-mixture}),
  which, since $S(y)=\lr{y,-yy^T/2}$ and
  $$
  \normev[E]{S(y)}=\lr{\|y\|^2+\frac14\mathrm{Tr}(yy^Tyy^T)}^{1/2}=O\lr{\|y\|+\|y\|^2}\;,
  $$
  is equivalent to~(\ref{eq:cond-checkVarphi-linear-mixture-elliptical}).
  Let
  us compute both sides of the
  equation~(\ref{eq:wstar-defined-check-varphiQ-new-sol-bis}) which
  characterizes the solution $\xi^*=\changevar(\vartheta^*)$.
  The left-hand side of~(\ref{eq:wstar-defined-check-varphiQ-new-sol-bis}) reads
  \begin{align*}
  \PE_{\vartheta^*,\tilde{\tau}_0}\lrb{L^T\circ S(Y)}&=\PE_{\vartheta^*,\tilde{\tau}_0}\lrb{Z\, S(Y)}\;,
  \end{align*}
  where $Z=\normop{L}$ has density with respect
  to $\check{\tau}_0$ given by~(\ref{eq:Zdens-elliptical}).
  The right-hand side of~(\ref{eq:wstar-defined-check-varphiQ-new-sol-bis}) reads
  \begin{align*}
   \PE_{\vartheta_0,\tilde{\tau}_0}\lrb{L^T\circ S(\tilde{Y})}& = \PE_{\vartheta_0,\tilde{\tau}_0}\lrb{Z\,S(\tilde{Y})} \;,
  \end{align*}
  where $Z$ is the same as above and the conditional distribution of
  $\tilde{Y}$ given $Z=z$ has density given by~(\ref{eq:tildeY-cond-dens-elliptical}).
  Finally, the solution $\xi^*=\changevar(\vartheta^*)$ is given by the
  equation
  \begin{equation}
    \label{eq:gaussian-expo-setting-final-update-eq}
  \PE_{\vartheta^*,\tilde{\tau}_0}\lrb{Z\, S(Y)}
  = \PE_{\vartheta_0,\tilde{\tau}_0}\lrb{Z\,S(\tilde{Y})}
  \end{equation}
  and this equation leads to the
  solutions~(\ref{eq:update-mean-elliptical}) and~(\ref{eq:update-cov-elliptical}).
  Indeed, note that 
  \begin{align}
    \PE_{\vartheta^*,\tilde\tau_0}[Y|Z]=m^*\, \quad \mbox{and}\quad \PE_{\vartheta^*,\tilde\tau_0}[YY^T|Z]-m^* (m^*)^T=\Sigma^*/Z \label{eq:cond:elliptique}
  \end{align}
  Applying the first identity in \eqref{eq:cond:elliptique}, then the tower property and the identity \eqref{eq:gaussian-expo-setting-final-update-eq} applied to the first component of $S(Y)=(Y,-YY^T/2)$,  
  \begin{align*}
    m^*=\frac{\PE_{\vartheta^*,\tilde\tau_0}[Zm^*]}{\PE_{\vartheta^*,\tilde\tau_0}[Z]}=\frac{\PE_{\vartheta^*,\tilde\tau_0}[ZY]}{\PE_{\vartheta^*,\tilde\tau_0}[Z]}=\frac{\PE_{\vartheta_0,\tilde{\tau}_0}\lrb{Z\tilde{Y}}}{\PE_{\vartheta^*,\tilde\tau_0}[Z]}=\frac{\PE_{\vartheta_0,\tilde{\tau}_0}\lrb{Z\tilde{Y}}}{\PE_{\vartheta_0,\tilde\tau_0}[Z]}
  \end{align*}
  where the last equality follows from the fact that under
  $\PP_{\vartheta^*,\tilde\tau_0}$ or
  $\PP_{\vartheta_0,\tilde\tau_0}$, $Z$ has the density $\tilde
  \varphi$ with respect to $\check
  \tau_0$.  Similarly, using the second equality in
  \eqref{eq:cond:elliptique}, then the tower property and the identity
  \eqref{eq:gaussian-expo-setting-final-update-eq} applied to the
  second component of $S(Y)=(Y,-YY^T/2)$
  \begin{align*}
    \Sigma^*&=\PE_{\vartheta^*,\tilde\tau_0}[Z (\Sigma^*/Z)]=\PE_{\vartheta^*,\tilde\tau_0}[Z YY^T]-\PE_{\vartheta^*,\tilde\tau_0}[Z] m^* (m^*)^T\\
    &=\PE_{\vartheta_0,\tilde{\tau}_0}\lrb{Z\tilde{Y}\tilde{Y}^T}-\PE_{\vartheta^*,\tilde\tau_0}[Z] m^* (m^*)^T
  \end{align*}
This proves \eqref{eq:update-mean-elliptical} and \eqref{eq:update-cov-elliptical}. 
\end{enumerate}    
\end{proof}  
\begin{rem}\label{rem:elliptical-student}
  The probability density function~(\ref{eq:eliptical-family}) is
  called an elliptical distribution. In particular if
  $\check{\tau}_0$ is the $\chi^2$ distribution with
  $\nut_0$ degrees of freedom, then it is the density of the Student's
  $t$ distribution with parameter
  $(m,\Sigma,\nut_0)$.  In this case we have
  $$
  g_{\check{\tau}_0} (u,v)
  =\frac{(\nut_0/2)^{\nut_0/2}\Gamma(\nut_0/2+u)}{(\nut_0/2+v)^{\nut_0/2+u}\,\Gamma(\nut_0/2)}\;,  
  $$
  and
  $\tilde{\mathrm{m}}_{\vartheta_0,\tau_0}$ and
  $\tilde{\mathrm{s}}_{\vartheta_0,\tau_0}$
  in~(\ref{eq:gaussian-expo-setting-m-cond-def})
  and~(\ref{eq:gaussian-expo-setting-s-cond-def}) read, for any $\epsilon_0\in(0,\nut_0/2]$,
  \begin{align*}
  &  \tilde{\mathrm{m}}_{\vartheta_0,\tau_0}(y)=\frac{\nut_0/2+d/2}{\nut_0/2+q_{\vartheta_0}(y)}=O(1)\;,\\
  &  \tilde{\mathrm{s}}_{\vartheta_0,\tau_0}(y)=(\nut_0/2+d/2)\log\lr{\frac{\nut_0/2+q_{\vartheta_0}(y)}{\nut_0/2-\epsilon_0+q_{\vartheta_0}(y)}}=O(1)\;,
  \end{align*}
  so that Condition~(\ref{eq:cond-checkVarphi-linear-mixture-elliptical}) boils
  down to $\int_\Yset \check{\varphi}(y)\,\|y\|^2\; \nu(\rmd y)<\infty$.  
\end{rem}

The following lemma is used in \Cref{ex:student-mixture}.
\begin{lem}\label{lem:kappa-def-invertible}
  Define, for all $x\in(0,\infty)$,
  $$\kappa(x)=\log(x)+\Gamma'(x)/\Gamma(x)\;.$$
  Then
  $\kappa$ is increasing and bijective from $(0,\infty)$ to
  $\rset$. 
\end{lem}
\begin{proof}
  This result follows from the fact that the digamma function
  $x\mapsto\Gamma'(x)/\Gamma(x)$ is increasing   from $(0,\infty)$ to
  $\rset$. 
\end{proof}

\begin{proof}[Proof of \Cref{ex-thm:student-mixture}]
We apply \Cref{thm:EM:MixtureModel} in the setting of~\Cref{ex:student-mixture}. Condition~(\ref{eq:posMixtureW})
holds by \Cref{thm:WeightsMixture} and the
update~\ref{item:update_weights_student} in \Cref{ex:student-mixture}.

The remainder of the proof consists in proving
Condition~(\ref{eq:posMixtureP}) for given $n\geq1$ and
$j=1,\dots,J$. We apply
\Cref{prop:elleptical}, whose Assertion~\ref{item:prop:elleptical:assertion1} with
\Cref{rem:elliptical-student} provides an ELM family $k^{(1)}$
such that, for any
$\vartheta=(m,\Sigma)\in\rset^d\times\mathcal{M}_{>0}(d)$ and $\nut>0$, $\check{k}(\vartheta,y)$
in~(\ref{eq:eliptical-family}) with
$\check{\tau}_0=\check{\tau}_{\nut}$ given
by~(\ref{eq:update-student-tau-def}) is
$k^{(1)}((\vartheta,\tau_0),y)$ but also 
$k(\theta,y)$ in \Cref{ex:student-mixture} for $\theta=(m,\Sigma,a)$. 
We will use
\Cref{prop:argmax-linear-mixture-exponential} to this ELM family with $\check{\varphi}=\normratiot$, 
$\theta=(\vartheta,\tau)$ defined with
$\vartheta=(m_{j,n+1},\Sigma_{j,n+1})$ and $\tau$ the measure on
obtained by pushing forward $\check{\tau}_{\nut_{j,n+1}}$ through the
mapping $z\mapsto z\mathrm{Id}$ and
$\theta_0=(\vartheta_0,\tau_0)$ defined similarly but using the parameters
$(m_{j,n},\Sigma_{j,n},\nut_{j,n})$. Then
conclusion~(\ref{eq:generic-decrease-theta-cond}) of
\Cref{prop:argmax-linear-mixture-exponential} is exactly
Condition~(\ref{eq:posMixtureP}) with $k$ as in
\Cref{ex:student-mixture}. Hence, to conclude the proof we only need
to show that \Cref{prop:argmax-linear-mixture-exponential} applies in
this setting. Namely, we need to
prove~(\ref{eq:update-cond-tau-general_mixture})
and~(\ref{eq:update-cond-theta-general_mixture}), which respectively
read
\begin{align}
  \label{eq:update-cond-tau-general_mixture-student-}
  &     \int_{\Yset} \lr{\int_0^{\infty}
               \varphi_{j,n}(z, y)\;
                \log
               \lr{\frac{\rmd\check{\tau}_{\nut_{j,n+1}}}{\rmd\check{\tau}_{\nut_{j,n}}}(z)}\;\check{\tau}_{\nut_{j,n}} (\rmd z)} \rmd y \geq0\;,\\
  \label{eq:update-cond-theta-general_mixture-student}
 &    \int_{\Yset} \lr{\int_{\Zset}
               \varphi_{j,n}(\normop{\ell}, y) 
                \log
               \lr{\frac{k^{(o)}(\ell \circ \changevar(\vartheta),y)}{k^{(o)}(\ell \circ \changevar(\vartheta_0),y)}}
               \;\tau_{\nut_{j,n}} (\rmd \ell)}\rmd y \geq0\;,
\end{align}
where $\varphi_{j,n}$ is defined
by~(\ref{eq:update-student-varphi-def}) and $k^{(o)}$ is the Gaussian
canonical kernel defined by~(\ref{eq:gaussian-expo-setting--q-def}).
We thus conclude with the proof of these two conditions in the
opposite order.

\noindent\textbf{Proof
  of~(\ref{eq:update-cond-theta-general_mixture-student})}
Assertion~\ref{item:prop:elleptical:assertion2}
of \Cref{prop:elleptical} with \Cref{rem:elliptical-student} gives that
the ELM family $k^{(1)}$ satisfies
\ref{asump:Upsilon-image}--\ref{item:condB-ident} with
$\tilde{\mathrm{m}}_{\vartheta_0,\tau_0}$ and
$\tilde{\mathrm{s}}_{\vartheta_0,\tau_0}$ such that
Condition~(\ref{eq:cond-checkVarphi-linear-mixture-elliptical}) boils
down to $\int \check{\varphi}(y)\,\|y\|^2\; \rmd\nu<\infty$. This latter
condition is
satisfied with $\check{\varphi}=\normratiot$ as a consequence
of~(\ref{eq:cond-student-on-p}) by using
\Cref{lem:cond-p-varphi}. Therefore
Condition~(\ref{eq:cond-checkVarphi-linear-mixture-elliptical}) holds and
Assertion~\ref{item:prop:elleptical:assertion3}
of \Cref{prop:elleptical} provides the solution of the argmax
problem~(\ref{eq:argrmax-objective:armax-prob-check-varphi-Q }). A
careful comparison of~(\ref{eq:update-mean-elliptical})
and~(\ref{eq:update-cov-elliptical})
with~(\ref{eq:update-student-m-param})
and~(\ref{eq:update-student-Sigma-param}) tells us that 
$\vartheta=(m_{j,n+1},\Sigma_{j,n+1})$ is the unique argmax
of~(\ref{eq:argrmax-objective:armax-prob-check-varphi-Q
}). Applying~\Cref{cor:monot-decr-general-LM} in this setting finally
gives us~(\ref{eq:update-cond-theta-general_mixture-student}).

\noindent\textbf{Proof
  of~(\ref{eq:update-cond-tau-general_mixture-student-})} Since the left-hand side of~(\ref{eq:update-cond-tau-general_mixture-student-}) is zero for $\nut=\nut_{j,n}$, it suffices to show that $\nut_{j,n+1}$
in~(\ref{eq:update-student-nut-param}) satisfies
\begin{align}
  \label{eq:update-cond-tau-general_mixture-student}
  \nut_{j,n+1}=\argmax_{\nut\in(0,\infty)}     \int_{\Yset} \lr{\int_0^{\infty}
               \varphi_{j,n}(z, y)\;
                \log
               \lr{\frac{\rmd\check{\tau}_{\nut}}{\rmd\check{\tau}_{\nut_{j,n}}}(z)}\;\check{\tau}_{\nut_{j,n}}
    (\rmd z)} \rmd y \;.
\end{align}
By~(\ref{eq:update-student-tau-def}), we have, for all $z>0$,
$$
\log
\lr{\frac{\rmd\check{\tau}_{\nut}}{\rmd\check{\tau}_{\nut_{j,n}}}(z)}
= \frac\nut2\log\lr{\frac \nut2}+\log \Gamma\lr{\frac
\nut2}+\lr{\frac\nut2-\frac{\nut_{j,n}}2}\,\lr{\log z -z}+ C_{j,n}\;,
$$
where $C_{j,n}$ is a positive constant not depending on $\nut$.
The derivative of $x\mapsto x \log\lr{x}+\log \Gamma\lr{x
}$ is $\kappa$ as defined in Step~\ref{item:update_param_student}
of \Cref{ex:student-mixture} and in
\Cref{lem:kappa-def-invertible}. Is is easy to check that
$\int \varphi_{j,n}(z, y)\check{\tau}_{\nut_{j,n}}
    (\rmd z)\rmd y =\int\check{\varphi}(y)\rmd
  y=1$. Thus~(\ref{eq:update-cond-tau-general_mixture-student})
  follows from~(\ref{eq:update-student-nut-param}) with
  \Cref{lem:kappa-def-invertible} if we can show that
  $$
  \int_{\Yset}\lr{\int_0^{\infty} \varphi_{j,n}(z, y)\,\lr{\lrav{\log z}+z}\;\check{\tau}_{\nut_{j,n}}
    (\rmd z)}\rmd y<\infty \;.
  $$
  In fact, we will show that there exists $\epsilon>0$ such that for
  all $t\in(-\epsilon,1]$,
  \begin{equation}
    \label{eq:last-moment-student-easy}
  \int_{\Yset}\lr{\int_0^{\infty} \varphi_{j,n}(z, y)\,z^t\;\check{\tau}_{\nut_{j,n}}
    (\rmd z)}\rmd y<\infty \;,
  \end{equation}
  which indeed implies the previous display. By definition of
  $\varphi_{j,n}$ in~(\ref{eq:update-student-varphi-def}), we find
  that~(\ref{eq:last-moment-student-easy}) is implied by
  $$
  I(t):=\int_{\Yset}\lr{\int_0^{\infty} \normratiot(y) z^{t+d/2}\,\rme^{-z
    q_{j,n}(y)}\,\lr{1+\frac2{\nut_{j,n}}q_{j,n}(y)}^{(\nut_{j,n}+d)/2}\;\check{\tau}_{\nut_{j,n}}
    (\rmd z)}\rmd y<\infty\;,
  $$
  where we set
  $q_{j,n}(y)=\frac12(y-m_{j,n})^T\Sigma_{j,n}(y-m_{j,n})$.
  Since, as in \Cref{rem:elliptical-student}, for all $t>-d/2$,
  $$
  \int_0^{\infty} z^{t+d/2}\,\rme^{-z
    q_{j,n}(y)}\;\check{\tau}_{\nut_{j,n}}
  (\rmd z) 
  =\frac{(\nut_{j,n}/2)^{\nut_{j,n}/2}\Gamma(\nut_{j,n}/2+t+d/2)}{(\nut_{j,n}/2+q_{j,n}(y))^{\nut_{j,n}/2+t+d/2}\,\Gamma(\nut_{j,n}/2)}\;,
  $$
  we get that for all  $t>-d/2$, there exists $C'_{j,n}(t)>0$ such
  that 
  $$
   I(t) =
   C'_{j,n}(t)\,\int_{\Yset}\normratiot(y)\,\lr{1+\frac2{\nut_{j,n}}q_{j,n}(y)}^{-t}\;\rmd
   y\;.
   $$
   Since $q_{j,n}(y)$ is a quadratic form and we already checked that
   $\int \normratiot(y)\,(1+\|y\|^2)\; \rmd\nu<\infty$ in the proof   of~(\ref{eq:update-cond-theta-general_mixture-student}), we finally find that
   $I(t)<\infty$ if $t>-d/2$ and $t\geq-2$, which concludes the proof of~(\ref{eq:last-moment-student-easy}).
\end{proof}

\section{Additional Numerical Experiments}

\label{sec:addNumExp}

In this section we provide further plots based on the numerical experiments from \Cref{subsec:Exp}.


\begin{figure}[h]
  \centering
  \begin{tabular}{ccc}
    & $J = 10$ & $J = 50$ \\
    \ref{itemExEWGMM} \vspace{-0.5cm} & \\
    & \includegraphics[width=7cm]{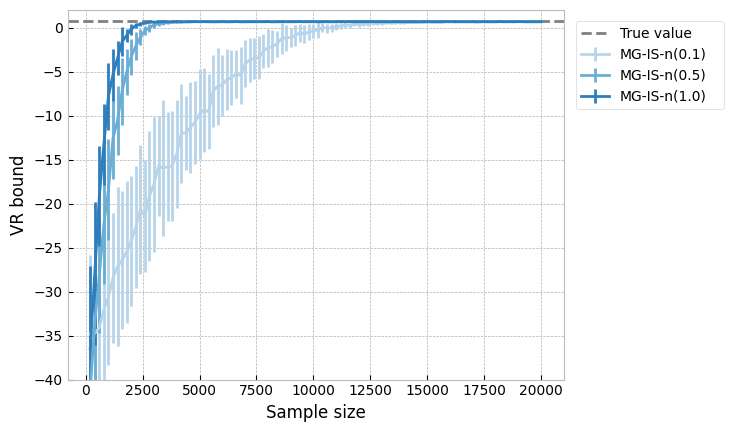} & \includegraphics[width=7cm]{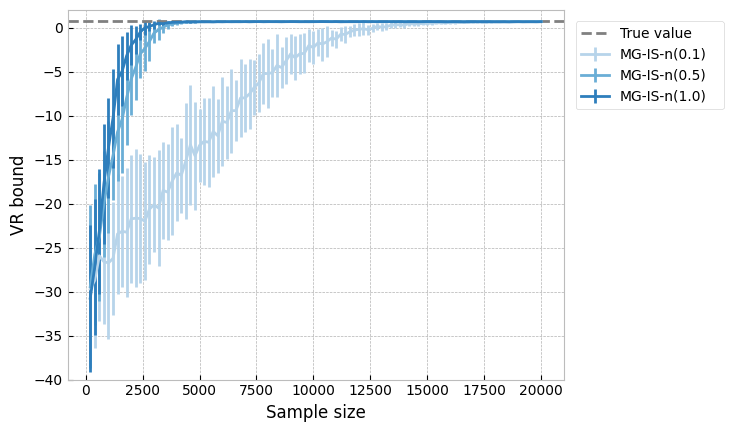} \\
    \ref{itemExUWGMM} \vspace{-0.5cm} & \\ 
    &\includegraphics[width=7cm]{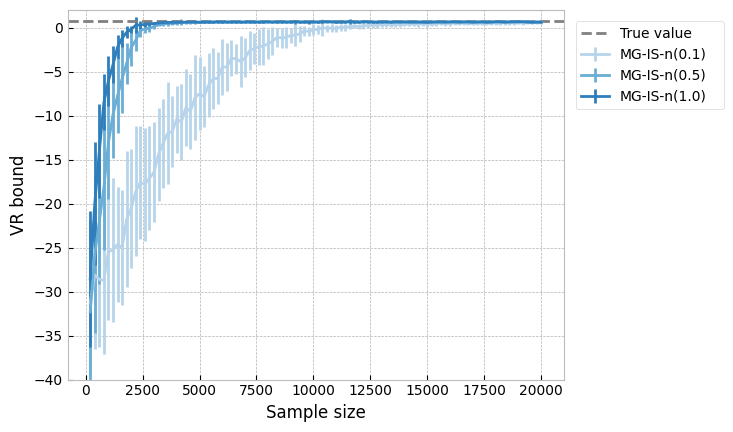} & \includegraphics[width=7cm]{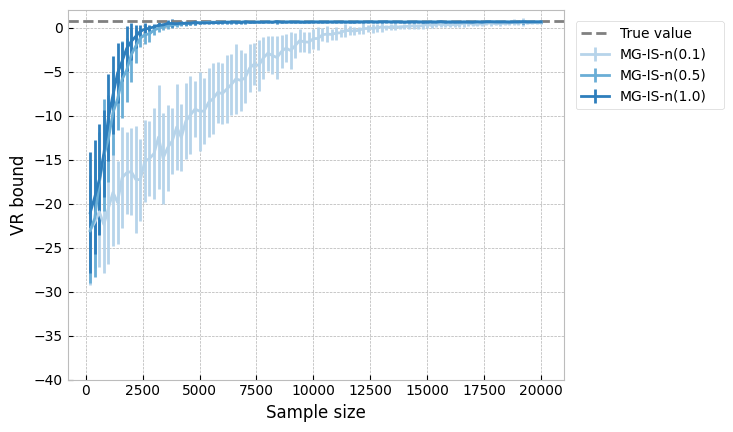} \\
    \ref{itemExEWSMM} \vspace{-0.5cm} & \\
    &\includegraphics[width=7cm]{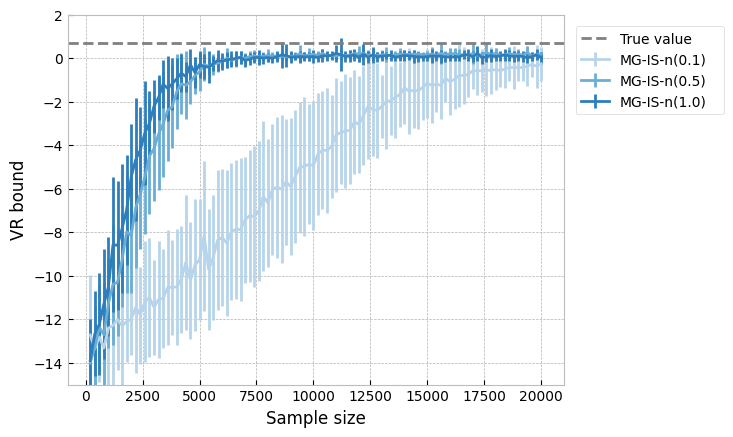} & \includegraphics[width=7cm]{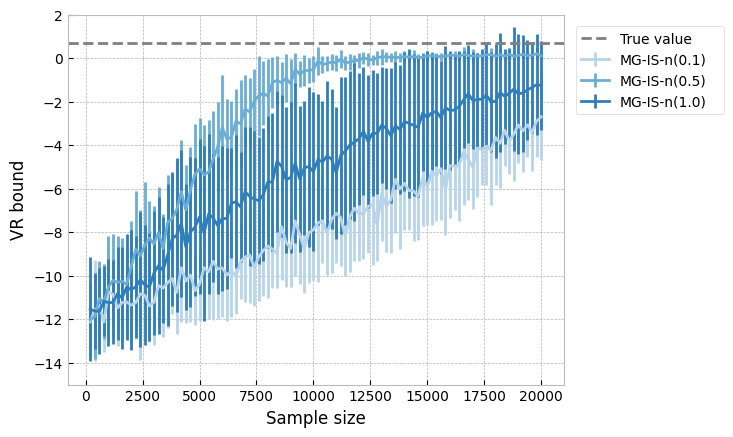} \\
    \end{tabular}
    
  \caption{Error bounds for the Monte Carlo estimate of the VR Bound in the MG-IS-n approach (fixed mixture weights) when considering each of the target distributions \ref{itemExEWGMM}, \ref{itemExUWGMM} and \ref{itemExEWSMM}. } 
  \label{fig:eta0appEBMGISn}
\end{figure}

\begin{figure}[h]
  \centering
  \begin{tabular}{ccc}
    & $J = 10$ & $J = 50$ \\
    \ref{itemExEWGMM} \vspace{-0.5cm} & \\
    & \includegraphics[width=7cm]{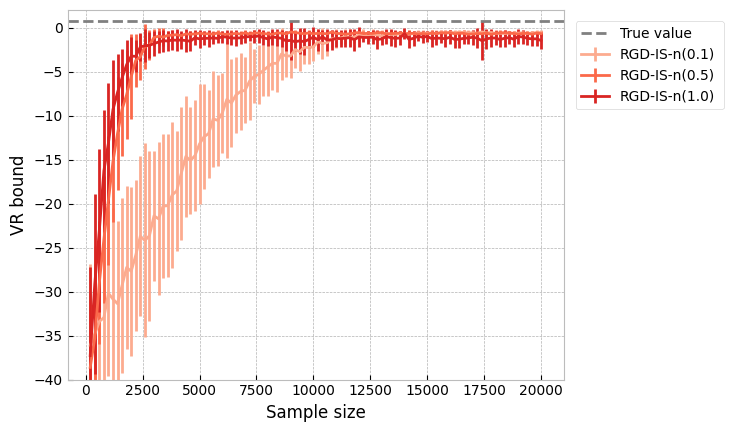} & \includegraphics[width=7cm]{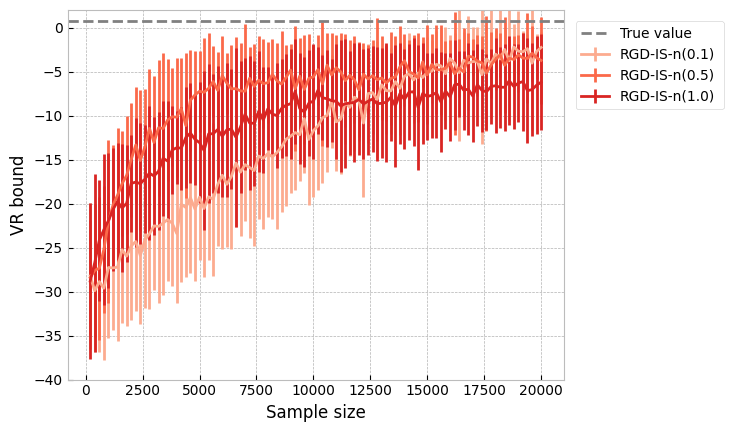} \\
    \ref{itemExUWGMM} \vspace{-0.5cm} & \\ 
    &\includegraphics[width=7cm]{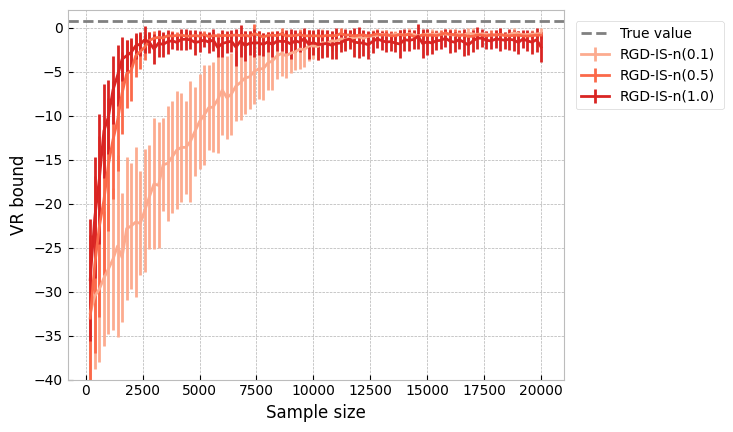} & \includegraphics[width=7cm]{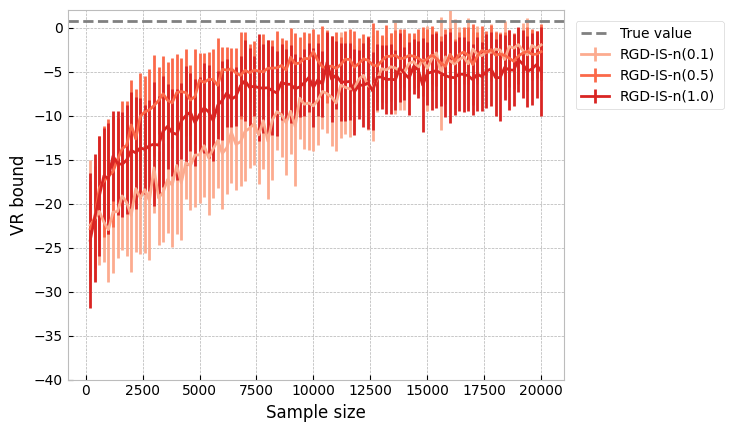} \\
    \ref{itemExEWSMM} \vspace{-0.5cm} & \\
    &\includegraphics[width=7cm]{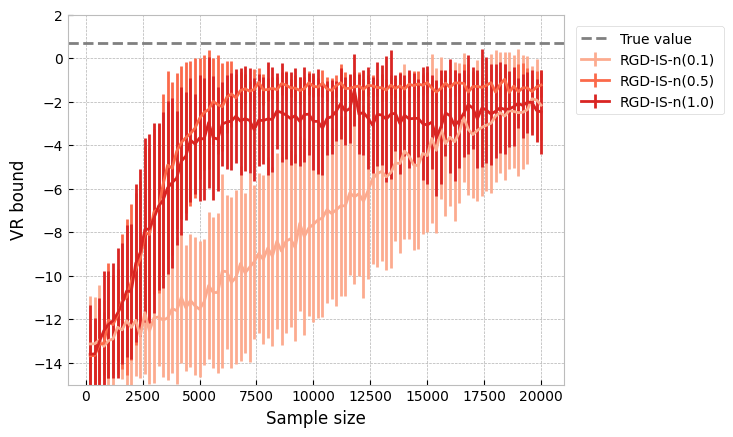} & \includegraphics[width=7cm]{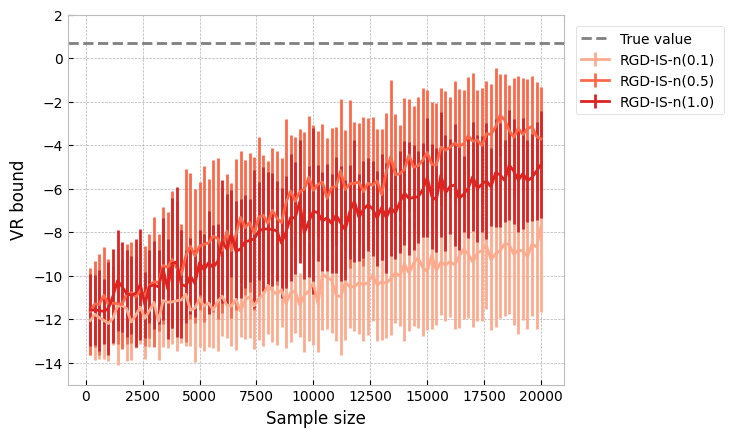} \\
    \end{tabular}
    
  \caption{Error bounds for the Monte Carlo estimate of the VR Bound in the RGD-IS-n approach (fixed mixture weights) when considering each of the target distributions \ref{itemExEWGMM}, \ref{itemExUWGMM} and \ref{itemExEWSMM}.} 
  \label{fig:eta0appEBRGDISn}
\end{figure}





\begin{figure}[h]
  \centering
  \begin{tabular}{ccc}
    & $J = 10$ & $J = 50$ \\
    \ref{itemExEWGMM} \vspace{-0.5cm} & \\
    & \includegraphics[width=7cm]{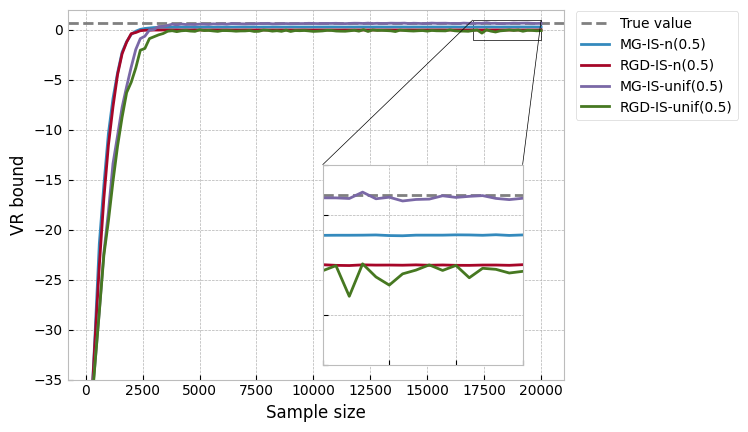} & \includegraphics[width=7cm]{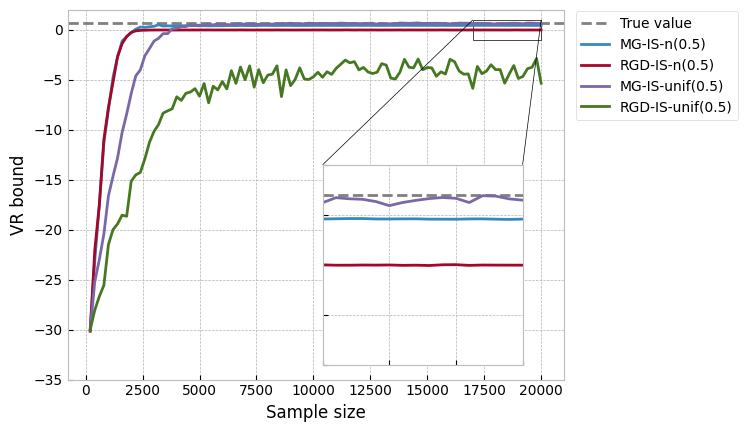} \\
    \ref{itemExUWGMM} \vspace{-0.5cm} & \\
    &\includegraphics[width=7cm]{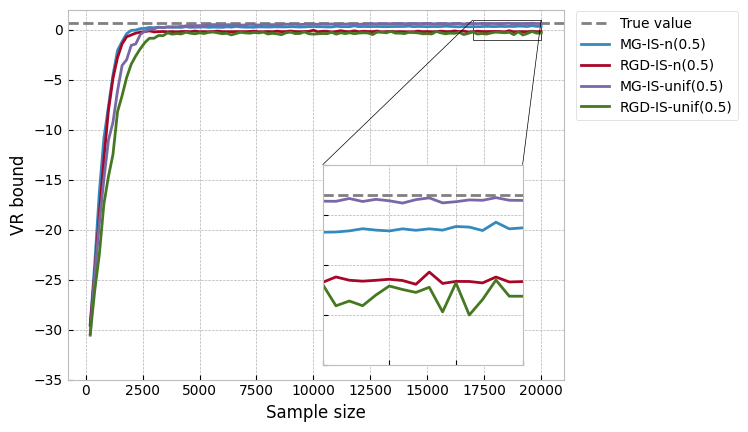} & \includegraphics[width=7cm]{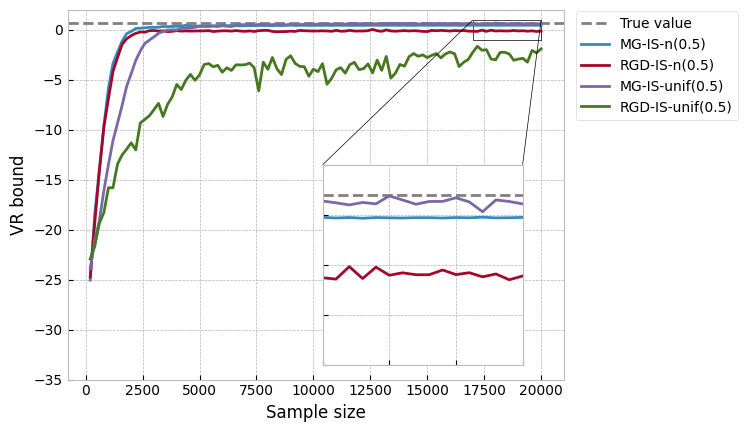} \\
    \ref{itemExEWSMM} \vspace{-0.5cm} & \\
    &\includegraphics[width=7cm]{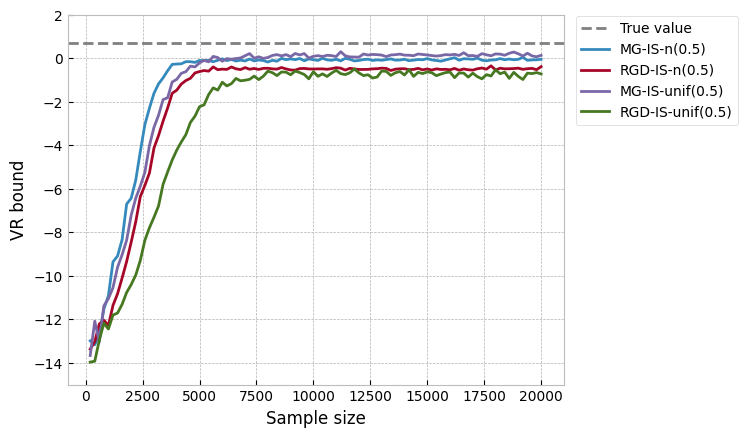} & \includegraphics[width=7cm]{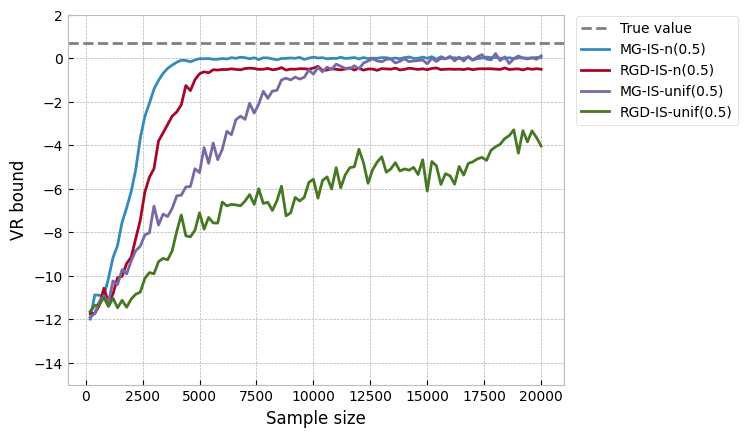} \\    
    \end{tabular}
  
    \caption{Monte Carlo estimate of the VR Bound for the RGD and the MG approaches ($\eta = 0.05$) when considering each of the target distributions \ref{itemExEWGMM}, \ref{itemExUWGMM} and \ref{itemExEWSMM}.}
    \label{fig:eta0dot05} 
  \end{figure}

  \begin{figure}[h]
      \centering
      \begin{tabular}{ccc}
        & $J = 10$ & $J = 50$ \\
        \ref{itemExEWGMM} \vspace{-0.5cm} & \\
    & \includegraphics[width=7cm]{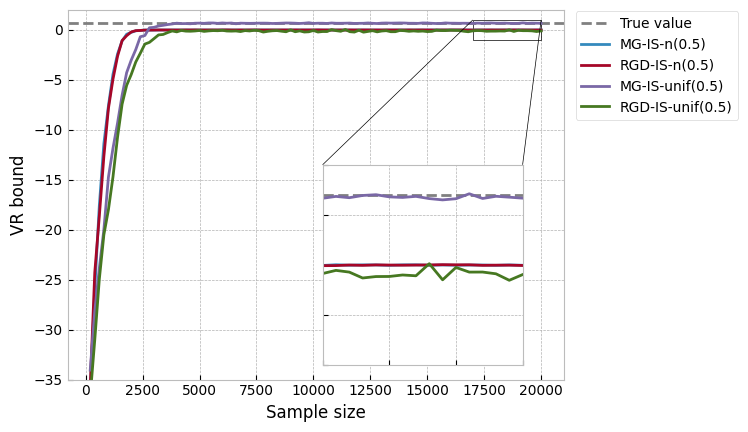} & \includegraphics[width=7cm]{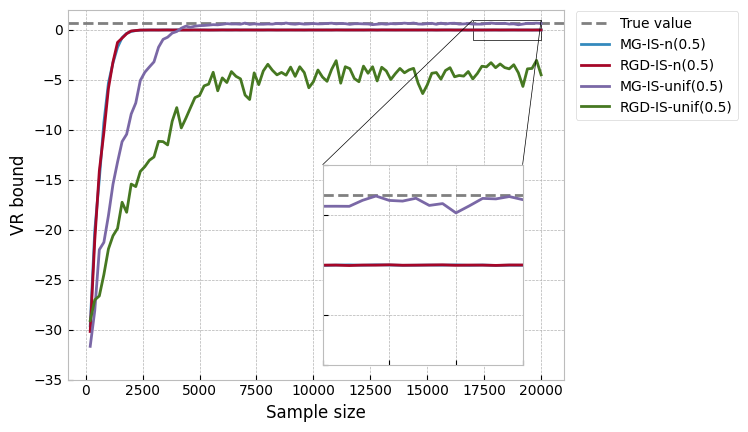} \\
    \ref{itemExUWGMM} \vspace{-0.5cm} & \\
    &\includegraphics[width=7cm]{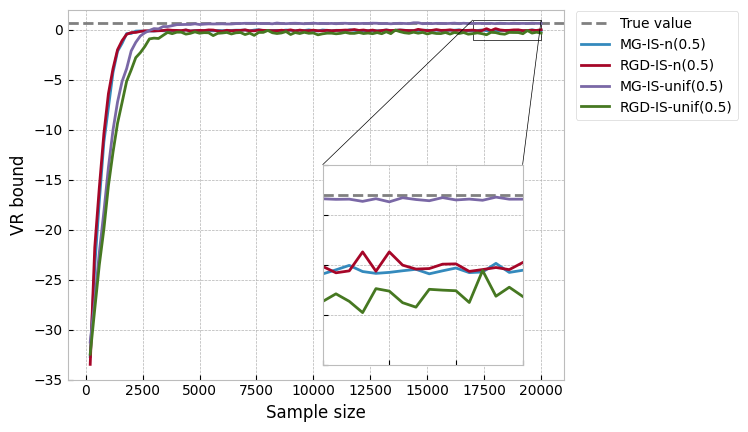} & \includegraphics[width=7cm]{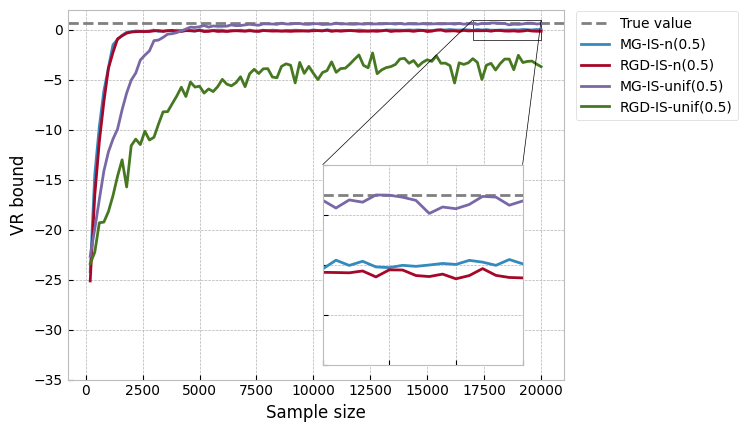} \\
    \ref{itemExEWSMM} \vspace{-0.5cm} & \\
    &\includegraphics[width=7cm]{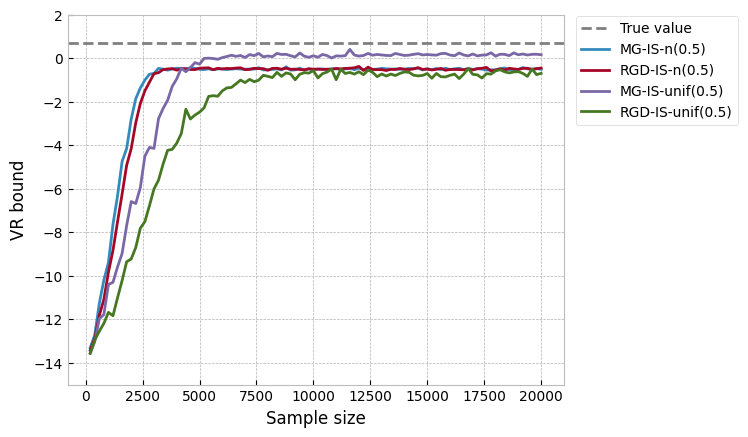} & \includegraphics[width=7cm]{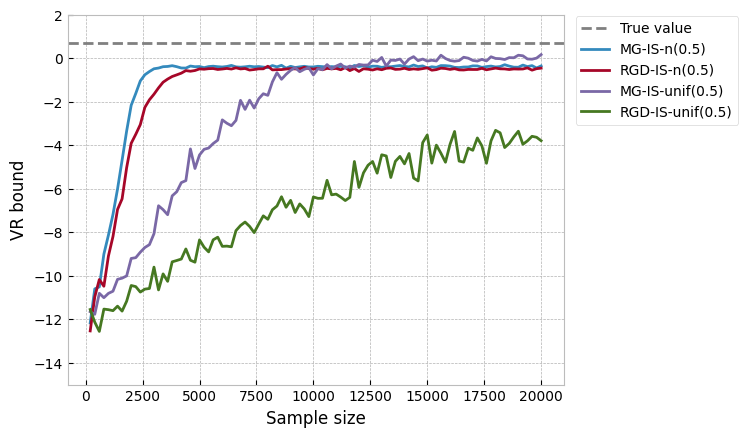} \\    
    \end{tabular}

        \caption{Monte Carlo estimate of the VR Bound for the RGD and the MG approaches ($\eta = 0.5$) when considering each of the target distributions \ref{itemExEWGMM}, \ref{itemExUWGMM} and \ref{itemExEWSMM}.} 
        \label{fig:eta0dot5}
      \end{figure}

      \begin{table}
        \centering
      \begin{tabular}{clccccccc}
            \toprule
          & & & $J = 10$ & & & $J = 50$ & \\  
          & & $\gamma = 0.1$ &  $\gamma = 0.5$ &  $\gamma = 1.0$ & $\gamma = 0.1$ & $\gamma = 0.5$ & $\gamma = 1.0$  \\
             \midrule
          \ref{itemExEWGMM} & RGD-IS-n($\gamma$) & 0.372 & 0.510 & 0.384 & -0.616 & -0.713 & -0.778   \\
          & MG-IS-n($\gamma$) &1.104 & 1.074 & 0.387 & 1.135 & -0.077 & -0.060   \\
          & RGD-IS-unif($\gamma$) & 0.359 & 0.469 & 0.458 & -0.688 & -0.670 & -0.583 \\
          & MG-IS-unif($\gamma$) & \textbf{-0.200} & \textbf{-0.229} & \textbf{-0.515} & \textbf{-1.500} & \textbf{-1.462} & \textbf{-1.246} \\
          \midrule
          \ref{itemExUWGMM} & RGD-IS-n($\gamma$) & -0.025 & -0.056 & -0.087 & -1.027 & -0.997 & -0.969 \\
          & MG-IS-n($\gamma$) & -0.270 & -0.126 & 0.235 & -0.269 & -0.417 & -0.487  \\
          & RGD-IS-unif($\gamma$) & -0.121 & -0.111 & 0.052 & -1.097 & -0.966 & -0.883  \\
          & MG-IS-unif($\gamma$) & \textbf{-1.120} & \textbf{-0.938} & \textbf{-0.957} & \textbf{-1.764} & \textbf{-1.889} & \textbf{-1.192}  \\
          \midrule
          \ref{itemExEWSMM} & RGD-IS-n($\gamma$) & -0.329 & -0.197 & -0.238 & -1.691 & -1.612 & -1.637 \\
          & MG-IS-n($\gamma$) & 1.101 & 0.758 & 0.524 & 0.181 & -0.181 & 0.893 \\
          & RGD-IS-unif($\gamma$) & -0.370 & -0.224 & -0.212 & -1.708 & -1.627 & \textbf{-1.649} \\
          & MG-IS-unif($\gamma$) &\textbf{-1.211} & \textbf{-1.313} & \textbf{-1.083} & \textbf{-2.013} & \textbf{-1.882} & {-0.491}  \\
             \bottomrule
      \end{tabular}
      \caption{LogMSE for the RGD and the MG approaches ($\eta = 0.1$) when considering each of the target distributions \ref{itemExEWGMM}, \ref{itemExUWGMM} and \ref{itemExEWSMM}.}   \label{table:logMSEEtaO1}
      \end{table}

      \begin{table}
        \centering
      \begin{tabular}{clcccccccc}
            \toprule
        &  & &  $J = 10$  & & & $J = 50$ &  \\  
        &  &  $\eta = 0.05$ &  $\eta = 0.1$ & $\eta = 0.5$ &  $\eta = 0.05$ &  $\eta = 0.1$ & $\eta = 0.5$  \\
        \midrule
        \ref{itemExEWGMM} & RGD-IS-n($\gamma$)  & 0.045 & 0.510 & 1.299 & -1.355 & -0.713 & 0.924 \\
          & MG-IS-n($\gamma$)  & 0.087 & 1.074 & 1.343 &  -1.205 & -0.077 & 1.329 \\
          & RGD-IS-unif($\gamma$)  & -0.018 & 0.469 & 1.328 & -1.385 & -0.670 & 0.928 \\
          & MG-IS-unif($\gamma$) & \textbf{-1.244} & \textbf{-0.229} & \textbf{1.100} & \textbf{-2.524} & \textbf{-1.462} & \textbf{0.309} \\
          \midrule
          \ref{itemExUWGMM} & RGD-IS-n($\gamma$) & -0.096 & -0.056 & 0.522 &  -1.509 & -0.997 & 0.542 \\
          & MG-IS-n($\gamma$) &  -0.629 & -0.126 & 0.100 & -1.430 & -0.417 & 0.348 \\
          & RGD-IS-unif($\gamma$)  & -0.195 & -0.111 & 0.489 & -1.542 & -0.966 & 0.529 \\
          & MG-IS-unif($\gamma$)  & \textbf{-1.814} & \textbf{-0.938} & \textbf{-0.149} & \textbf{-1.711} & \textbf{-1.889} & \textbf{-0.282}  \\
          \midrule
          \ref{itemExEWSMM} & RGD-IS-n($\gamma$)  & -0.091 & -0.197 & 0.339 & -1.596 & -1.612 & -0.282 \\
          & MG-IS-n($\gamma$) & -0.772 & 0.758 & 1.358 & -0.878 & -0.181 & 0.927 \\
          & RGD-IS-unif($\gamma$)  & -0.113 & -0.224 & 0.322 & -1.611 & -1.627 & -0.300 \\
          & MG-IS-unif($\gamma$) & \textbf{-1.608} & \textbf{-1.313} & \textbf{-0.253} & \textbf{-1.879} & \textbf{-1.882} & \textbf{-0.716} \\
             \bottomrule
      \end{tabular}
      \caption{LogMSE for the RGD and the MG approaches ($\gamma = 0.5$) when considering each of the target distributions \ref{itemExEWGMM}, \ref{itemExUWGMM} and \ref{itemExEWSMM}.}   \label{table:logMSEVaryingEta} 
      \end{table}
      
\begin{figure}[h]
  \centering
  \begin{tabular}{ccc}
    $\eta$ \vspace{1.5cm} & $J = 50$ & $J = 100$ \\
    $0.1$ \vspace{-1.7cm} & & \\
    & \includegraphics[width=7cm]{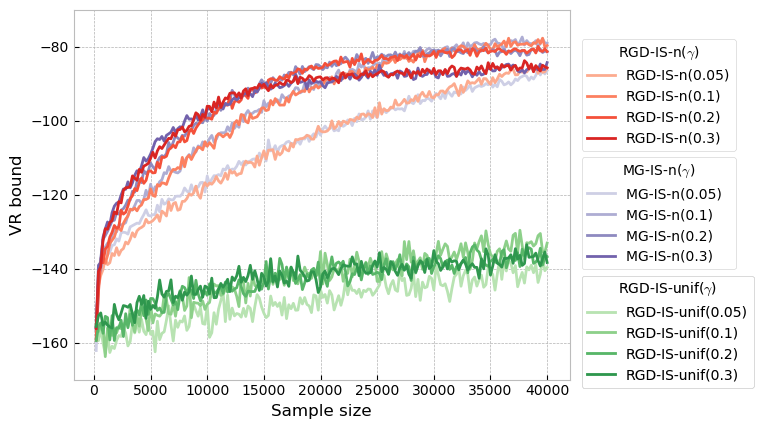} & \includegraphics[width=7cm]{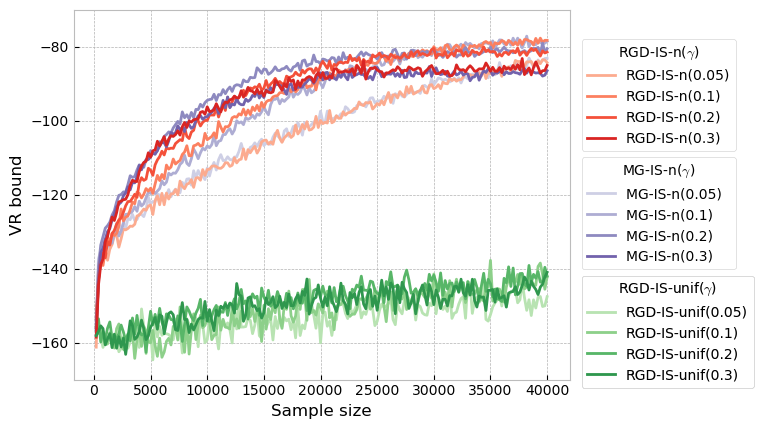} \vspace{1.1cm} \\
    $0.5$ \vspace{-1.7cm} & & \\
    & \includegraphics[width=7cm]{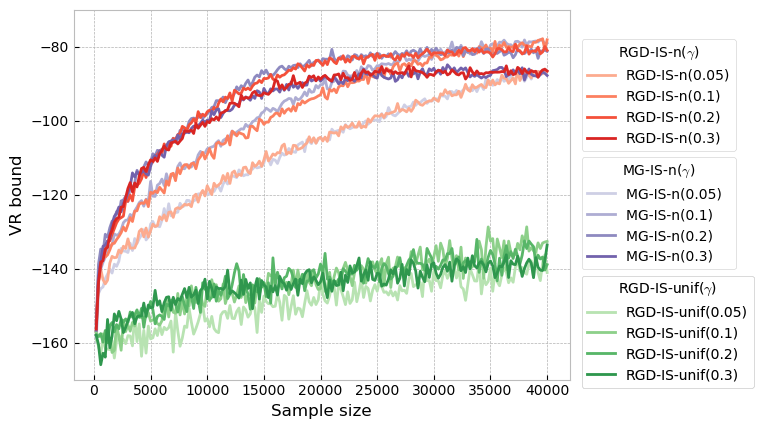} & \includegraphics[width=7cm]{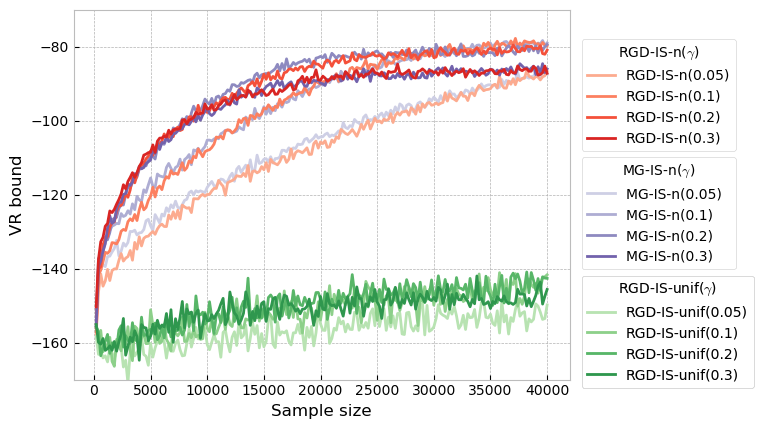} \\
    \end{tabular}

  \caption{Monte Carlo estimate of the VR Bound for the RGD and the MG approaches when considering the Bayesian Logistic Regression on the Covertype data set.} 
  \label{fig:blr:app}
\end{figure}


\bibliography{ddr_2023_final}

\end{document}